\documentclass[11.5pt]{article}
\usepackage{lipsum}

\usepackage[margin=1in]{geometry}
\usepackage{fancyhdr}
\usepackage{graphicx}
\usepackage{mathtools}
\usepackage{amsmath, amssymb, amsfonts, amsbsy, amsthm, latexsym, color}
\usepackage{ulem}
\usepackage{mathtools}
\usepackage{mathrsfs}
\usepackage{cancel}

\usepackage{hyperref} 
\usepackage{empheq}

\usepackage[T1]{fontenc}
\usepackage{authblk}
\usepackage{graphicx}
\usepackage{enumerate}
\usepackage{caption}
\usepackage{subcaption}
\newcommand{\RN}[1]{%
  \textup{\uppercase\expandafter{\romannumeral#1}}%
}
\everymath{\displaystyle}
\usepackage{float}
\renewcommand*{\arraystretch}{1.5}
\usepackage[ruled,vlined]{algorithm2e}
\newtheorem{theorem}{Theorem}[section]
\newtheorem{proposition}[theorem]{Proposition}
\newtheorem{lemma}[theorem]{Lemma}

\theoremstyle{definition}
\newtheorem{definition}{Definition}[section]
\newtheorem{Assumption}{Assumption}[section]
\newtheorem{remark}[theorem]{Remark}

\numberwithin{equation}{section}

\DeclareMathOperator*{\argmin}{arg\,min}
\usepackage{romannum}
\begin{document}                        
\pagenumbering{arabic}

\title{Sigma Delta quantization for images}

\author{He Lyu, Rongrong Wang}
\affil{Michigan State University}

\date{\today}

\maketitle

\begin{abstract}
In signal quantization, it is well-known that introducing adaptivity to quantization schemes can improve their stability and accuracy in quantizing bandlimited signals. However, adaptive quantization has only been designed for one-dimensional signals. The contribution of this paper is two-fold: i). we propose the first family of two-dimensional adaptive quantization schemes that maintain the same mathematical and practical merits as their one-dimensional counterparts, and ii). we show that both the traditional 1-dimensional and the new 2-dimensional quantization schemes can effectively quantize signals with jump discontinuities. These results immediately enable the usage of adaptive quantization on images. Under mild conditions, we show that the adaptivity is able to reduce the reconstruction error of images from the presently best $O(\sqrt P)$ to the much smaller $O(\sqrt s)$, where  $s$ is the number of jump discontinuities in the image and $P$ ($P\gg s$) is the total number of samples.
This $\sqrt{P/s}$-fold error reduction is achieved via applying a total variation norm regularized decoder, whose formulation is inspired by the mathematical super-resolution theory in the field of compressed sensing. Compared to the super-resolution setting, our error reduction is achieved without requiring adjacent spikes/discontinuities to be well-separated, which ensures its broad scope of application.

We numerically demonstrate the efficacy of the new scheme on medical and natural images. We observe that for images with small pixel intensity values, the new method can significantly increase image quality over the state-of-the-art method.
\end{abstract}

\section{Introduction}
\subsection{Quantization}
In signal processing, quantization is the operation in which a signal's real-valued samples get converted into a finite number of bits. As such, quantization digitizes the signal and makes it ready for digital processing. Mathematically, given a signal class $\mathcal{S} \subseteq \mathbb{R}^N $ and a fixed codebook $\mathcal{C}$, the goal of quantization is to map every signal $x$ in $\mathcal{S}$ to a codebook representation $q \in \mathcal{C}$ that can be stored digitally. We use $Q$ to denote the quantization map between the signal space $\mathcal{S}$ and the codebook $\mathcal{C}$
\[
Q: \mathcal{S} \rightarrow \mathcal{C}: x \rightarrow q.
\]
In order to digitally restore the original signals, we usually assign a decoder, also called a reconstruction algorithm, to each quantization scheme. The decoder then recovers each signal from their encoded bits up to a small error. To ensure practicability, we only consider decoders that run in polynomial time. Let  $\Delta$ be a decoder, the error in the decoded signal $\hat{x}=\Delta(q)$ is called distortion 
\[
\textrm{Distortion} \coloneqq\|\hat{x}-x\|_2 \equiv \|\Delta(q)-x\|_2. 
\]
For a given signal class $\mathcal{S}$, we define the optimal quantization $Q$ to be the one that minimizes the distortion subject to a fixed bit budget. More precisely, let the fixed bit budget be $R$, then among all codebooks $\mathcal{C}$ that are representable in $R$ bits and all quantization maps $Q$ from $\mathcal{S}$ to $\mathcal{C}$, the optimal quantization is the one that minimizes the minimax distortion
$$\hat{Q} = \argmin_{Q: \mathcal{S}\rightarrow \mathcal{C}, | \mathcal{C}|=2^R} \min_{\Delta \in \mathcal{D}} \max_{x\in \mathcal{S} }\| \Delta \circ Q(x) - x\|_2, $$
where $\mathcal{D}$ is the set of polynomial-time decoders.

When the signal class $\mathcal{S}$ forms a compact set in a metric space, an optimal quantization can be found through the following information-theoretic argument. Given a fixed error tolerance level $\epsilon$, one can find for the signal class $\mathcal{S}$ an $\epsilon$-net with the minimum possible cardinality $N(\epsilon)$, known as the covering number of $\mathcal{S}$. With this optimal $\epsilon$-net, we carry out the quantization as follows. 
First, we assign the center of each $\epsilon$-ball a symbol in the codebook (different centers are assigned   different symbols).  Then for each point/signal in $\mathcal{S}$, we quantize it to the symbol of the closest center. The total number of symbols used in this quantization equals the number of centers, $N(\epsilon)$, so they can be encoded in $R=\log_2 N(\epsilon)$ bits.  This encoding is optimal as by definition, $N(\epsilon)$ is the minimal cardinality of the net needed to cover the set. For the special case of $\mathcal{S}$ being the unit $\ell_2$-ball in $\mathbb{R}^d$, we have $N(\epsilon)\sim \epsilon^{-d}$ and $R=\log_2 N(\epsilon)\sim d\log_2 \epsilon^{-1}$ or equivalently, $\epsilon \sim 2^{-R/d}$ (see, e.g., \cite{boufounos2007quantization,boufounos2015quantization,SWY18}). This is known as the exponential rate-distortion relation: the error decays exponentially with the number of bits. However, the quantization built this way suffers from the following impracticalities: 1) unless $\mathcal{S}$ has a regular shape,  the computation of the optimal $\epsilon$-net of $\mathcal{S}$ is subject to the curse of dimensionality; and 2) since the nearest center to a signal can only be found after all samples of that signal are received, this scheme cannot be operated in an online manner.  

These issues inspire people to impose the following requirements on the quantization: 
\begin{itemize}
\item the quantization $q$ of a vector $x$ should have the same length as $x$;
\item the quantization should be operated in an online manner, which means   $q_i$ (the $i^{\textrm{th}}$ entry of $q$) only depends on the past and current inputs $x_1,\cdots,x_i$, not the future ones $x_{i+1},\cdots$; 
\item  the alphabets $\mathcal{A}_i$ for each $q_i$ are the same and fixed in advance, i.e., $\mathcal{A}_1 =\mathcal{A}_2= \cdots= \mathcal{A}_n =\mathcal{A}$. Together they form the codebook $\mathcal{C}=\mathcal{A}^n$; 
\item since quantization is implemented in the analog hardware, the mathematical operations involved should be as simple as possible. In particular, addition and subtraction are preferred over multiplication and division. 
\end{itemize}
Here, for simplicity, we set the alphabet $\mathcal{A}$ to be a bounded, evenly spaced grid with step-size $\delta$,
\begin{equation}\label{eq:alphabet}
\mathcal{A}_{\delta}=\{c+J\delta, \ \ J\in \mathbb{Z},\ J_1\leq J\leq J_2\},
\end{equation}
where $J_1$, $J_2$, and $c$ are pre-selected numbers.
\subsection{Memoryless scalar quantization and Sigma Delta quantization}
To avoid distraction, we review the existing quantization schemes directly in the context of image quantization. Let $X \in [a,b]^{N,N}$ be the matrix that stores the pixel values of a grayscale image. From now on, we will deem the matrix $X$ as an image. 

\noindent\textbf{Memoryless Scalar Quantization (MSQ)}: Suppose the alphabet $\mathcal{A}=\mathcal{A}_\delta$ is defined as in \eqref{eq:alphabet}, then the scalar quantization, denoted as $Q_{\mathcal{A}}: [a,b]\rightarrow \mathcal{A}$  will take a scalar as input and round it off to the nearest element in the alphabet
\[
Q_{\mathcal{A}}(z) \in \arg\min\limits_{v\in \mathcal{A}} |v-z|, \quad z\in [a,b].
\]
When the input is a sequence, we can apply scalar quantization to each entry of the sequence independently, resulting in the Memoryless Scalar Quantization (MSQ). In terms of image quantization, MSQ quantizes each pixel of $X$ independently,    
\[
\mathcal{A}^{N,N} \ni q = Q_{\mathcal{A}}^{MSQ} (X),  \textrm{ with } q_{i,j} = Q_{\mathcal{A}}(X_{i,j}),\quad 1\leq i,j\leq N.
\]
Here $q_{i,j}$  and $X_{i,j}$ are the $(i,j)^{\textrm{th}}$ entries of the quantized and the original images $q$ and $X$, respectively. MSQ is the state-of-the-art quantization in imaging devices, such as cameras.

\noindent\textbf{$\Sigma\Delta$ quantization:}
$\Sigma\Delta$ quantization was first proposed for bandlimited functions (see, e.g. \cite{inose1963unity,daub-dev,G-exp,schreier2005understanding}). As an adaptive quantization, it was shown to be more efficient than MSQ in a variety of applications \cite{powell2013quantization,krahmer2014sigma,wFourier,SWY18,SWY17,LR19}. The adaptiveness comes from the fact that it utilizes quantization errors of previous samples to increase the overall accuracy of the sequence. Suppose the sample sequence is $y=(y_1,\cdots,y_m)$, the first-order $\Sigma\Delta$ quantization (e.g., \cite{DGK10,SWY18}) $q =Q_{\mathcal{A}}^{\Sigma\Delta,1}(y)$ is defined by the following iterations
\begin{align}
\begin{split}
q_i& = Q_{\mathcal{A}}(y_i+u_{i-1}),   \label{eq:sd} \\
(D u)_i &:= u_i-u_{i-1}=y_i-q_i.
\end{split}
\end{align}
Here $D$ is the forward finite difference operator/matrix, with $1$s on the diagonal and $-1$s on the sub-diagonal. The $u_i$ in these equations is called the state variable, and it stores the accumulated quantization error up to the $i^\text{th}$ iteration. $u_0$ is usually initialized to 0, later $u_i$'s are computed from  the second equation in \eqref{eq:sd}. The first equation of \eqref{eq:sd} computes the quantization at each iteration. We see that instead of directly quantizing the $i^{\text{th}}$ input $y_i$, we now add to it the previous state variable $u_{i-1}$ and then apply the scalar quantization to the sum. The addition of $u_{i-1}$ to $y_i$ ensures the use of feedback information. We  call \eqref{eq:sd} the first-order quantization scheme because it only uses the latest state variable $u_{i-1}$. More generally, one can define the $r^\text{th}$-order $\Sigma\Delta$ quantization by utilizing $r$ previous state variables, $u_{i-1}$, $u_{i-1}, u_{i-2},\cdots, u_{i-r}$. More precisely, denote the $r^\text{th}$-order quantization by $q =Q_{\mathcal{A}}^{\Sigma\Delta,r}(y)$, then each entry $q_i$ of $q$ is obtained by 
\begin{align*}
q_i& = Q_{\mathcal{A}}(g_{r}(u_{i-1},\cdots,u_{i-r})+ y_i), \\
(D^r u)_i &:=y_i-q_i, 
\end{align*} 
\sloppy where the $r^{\text{th}}$-order finite difference operator is defined via  $D^r u:=D(D^{r-1}u)$, and $g_r$ is some general function aggregating the previous state variables $u_{i-1},\cdots,u_{i-r}$. In this paper, we employ the usual choice of $g_r$,
\[
g_{r}(u_{i-1},\cdots,u_{i-r}) = \sum\limits_{j=1}^r(-1)^{j-1}\binom{r}{j}u_{i-j}.
\]
When using the $\Sigma\Delta$ quantization scheme on a 2D image $X$, we need to convert the image $X$ into sequences. One way to do this is by applying $\Sigma\Delta$ quantization independently to each column (or row), 
\[
q = Q_{\mathcal{A}}^{\Sigma\Delta, r}(X):= [q_1,\cdots,q_N] , \ \ q_j = Q^{\Sigma\Delta,r}_{\mathcal{A}_\delta}(X_j), \ \ j=1,\cdots,N,
\]
where $X_j$ is the $j^{\text{th}}$ column of $X$.
As this procedure may create discontinuity across columns/rows, we will propose in Section \ref{sec:highSD} a 2D $\Sigma\Delta$ quantization scheme that allows a more continuous quantization/reconstruction.
\subsection{Noise shaping effect of the adaptive quantizers}
The main advantage of Sigma Delta quantization over MSQ is its adaptive usage of feedback information. The feedback information allows a quantizer to wisely allocate bits to efficiently store the entire signal.  For example, adaptive quantization has been known to be extremely efficient for the class of slowly varying signals \cite{daub-dev}. Intuitively speaking, this is because adaptive quantizers can push most quantization errors to the high-frequency region, known as the  noise-shaping effect \cite{chou2015noise,benedetto2006sigma,chou2013beta}, so that the low-frequency region where the slowly varying signals reside are relatively clean. As displayed in Figure \ref{fig:compare}, via a comparison with MSQ, we see that with the same random sequence as input, the error of MSQ is uniformly distributed across the entire spectrum, while that of Sigma Delta quantization is mostly concentrated in the high-frequency regions. Mathematically speaking, the noise shaping effect is a direct consequence of the definition of Sigma Delta quantization. Notice that the first-order Sigma Delta quantization \eqref{eq:sd} has an equivalent matrix form of
\[
y- q = Du, \ \  \|u\|_{\infty} \leq \delta/2 ,
\]
where $\delta$ is the quantization step size in \eqref{eq:alphabet} and $D$ is the finite difference matrix. 
 This is saying that the quantization error $y-q$ is in the set $D (B_{\|\cdot\|_{\infty} }(\delta/2))$, which  is the  $\ell_{\infty}$-ball with radius $\delta/2$   linearly transformed by the matrix $D$. 
 Likewise, for the $r^{\text{th}}$-order $\Sigma\Delta$ quantization ($r \in \mathbb{Z}_+$), we have
\[
y - q = D^ru, \ \ \|u\|_{\infty} \leq \delta/2 ,
\]
which means that the quantization error $y-q$ lies in the $\ell_{\infty}$-ball linearly transformed by $D^r$. 

This is to say that, along each singular vector direction of $D^r$, a scaling by the corresponding singular value is applied to the $\ell_{\infty}$-ball.  We have known from \cite{GLPSY13} that the singular vectors of $D$ are almost aligned with the Fourier basis, and the singular values of $D$ increase with frequencies. Therefore, when computing $Du$ for some $u$ in the $\ell_{\infty}$-ball, the low-frequency components of $u$ would be compressed more than the high-frequency ones as they correspond to smaller singular values of $D$. One can numerically verify this unbalanced scaling by $D^r$ for various values of $r$. Specifically, we hit $D^r$ by sinusoids with various frequencies and compute the ratio 
\[\rho_r(w)=\frac{\|D^r \{\sin(wt_n)\}_{n=1}^N\|_2}{\| \{\sin(wt_n)\}_{n=1}^N\|_2},\]
{where $t_n$, with $n=1,...,N$, are uniform samples in $[0,2\pi]$.}  Figure \ref{fig:noise_shaping} is a plot of $\rho_r(w)$ when setting $N=200$, $r=1,2,3$, and $1\leq w\leq 99$. One can see that sinusoids with lower frequencies have less energy left after hitting with $D^r$, especially when the order of quantization $r$ is large, thus confirming the assertion that components with lower frequencies are cleaner under Sigma Delta quantization.
\begin{figure}[H]
\center
\includegraphics[scale=.5]{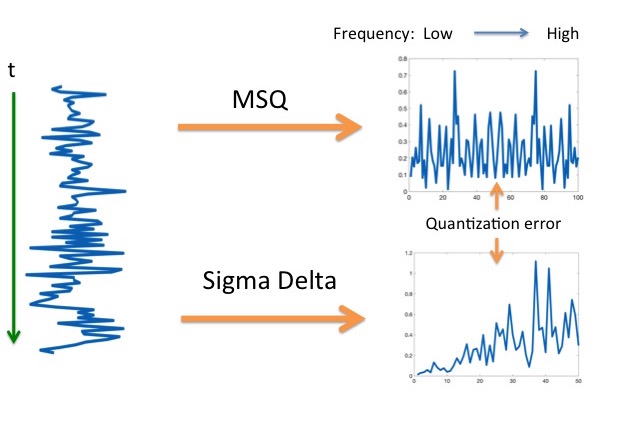}
\caption{The quantization errors of MSQ and $\Sigma\Delta$ in the Fourier domain.}
\label{fig:compare}
    \end{figure}
    \begin{figure}[H]
    \center
\includegraphics[scale=.5]{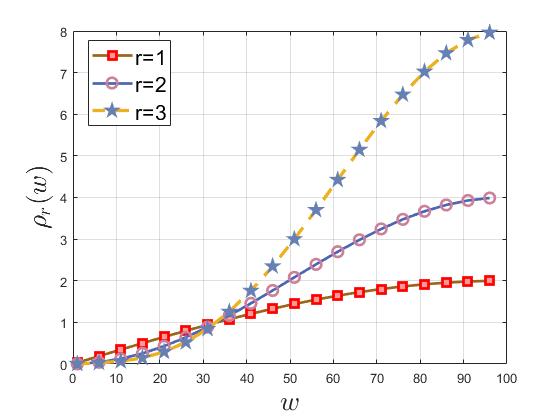}
\caption{The compression rate of $D$, $D^2$, and $D^3$ on the discrete sinusoid $\{\sin(wt_n)\}_{n=1}^N$ with various frequencies $w\in [1,99]$.}
\label{fig:noise_shaping}
    \end{figure}
Since Sigma Delta quantization introduces smaller errors for low-frequency components, it was primarily used for quantizing low-frequency (i.e., slowly-varying) vectors. 
For instance, the sequence of dense (over Nyquist-rate) samples of audio signals can be deemed as low-frequency vectors, for which the superior performance of Sigma Delta quantization has been shown in \cite{daub-dev,G-exp}. Images, on the other hand, are not purely made of low-frequencies. Because sharp edges, as important components of images, have slowly decaying Fourier coefficients. Therefore it is not obvious whether applying Sigma Delta quantization to images is beneficial.
\subsection{A quick review of image quantization}\label{sec:image}
Despite the importance of quantization in image acquisition, the non-adaptive MSQ is thus far the state-of-the-art method in digital imaging devices. The major drawback of MSQ is that when the bit-depth (i.e., the number of bits used to represent each pixel) is small, it has a color-banding artifact: similar colors merge into one, causing fake contours and plateaus in the quantized image (see (B) of Figure \ref{fig:three cat graphs}). A famous technique called dithering \cite{roberts1962picture,schuchman1964dither} reduces color-banding by randomly perturbing the pixel values (e.g., adding random noise) before quantization. It then breaks artificial contour patterns into less harmful random noises.  However, this random noise is still quite visible as shown in (C) of Figure \ref{fig:three cat graphs}. A more fundamental issue is that dithering only randomizes the quantization error instead of reducing it. The same amount of error still exists in the quantized image and will manifest itself in other ways.

Another method to avoid color-banding is digital half-toning, proposed in the context of binary printing, where pixel values are converted to 0 or 1, leading to a possibly severe color-banding artifact. To mitigate it, the digital half-toning was proposed based on the ideas of sequential pixel quantization and error diffusion. Error diffusion means the quantization error of the current pixel will spread out to its neighbors to compensate for the overall under/over-shooting. The rate of spreading is set to empirical values that minimize the overall $\ell_2$ quantization error of an entire image class. Error diffusion works under a similar assumption as the Sigma Delta quantization that the image intensity varies slowly and smoothly. In a sense, it trades color-richness with spatial resolution. Similar to dithering, error diffusion does not reduce the overall noise but only redistributes it.
\begin{figure}[htbp]
     \centering
     \begin{subfigure}[t]{0.3\textwidth}
         \centering
         \includegraphics[width=\textwidth]{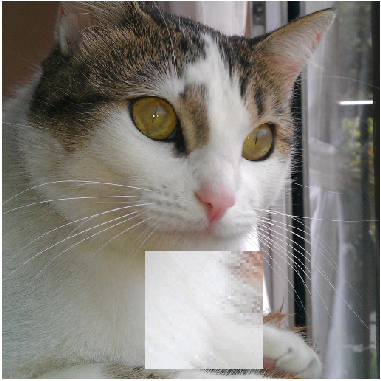}
         \caption{Original image, bit-depth=8.}
     \end{subfigure}
     \hfill
     \begin{subfigure}[t]{0.3\textwidth}
         \centering
         \includegraphics[width=\textwidth]{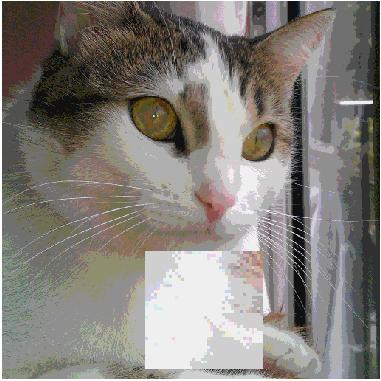}
         \caption{Quantized image, bit-depth=3, with no dithering. }
         \label{fig:8nodither}
     \end{subfigure}
     \hfill
     \begin{subfigure}[t]{0.3\textwidth}
         \centering
         \includegraphics[width=\textwidth]{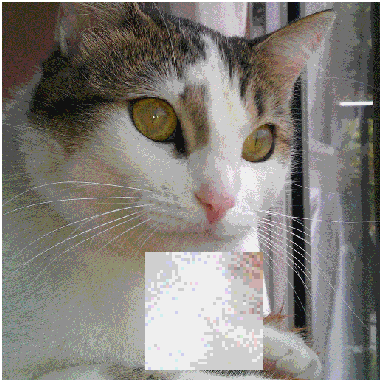}
         \caption{Quantized image, bit-depth=3, with dithering.}
     \end{subfigure}
        \caption{3-bit MSQ quantization. (A) Original image; (B) Quantized image with MSQ (no dithering), we observe a strong color-banding effect. (C) Quantized image with dithering and MSQ. The image contains observable noisy spots.}
        \label{fig:three cat graphs}
\end{figure}
\subsection{Contribution}
From the discussion in Section \ref{sec:image}, we see that both dithering and digital halftoning are only redistributing the quantization error instead of compressing it.  In contrast, the method we introduce in this paper achieves a real reduction of the quantization error upon that of MSQ.  Explicitly, suppose $N^2$ is the total number of pixels and $s$ is the number of pixels representing curve discontinuities (e.g., edges) in the image, then our method reduces the quantization error from $O(N)$ to $O(\sqrt s)$. This is achieved by combining Sigma Delta quantization with an optimization-based reconstruction algorithm. We observe in the numerical experiments that both the low and high frequency errors are reduced. 

Due to the use of the total variation norm in the decoding optimization, our result is closely related to super-resolution theory \cite{candes2013super,candes2014towards,li2017elementary} in compressed sensing, where it has been demonstrated that a sparse signal can be super-resolved using an $\ell_1$-norm minimization if all the spikes in the signal are well-separated.  We stress that the error reduction we achieve for image quantization does not require the edges of the image to satisfy this separation condition, although if the separation condition is met, a further error reduction can be achieved.

Besides Sigma Delta quantization, there exist other adaptive encoders  (e.g.,  Beta encoder \cite{chou2013beta,chou2016distributed,chou2017distributed,gunturk2019high}). These encoders had been used successfully on 1D audio signals to improve the bit-rate-distortion over MSQ, but none of them has been used for images. One major reason is that previous analyses (e.g., \cite{benedetto2006sigma,wFourier,blum:sdf,daub-dev,DGK10}) all indicated that these adaptive schemes can only compress low-frequency noise by sacrificing the high-frequency accuracy. While this might be a good idea for audio signals, one does not want to make such a sacrifice when it comes to images. In this paper, we demonstrate that a carefully designed decoder can help retain the high-frequency information while compressing the low-frequency noise. Therefore, the overall result outperforms MSQ. We only consider Sigma Delta quantization in this paper and leave the study of other adaptive quantizers on images as future work.  

Unlike previous works on quantization under frame or compressed sensing measurements (see, e.g., \cite{powell2013quantization,bodmann2007frame,KSW12,blum:sdf,GLPSY13,krahmer2014sigma,benedetto2006sigma,lammers2010alternative,goyal1998quantized,baraniuk2017exponential,dirksen2018fast,jacques2013robust,plan2012robust,plan2014dimension,huynh2020fast,dirksen2020one,dirksen2018fast,feng2019quantized}), where samples are assumed to be taken by random Gaussian/sub-Gaussian or Fourier measurements, here we allow a direct quantization on each pixel and therefore ensure maximal practicality. 

Another contribution of the work is an extension of Sigma Delta quantization to higher dimensions. We found that the proposed two-dimensional Sigma Delta quantization can effectively reduce the artifacts in the reconstruction while being as fast as the 1D quantization. 
\section{Proposed Method}\label{sec:method}
We propose an adaptive quantization framework for natural images. Given an input image, the quantization workflow involves:
\begin{itemize} 
\item segmentation: divide the image into columns or rectangular patches. Here and throughout the paper, we use the term ``patch'' to refer to the set of image pixels restricted to a rectangular window of a certain size;
\item quantization: use the existing 1D $\Sigma\Delta$ quantization in the literature to quantize each column in parallel or use the proposed 2D $\Sigma\Delta$ quantization (in Section \ref{sec:highSD}) to quantize each rectangular patch in parallel; 
\item reconstruction: run the proposed decoding algorithms (Section \ref{sec:theorems}) to restore the columns or patches and stack them into the final reconstructed image.
\end{itemize}
\subsection{The encoders}\label{sec:encoder}
We consider two types of encoders/quantizers in this paper.
\begin{itemize}
\item Encoder 1 ($Q_{col}$): 1D $r^{\textrm{th}}$-order Sigma Delta quantization applied to each column of the image (i.e., column-by-column quantization); 
\item Encoder 2 ($Q_{2D}$): 2D $r^{\textrm{th}}$-order Sigma Delta quantization applied to each patch of the image (i.e., patch-by-patch quantization).
\end{itemize}

An important question one may ask is the practicality of the proposed adaptive quantizers when used in commercial cameras. A natural concern is the waiting time.  Unlike MSQ that quantizes each pixel in parallel, Sigma Delta quantization can only be performed sequentially, which seems to inevitably introduce extra waiting time. However, this is not the case as current cameras are already using sequential quantization architectures for consistency, energy, and size considerations. More specifically, in current cameras, to reduce the number of ADC (Analog to Digital Converters) and save energy, one column of pixels or the whole image are assigned to one ADC,
which means the pixels in one column or those of the entire image need to wait in a queue to be quantized. This architecture is perfect for Sigma Delta quantization. The only minor change one needs to make is adding additional memory units to the circuit to store the state variables.
 
The structure of the rest of the paper is as follows. In Section \ref{sec:highSD}, we introduce the proposed 2D Sigma Delta quantization along with some of its properties. In Section \ref{sec:assumptions}, we introduce three image models which are of interest to this paper. In Section \ref{sec:decoders}, we present three decoders associated with each of the three image models and summarize the reconstruction accuracy. The main theorems containing the reconstruction error bounds and their derivations can be found in Section \ref{sec:theorems}. In Section \ref{sec:optimization} of the appendix, we describe an efficient algorithm for solving the proposed optimization problems in the decoding process. Finally, in Section \ref{sec:experiments}, we perform numerical experiments to verify the conclusion of the theorems and to provide more evidence of the efficacy of the proposed method in real applications.  
\subsection{High dimensional Sigma Delta quantization}\label{sec:highSD}
Although we can apply 1D $\Sigma\Delta$ quantization column by column to an image, it is likely to create discontinuities along the horizontal direction. As images are two-dimensional arrays, a two-dimensional quantization scheme is needed in maintaining the overall continuity. For this purpose, we propose the first high-dimensional Sigma Delta quantization. 

Recall that in a nutshell, the existing 1D first-order Sigma Delta quantization $Q_{\mathcal{A}}^{\Sigma\Delta,1}$:  $[a,b]^N \ni y \rightarrow q\in \mathcal{A}^N$ ($\mathcal{A}$ is the alphabet) is defined by constructing for any $y\in[a,b]^N$, two vectors $q$ (the quantization) and $u$ (the state variable) of the same length as $y$ and obeying
\begin{enumerate}
\item[(A1)]  (boundedness/stability):\label{it:A1} $\|u\|_{\infty} \leq C$, for some constant $C$ independent of $N$;
\item[(A2)] (adaptivity): \label{it:A2}
$ u_{i} = u_{i-1} +y_i-q_i,  \forall i,$ with $u_i$, $y_i$, $q_i$ being the $i^{\text{th}}$ component of the vectors $u$, $x$ and $q$, respectively;
\item [(A3)] (causality): \label{it:A3} $q$ only depends on the history of the input $y$, that is $q_i = f(y_{i}, y_{i-1}, \cdots, y_{1})$, for any $i$ and some function $f$.
\end{enumerate}

We now extend these conditions to two dimensions, and the extensions to higher dimensions are similar. A 2D Sigma Delta quantization is well-defined and denoted as $Q_{\mathcal{A},2D}^{\Sigma\Delta}$:  $[a,b]^{N,N} \ni y \rightarrow q\in \mathcal{A}^{N,N}$, if we can construct for any input image $y\in[a,b]^{N,N}$, two $N\times N$ matrices $q$ and $u$ satisfying the following three conditions: 
\begin{enumerate} 
\item[(A1')] (boundedness/stability): $\|u\|_{{\max}} \leq C$ ($\|\cdot\|_{{\max}}$ denotes the entry-wise maximal magnitude of a matrix) for some constant $C$; 
\item[(A2')] (adaptivity):  $ u_{i,j} = u_{i,j-1}+u_{i-1,j} -u_{i-1,j-1} + y_{i,j}-q_{i,j}$ which has a matrix representation of $DuD^T = y-q$; and 
\item[(A3')] (causality): $q_{i,j} = f(\{y_{i',j'}\}_{i'\leq i, j'\leq j}$) for some function $f$. 
\end{enumerate}

Provided that the quantization alphabet is large enough, one can show that a pair of $u$ and $q$ that satisfies (A1')-(A3') can be constructed through the recursive formula
\begin{align}\label{eq:2Dq}
\begin{split}
q_{i,j} = Q_{\mathcal{A}}(u_{i,j-1}+u_{i-1,j}-u_{i-1,j-1}+y_{i,j}),\\
u_{i,j} = u_{i,j-1}+u_{i-1,j}-u_{i-1,j-1} + y_{i,j}-q_{i,j}, 
\end{split}
\end{align}
where $i,j>1$. When $i=1$ or $j=1$, the first row and column can be initialized using the 1D $\Sigma\Delta$ quantization. In the extreme case of  $1$-bit quantization, the existence of $u$ and $q$ obeying (A1')-(A3') for all $y\in[a,b]^{N\times N}$ is unknown, and we leave this as further work.  Here we focus on the case when the bit-depth is greater than or equal to two. We first show in Proposition \ref{pro:bit} that a stable 2D Sigma Delta quantization exists and then in Proposition \ref{pro:opt} that the uniform alphabet with a certain step-size has the optimal stability among all alphabets of the same bit-depth.
\begin{proposition}\label{pro:bit}
For a given pair of real values $a,b$ with $a<b$ and a bit-depth $d\geq 2$, there exists an alphabet  $\mathcal{A}$ such that for any 2D array $y\in[a,b]^{N,N}$, the $u$ and $q$ generated by \eqref{eq:2Dq} satisfy (A1')-(A3') with $C=\frac{b-a}{2(2^d-3)}$.
\end{proposition}
\begin{proof}
 Let $C = \frac{b-a}{2(2^d-3)}$, and create the alphabet  as
$$\mathcal{A}=\{a-2C,a,a+2C,\cdots,b,b+2C\}.$$
Then $|\mathcal{A}|=\frac{b+2C-(a-2C)}{2C}+1=2^d$. Now we use the second principle of induction to show that $u$ generated by \eqref{eq:2Dq} satisfies $\|u\|_{\textrm{max}}\leq C$.
\begin{itemize}
\item Induction hypothesis: if for all the pairs $(m,n)$ such that $m\leq i,\ n\leq j,\ m+n<i+j,\ |u_{m,n}|\leq C$, then $|u_{i,j}|\leq C$.

\item Base case: $|u_{1,1}|=|y_{1,1}-q_{1,1}|=|y_{1,1}-Q_\mathcal{A}(y_{1,1})|\leq C$.

\item Induction step: if $i=1,\ q_{i,j}=Q_\mathcal{A}(y_{i,j}+u_{i,j-1})$, then by the induction hypothesis, we have $a-C\leq y_{i,j}+u_{i,j-1}\leq b+C$, thus $|u_{i,j}|=|y_{i,j}+u_{i,j-1}-q_{i,j}|\leq C$. The same reasoning follows when $j=1$.

If $i,j\geq 2$, by the induction hypothesis, we have $a-3C\leq y_{i,j}+u_{i,j-1}+u_{i-1,j}-u_{i-1,j-1}\leq b+3C$. Hence $|u_{i,j}|=|y_{i,j}+u_{i,j-1}+u_{i-1,j}-u_{i-1,j-1}-q_{i,j}|\leq C$.
\end{itemize}
\end{proof}
The next proposition shows that the stability constant $C = \frac{b-a}{2(2^d-3)}$ associated with the uniform alphabet
$$\mathcal{A}=\{a-2C,a,a+2C,\cdots,b,b+2C\}$$
in Proposition \ref{pro:bit} is optimal. We provide the proof of Proposition \ref{pro:opt} in Section \ref{sec:append1} of the appendix.
\begin{proposition}\label{pro:opt}
For a fixed bit-depth $d\geq 2$, the alphabet $\mathcal{A}$ for the 2D $\Sigma\Delta$ quantization given in Proposition \ref{pro:bit} is optimal, in the sense that if $\widetilde C$ is the stability constant of any other $d$-bit alphabet $\widetilde{\mathcal{A}}$ (not necessarily equally spaced), then it is necessary that $\widetilde{C}\ge C$, where $C$ is the stability constant of $\mathcal{A}$.
\end{proposition}
\begin{remark} 
The computational complexity of the 2D quantization scheme \eqref{eq:2Dq} is $O(N^2)$. However, noticing that for a fixed $t \in \{2, 3,\cdots, 2N-1 \}$, all $q_{i,j}$ with $i + j = t$ (the points on the $t^{\textrm{th}}$ anti-diagonal)  can be computed in parallel, therefore the quantization time can be reduced to $O(N)$.
\end{remark} 
\begin{remark} The matrix representation of \eqref{eq:2Dq} is
\[
y-q = DuD^T.
\]
It is straightforward to extend this first-order quantization to higher orders. If $r\geq 1$, the $r^{\textrm{th}}$-order quantization obeys the matrix form recursive formula
\[
y - q = D^ru(D^r)^T.
\]\end{remark}
\subsection{Notation and Assumptions}\label{sec:assumptions}
Throughout, we assume that the image $X$ to be quantized is an $N\times N$ matrix. Everything discussed in the paper can be easily generalized to rectangular matrices. We write $X = (\underline x_1,\underline x_2,\cdots,\underline x_{N})=(\bar x_1,\bar x_2,\cdots,\bar x_{N})^T$, where $\underline x_i$ denote its columns and and $\bar x_i$ the rows. Let $D$ be the $N\times N$ difference matrix with $1$s on the diagonal and $-1$s on the sub-diagonal and $D_1$ be the circulant difference matrix with an extra -1 at the upper right corner. For vectors, we denote by $\|\cdot\|_1$ the  $\ell_1$-norm and by $\|\cdot\|_{\infty}$ the  $\ell_\infty$-norm. For a continuous function $f$, we denote by $\|f\|_{L^{\infty}}$ its $L^{\infty}$-norm. For matrices, $\|\cdot\|_{\max}$ stands for the maximum absolute entry, and $\|\cdot\|_1$ is the entry-wise $\ell_1$-norm, i.e., $\|A\|_1=\|\textrm{vec(A)}\|_1$. For both vectors and matrices, $\|\cdot\|_0$ is the number of nonzero entries. Also, $F=\{\omega^{kj}\}_{k,j=0}^{N-1}$ ($\omega = e^{-2\pi i/N}$) denotes the unnormalized $N\times N$ DFT matrix. We use $\mathcal{F}_k$ to represent the row of $F$ corresponding to frequency $k$, which can also be recognized as the operator that maps a vector to its discrete Fourier coefficient at frequency $k$. We use $F_L$ for the matrix containing the rows of $F$ associated with frequencies in the set $\{-L,-L+1,\cdots,L\}$, and denote $P_L=\frac{1}{N}F_L^* F_L$. In addition, $\ast$ stands for the circulant convolution operator between two vectors, {$\mathbb{T}$ is the 1D torus,} and $A(n)\lesssim B(n)$ means that there exists some constant $c$ independent of $n$,  such that $A(n)\leq cB(n)$ for all $n$. 

Our general assumption is that images have nearly sparse gradients. To be more precise, we consider three classes of images satisfying one of the three assumptions below.

\begin{Assumption}\label{assumption1} (\textbf{Class 1: $\beta^{\text{th}}$-order column-wise or row-wise sparsity condition}) Suppose $X \in [a,b]^{N,N}$ is an image, the columns or rows of $X$ are piece-wise constant or piece-wise linear. Explicitly, for $\beta=1$ or $2$, there exists $s< N$, 
\[
\| (D^\beta)^T\underline x_i\|_0\leq s,\;\forall i=1,2,\cdots, N,\ \text{or}\;\| \bar x_j^T D^\beta\|_0\leq s,\; \forall j=1,2,\cdots, N.
\]
If $\beta=1$, the columns or rows of image $X$ are piece-wise constant, if $\beta=2$, they are piece-wise linear. 
\end{Assumption}

\begin{Assumption}\label{assumption2}(\textbf{Class 2: $\beta^{\text{th}}$-order 2D sparsity condition}) Suppose $X \in [a,b]^{N,N}$ is an image, both columns and rows in $X$ are piece-wise constant or piece-wise linear. Explicitly, for $\beta=1$ or $2$, there exists $s< 2N^2 $, such that 
    $$\| (D^\beta)^T X\|_0+\|X D^\beta\|_0\leq s.$$
\end{Assumption}
\begin{Assumption}\label{assumption3}(\textbf{Class 3: $\beta^{\text{th}}$-order minimum separation condition}) $X$ satisfies Assumption \ref{assumption1}. In addition, the $\beta^{\text{th}}$-order differences of each column or row of $X$ satisfy the $\Lambda_M$-minimum separation condition defined below for some small integer $M\ll N$. Explicitly, this means that for $\beta=1$ or $2$, $D_1^\beta\underline x_i$ with $i=1,2,\cdots,N$ or $\bar x_j^T (D_1^\beta)^T$ with $j=1,2,\cdots,N$ satisfy the $\Lambda_M$-minimum separation condition. Note that here $D_1$ is the circulant difference matrix.
\end{Assumption}
\begin{definition}\textbf{($\Lambda_M$-minimum separation condition)} For a vector $x\in\mathbb{R}^N$, let $S\subseteq\{0,1,\cdots,N-1\}$ be the support of $x$. We say that $x$ satisfies the $\Lambda_M$-minimum separation condition if
\begin{equation}\label{eq:mini sepa}
\min_{s,s'\in S,s\not= s'}\; \frac{1}{N}d(s,s')\geq \frac{2}{M},
\end{equation}where $d(s,s'):=\min\{|s-s'|, |s+N-s'|, |s-N-s'|\}$ is the wrap-around distance on the set $\{0,...,N-1\}$. 
\end{definition}
In addition, following \cite{li2017elementary}, we use the definition $C(\mathbb{T};\Lambda_M)$ for the space of trigonometric polynomials up to degree $M$ on the 1D torus $\mathbb{T}$, i.e.,
\[
C(\mathbb{T};\Lambda_M)=\{f\in C^\infty(\mathbb{T}):f(x)=\sum_{k=-M}^M a_k e^{i2\pi kx},a_k\in\mathbb{C}\}.
\]
\subsection{The proposed decoders and error bounds}\label{sec:decoders}
For each image class, we will set the encoder $Q$ to be either the $Q_{col}$ or $Q_{2D}$ defined in Section \ref{sec:encoder}, and let $X$ be the image of interest. The various decoders we shall propose for different classes of images will all be in the general form of 
\begin{equation}\label{eq:unify}
\hat X=\argmin_Z\ \ell(Z,\beta) \ \ \ \textrm{subject to}\ \;\rho(Z,r)\leq c.
\end{equation}
Here $\beta$ is 1 or 2 depending on whether the image is assumed to be (approximately) piece-wise constant or piece-wise linear, $\ell(Z,\beta)$ is some penalty that encourages sparsity in the gradients, $r$ is the order of Sigma Delta quantization, and $\rho(Z,r)\leq c$ is the feasibility constraint corresponding to a specific quantization scheme.  Let $\hat X$ be a solution to the general form \eqref{eq:unify}, we shall obtain reconstruction error bounds of the following type
\begin{equation}\label{eq:general_error}
    \|\hat X-X\|_F\leq C(\beta,r,N,\delta),
\end{equation}
where $\beta$, $r$ are the same as defined above \eqref{eq:general_error}, $N$ is the size of the image and $\delta$ is the step-size of the alphabet.

Now we specify what the $\ell(Z,\beta)$ and $\rho(Z,r)$ in \eqref{eq:unify} should be for each class of images and provide the associated error bound. Note that in Assumption \ref{assumption1} and Assumption \ref{assumption3}, the images are allowed to be either row-wise sparse or column-wise sparse. Since their treatments are the same, from now on, we assume the images are column-wise sparse. 
\begin{itemize}
    \item Class 1: $X$ satisfies Assumption \ref{assumption1} with order $\beta=1$ or $2$ and sparsity $s$. We use the encoder $Q_{col}$ proposed in Section \ref{sec:encoder} with order $r\geq \beta$, and an alphabet with a step-size $\delta$. We solve the following optimization problem for a reconstruction 
    \begin{equation}\label{eq:case1}
    \hat X=\argmin_Z \|(D^\beta)^T Z\|_1\quad \textrm{subject to}\ \|D^{-r}(Z-Q_{col}(X))\|_{\textrm{max}}\leq\delta/2,
    \end{equation}
    where $\|\cdot\|_{\textrm{max}}$ is the maximum entry-wise absolute value.
    Theorem \ref{thm1} below shows that the reconstruction error is \[
    \|\hat X-X\|_F\leq C\sqrt{sN}\delta,
    \]
   where $C$ is a universal constant.
    \item Class 2: $X$ satisfies Assumption \ref{assumption2} with order $\beta=1$ or $2$ and sparsity $s$. We use the encoder $Q_{2D}$ (proposed in Section \ref{sec:highSD}) with order $r\geq\beta$ and an alphabet of step-size $\delta$.  We use the following optimization 
    \begin{align}\label{eq:case3}
    \begin{split}
    &\hat X = \argmin_Z \|(D^\beta)^T Z\|_1+\|ZD^\beta\|_1\\ &\textrm{subject to}\ \|D^{-r}(Z-Q_{2D}(X))(D^{-r})^T\|_{\textrm{max}}\leq\delta/2,
    \end{split}
    \end{align}
    Theorem \ref{thm:2Drecon} shows that the reconstruction error is bounded by 
    \[
    \|\hat X-X\|_F\leq C\sqrt{s}\delta,
    \]
    where $C$ is a universal constant.
     \item Class 3: $X$ satisfies Assumption \ref{assumption3} with order $\beta =1\ \text{or}\ 2$ and sparsity $s \ll N $.  The encoder is $Q_{col}$ with an alphabet spacing of $\delta$ (proposed in Section \ref{sec:encoder}) and $r\geq \beta$. Here we define a new alphabet $\widetilde{\mathcal{A}}$ with a smaller step-size $\widetilde\delta\coloneqq\frac{2\delta}{(2N)^r}$ to quantize the last $r$ entries of each column:
     $${\widetilde{\mathcal{A}}_{\widetilde{\delta}}\coloneqq\{\widetilde a,\widetilde a+\widetilde\delta,\widetilde a+2\widetilde\delta,\cdots,\widetilde a+K\widetilde\delta,\widetilde b\},}$${where $\widetilde a=a-(2^{r-1}-\frac{1}{2})\delta$, $\widetilde b=b+(2^{r-1}-\frac{1}{2})\delta$, and  $K=\max\{j\in\mathbb{Z},\ \widetilde a+j\widetilde\delta<\widetilde b\}.$} The total number of boundary bits is of order $O(r\log N)$, which is negligible compared to the $O(N)$ bits needed for the interior pixels. 
     Since $\widetilde{\mathcal{A}}_{\widetilde{\delta}}$ is finer than $\mathcal{A}_{\delta}$, the feasibility constraint is tighter at the last $r$ entries of each column:
     $$\|[D^{-r}(X-Q_{col}(X))]_{N-r+1:N,:}\|_{\textrm{max}}\leq \left(\frac{1}{2N}\right)^r\delta,$$
      where $[D^{-r}(Z-Q_{col}(X))]_{N-r+1:N,:}$ refers to the last $r$ rows of $D^{-r}(Z-Q_{col}(X))$, and $\|\cdot\|_{\textrm{max}}$ is the largest magnitude of the entries.
     Overall, we use the following optimization to obtain the reconstructed image $\hat X$: 
\begin{align}\label{eq:case2}
\begin{split}
&\hat X=\argmin_Z \|D_1^\beta Z\|_1  \\
&\textrm{subject to}\left\{
\begin{aligned}
&\|D^{-r}(Z-Q_{col}(X))\|_{\textrm{max}}\leq\delta/2, \\
&\|[D^{-r}(Z-Q_{col}(X))]_{N-r+1:N,:}\|_{\textrm{max}}\leq \left(\frac{1}{2N}\right)^r\delta. 
\end{aligned}
\right.
\end{split}
\end{align}
In Theorem \ref{thm2}, we obtain the following error bound
\begin{equation*}
\|\hat X-X\|_F\leq C\frac{M^{r+\beta-2}}{N^{r-3}}\delta,
\end{equation*}
\end{itemize}
where $C$ is a universal constant and $M$ is the same as in Assumption \ref{assumption3}.
In the following, we discuss Class 1 in Section \ref{sec:cbc1TV}, Class 2 in Section \ref{sec:2dimquant}, and Class 3 in Section \ref{sec:cbcsepa}.
\section{Main theorems and their proofs}\label{sec:theorems}
Without loss of generality, we assume the image $X$ has bounded pixel values, i.e., $X\in [0,1]^{N,N}$, and images satisfying Assumption \ref{assumption1} and Assumption \ref{assumption3} are column-wise sparse. 
The same results can be similarly obtained for row-wise sparse images.
\subsection{Class 1: Images with no minimum separation }\label{sec:cbc1TV}
In this section, we consider Class 1, where the image $X$ satisfies Assumption \ref{assumption1} with $\beta =1$, or $\beta = 2$ for all the columns. We use the encoder $Q_{col}$ that performs column-by-column quantization and use the decoder \eqref{eq:case1}. Both the encoder and decoder can be decoupled into columns and run in parallel.

For each column $\underline{x}\in\mathbb{R}^N$, let $\underline{q}$ be its $r^{\text{th}}$-order Sigma Delta quantization, i.e., $\underline{q}= Q^{\Sigma\Delta,r}_{\mathcal{A}_\delta}(\underline{x})$ and take $r\geq\beta$. For the simplicity of notation, we use $x$ and $q$ to represent $\underline{x}$ and $\underline{q}$, respectively. The decoder \eqref{eq:case1} decouples into columns. For each column, it is 
\begin{align}\label{eq:cbccase1}
\hat x=\argmin_{z} \|(D^\beta)^T z\|_1\quad \textrm{subject to } \|D^{-r}(z -q)\|_{\infty} \leq \delta/2.
\end{align}
Since $(D^\beta)^T z$ represents the $\beta^{\text{th}}$-order difference of $z$, the $\ell_1$-norm here promotes sparsity of the gradients. The $\ell_\infty$ ball-constraint is a well-known feasibility constraint for Sigma Delta quantization  (e.g., \cite{SWY18}).

The following theorem provides the error bound for this decoder. 
\begin{theorem}\label{thm1} For any $x \in [0,1]^N$ satisfying $\|(D^{\beta})^Tx\|_0 \leq s$ with $\beta=1\ \text{or}\ 2$, and $q=Q^{\Sigma\Delta,r}_{\mathcal{A}_\delta}(x)$ with $\ r\geq \beta$.  
Let $\hat x$ be a solution to \eqref{eq:cbccase1}, then 
\begin{equation}\label{eq:ineq1}
\|\hat x-x\|_2\leq C\sqrt{s}\delta,
\end{equation}
{where $C$ is some absolute constant independent of the choice of $x$.}
\end{theorem} 
\begin{remark} The error bound in Theorem \ref{thm1} is for each column. Putting  all columns together, we have
$$\| \hat X-X\|_F\leq C\sqrt{sN}\delta.$$
\end{remark}
\begin{remark}
For comparison, we specify the quantization error of MSQ and that of the vanilla Sigma Delta decoder. In MSQ, the quantization error for each pixel is $\delta/2$. Since the pixels are quantized independently, the total quantization error of an $N\times N$ image in the Frobenius norm is $N\delta/2$. Similarly, when using $\Sigma\Delta$ quantizers ($Q_{col}$ or $Q_{2D}$) but decoding with the following naive decoder,
\[
\hat X=\textrm{Find } Z \;\;\ \textrm{subject to}\ \|D^{-r}(Z-Q_{col}(X))\|_{\textrm{max}}\leq\delta/2,
\]
the worse-case error is again $O(N\delta)$. This indicates that the TV-norm penalty in the proposed decoder \eqref{eq:cbccase1} is playing a key role in reducing the error to $O(\sqrt{sN}\delta)$.
\end{remark}
Next, we extend Theorem \ref{thm1} to  signals that does not satisfy Assumption \ref{assumption1}.
\begin{theorem}\label{cor1}
 Let $x \in [0,1]^N$ and $q=Q^{\Sigma\Delta,r}_{\mathcal{A}_\delta}(x)$.
Let $\hat x$ be a solution to \eqref{eq:cbccase1}, with $\beta=1\ \text{or}\ 2$, and $\beta\leq r$, then for any integer $1\leq s\leq N$, we have 
\begin{equation}
\|\hat x-x\|_2\leq  C(\sqrt{s}\delta+\sqrt{\sigma_s((D^{\beta})^Tx)\delta}),
\end{equation}
where $\sigma_s(z)$ is the $\ell_1$-tail of $z$, i.e., $\sigma_s(z) = \min_{v \textrm{ is $s$-sparse}} \|v-z\|_1$. $C$ is some absolute constant independent of the signal $x$.
\end{theorem}
\subsubsection{Proof of Theorem \ref{thm1} {and Theorem \ref{cor1}}}
\begin{proof}[Proof of Theorem \ref{thm1}]
Denote $h=(D^\beta)^T(\hat x-x)$. Assume the support set of $(D^\beta)^T x$ is $S$  with cardinality $s$ and the complement of $S$ is $S^C$. Since $\hat x$ is a solution to \eqref{eq:cbccase1}, we have
$$\|(D^\beta)^T x\|_1\geq \|(D^\beta)^T x+h\|_1\geq \|(D^\beta)^T x\|_1-\| h_S\|_1+\| h_{S^C}\|_1,$$
which gives $\| h_S\|_1\geq \| h_{S^C}\|_1$. We can then bound the $\ell_1$-norm of $h$ as
\begin{equation}\label{eq:h}\| h\|_1=\| h_S\|_1+\| h_{S^C}\|_1\leq 2\| h_S\|_1\leq 2^{\beta+r+1}s\delta,
\end{equation}
where the last inequality is due to
\begin{align*}
\| h\|_\infty&=\|(D^\beta)^T(\hat x-x)\|_\infty\\&=\| (D^\beta)^TD^{r}D^{-r}(\hat x-x)\|_\infty \leq 2^\beta\cdot 2^r\cdot \|D^{-r}(\hat x-x)\|_\infty\leq 2^{\beta+r}\delta.
\end{align*}
\eqref{eq:h} is equivalent to 
\begin{equation}\label{eq:thm1_infnorm}
\|(D^\beta)^T(\hat x-x)\|_1\leq 2^{\beta+r+1}s\delta.
\end{equation}
In addition, we also have
\begin{equation}\label{eq:thm1_infnorm2}
\|D^{-\beta}(\hat x-x)\|_\infty=\| D^{r-\beta}D^{-r}(\hat x-x)\|_\infty \leq 2^{r-\beta}\cdot \|D^{-r}(\hat x-x)\|_\infty \leq 2^{r-\beta}\delta.
\end{equation}
\eqref{eq:thm1_infnorm}, \eqref{eq:thm1_infnorm2} above are bounds in $\ell_1$-norm and $\ell_\infty$-norm, respectively. Due to the duality, we can bound the reconstruction error $\| \hat x-x\|_2$ as
\begin{align*}
    \langle\hat x-x,\hat x-x\rangle=\langle (D^\beta)^T(\hat x-x),D^{-\beta}(\hat x-x)\rangle\leq 2^{2r+1}s\delta^2.
\end{align*}
This is equivalent to saying that $\|\hat x-x\|_2\leq C\sqrt{s}\delta$.
\end{proof}
\begin{proof}[Proof of Theorem \ref{cor1}]
The proof follows a similar idea to Theorem \ref{thm1}. Assume that $(D^{\beta})^Ty$ is the best $s$-term approximation to $(D^{\beta})^Tx$, and assume the support of $(D^{\beta})^Ty$ is $S$. Since $\hat x$ is a solution to \eqref{eq:cbccase1}, we have
\begin{align*}
&\quad\ \|(D^{\beta})^Tx\|_1\geq\|(D^{\beta})^T\hat x\|_1
=\|(D^{\beta})^Ty+(D^{\beta})^T(x-y)+(D^{\beta})^T(\hat x-x)\|_1\\
&\geq \|(D^{\beta})^Ty\|_1-\left\|\left[(D^{\beta})^T(x-y)+(D^{\beta})^T(\hat x-x)\right]_S\right\|_1\\&\quad+\left\|\left[(D^{\beta})^T(x-y)+(D^{\beta})^T(\hat x-x)\right]_{S^C}\right\|_1\\
&\geq \|(D^{\beta})^Ty\|_1-\|[(D^{\beta})^T(x-y)]_S\|_1-\|[(D^{\beta})^T(\hat x-x)]_S\|_1+\|[(D^{\beta})^T(\hat x-x)]_{S^C}\|_1\\&\quad-\|[(D^{\beta})^T(x-y)]_{S^C}\|_1\\
&=\|(D^{\beta})^Ty\|_1-\|(D^{\beta})^T(x-y)\|_1+\|[(D^{\beta})^T(\hat x-x)]_{S^C}\|_1-\|[(D^{\beta})^T(\hat x-x)]_{S}\|_1.
\end{align*}
Here we only used the triangle inequality. Applying $\|(D^{\beta})^Tx\|_1\leq\|(D^{\beta})^Ty\|_1+\|(D^{\beta})^T(x-y)\|_1$ to the left hand side of the above inequality and after some simplification, we obtain
\[
\|[(D^{\beta})^T(\hat x-x)]_{S^C}\|_1\leq\|[(D^{\beta})^T(\hat x-x)]_{S}\|_1+2\|(D^{\beta})^T(x-y)\|_1.
\]
This further implies
\begin{equation}\label{eq:cor1norm}
\|(D^{\beta})^T(\hat x-x)\|_1\leq 2\|[(D^{\beta})^T(\hat x-x)]_S\|_1+2\|(D^{\beta})^T(x-y)\|_1.
\end{equation}
Following the same reasoning as in the proof of Theorem \ref{thm1}, we have \begin{equation*}
 2\|[(D^{\beta})^T(\hat x-x)]_S\|_1\leq 2^{\beta+r+1}s\delta,
\end{equation*}
which when plugged into \eqref{eq:cor1norm} gives us
\begin{equation}\label{eq:cor1norm2}
\|(D^{\beta})^T(\hat x-x)\|_1\leq 2^{\beta+r+1}s\delta+2\sigma_s((D^{\beta})^Tx).
\end{equation}
By \eqref{eq:thm1_infnorm2} and \eqref{eq:cor1norm2}, we have
\begin{align*}
    \langle\hat x-x,\hat x-x\rangle=\langle (D^\beta)^T(\hat x-x),D^{-\beta}(\hat x-x)\rangle\leq 2^{2r+1}s\delta^2+ 2^{r-\beta+1}\sigma_s((D^{\beta})^Tx)\delta.
\end{align*}
This is saying that $\|\hat x-x\|_2\leq C(\sqrt{s}\delta+\sqrt{\sigma_s((D^{\beta})^Tx)\delta})$.
\end{proof}
\subsection{Class 2: Decoding high-dimensional Sigma Delta quantized images}\label{sec:2dimquant}
In this section, we consider Class 2, where the image $X$ satisfying Assumption \ref{assumption2} is associated with the encoder $Q_{2D}$. For simplicity, we assume the patch number is 1 (there is only one patch identical to the original image) and $\beta=1$. Results for more than one patches and/or $\beta=2$ follow the same argument. The theorem below establishes the error bound for 2D reconstruction of $X$ from its quantization $Q_{2D}(X)$ using \eqref{eq:case3}.
\begin{theorem}\label{thm:2Drecon}
Suppose the original image $X \in [0,1]^{N,N}$ satisfies Assumption \ref{assumption2} with $\beta=1$ and sparsity $s$, and let $\hat X$ be a solution to \eqref{eq:case3} with $\beta = r=1$. Then {there exists some absolute constant $C$, such that}
\begin{equation}\label{eq:errbnd2D}
\| \hat X-X\|_F\leq C\sqrt{s}\delta.
\end{equation}
\end{theorem}
Similar to Theorem \ref{cor1}, we can extend Theorem \ref{thm:2Drecon} to images that does not satisfy Assumption \ref{assumption2}. This extension follows the same idea as the one we used for Theorem \ref{cor1} so its proof is omitted.
\begin{theorem}
Let $X \in [0,1]^{N,N}$ and let $\hat X$ be a solution to \eqref{eq:case3} with $\beta= r=1$. Then for any integer $1\leq s\leq N^2$, we have
\begin{equation}
\| \hat X-X\|_F\leq C\left(\sqrt{s}\sqrt\delta+\sqrt{\sigma_s(D^TX)
}+\sqrt{\sigma_s(XD})\right)\sqrt \delta,
\end{equation}
where for any matrix $Z$, $\sigma_s(Z)$ denotes the $\ell_1$-error between the vectorized $Z$ and its best $s$-term approximation.
\end{theorem}
\begin{proof}[Proof of Theorem \ref{thm:2Drecon}]
Denote $H_1=D^T(\hat X-X)$, $H_2 = (\hat X-X)D$, and denote by $S_A$ and $S_B$ the support sets of $D^T X$ and $XD$, respectively. The corresponding complement sets are $S_A^C$ and $S_B^C$. By Assumption \ref{assumption2}, we have $|S_A|+|S_B|\leq s$. Also notice that
\begin{align*}
&\| D^T X\|_1+\| XD\|_1\\&\quad\geq \| D^T\hat X\|_1+\|\hat X D\|_1\\ 
&\quad=\| D^T X+H_1\|_1+\|X D +H_2\|_1\\
&\quad\geq \| D^T X\|_1-\| (H_1)_{S_A}\|_1+\| (H_1)_{S_A^C}\|_1+\|X D\|_1-\|  ( H_2)_{S_B}\|_1+\| ( H_2)_{S_B^C}\|_1,
\end{align*}
which gives $\| (H_1)_{S_A^C}\|_1+\| (H_2)_{S_B^C}\|_1\leq \| (H_1)_{S_A}\|_1+\| (H_2)_{S_B}\|_1$. Hence
$$\| H_1\|_1+\| H_2\|_1\leq 2(\| (H_1)_{S_A}\|_1+\| (H_2)_{S_B}\|_1)\leq 16s\delta.$$
Here the last inequality is due to 
$|S_A|+|S_B|\leq s$, and 
\begin{align*}\| H_1\|_{\textrm{max}}&=\|D^TD(D^{-1}(\hat X-X)D^{-T})D^T\|_{\textrm{max}}\leq 8\delta,\\ 
\|H_2\|_{\textrm{max}}&=\|D(D^{-1}(\hat X-X)D^{-T})D^TD\|_{\textrm{max}}\leq 8\delta.
\end{align*}
Then we have the following inequalities:
\[
\| D^T(\hat X-X)\|_1\leq 16s\delta,\;\;\| D^{-1}(\hat X-X)\|_{\textrm{max}}\leq 2\delta.
\]
Similar to the proof of Theorem \ref{thm1}, the inequalities above lead to 
$$\|\hat X-X\|_F={\langle D^T(\hat X-X),D^{-1}(\hat X-X)\rangle}^{\frac{1}{2}}\leq C\sqrt{s}\delta,$$
which is \eqref{eq:errbnd2D}.
\end{proof}
\subsection{Class 3: Reconstruction of Images Meeting the Minimum Separation Condition}\label{sec:cbcsepa}
In this section, we consider Class 3, where the image $X$ satisfies Assumption \ref{assumption3}, which is stronger than Assumption \ref{assumption1} in that jump discontinuities are not required to be separated. We hope that this extra assumption, when satisfied by an image, can lead to a reduction of the reconstruction error. Same as in Class 1, we use $Q_{col}$ (column-by-column quantization) for encoding. The reason why we do not use $Q_{2D}$ is that the 2D minimum separation condition is not realistic for natural images.

For $\underline{x}\in\mathbb{R}^N$, its Sigma Delta quantization is $\underline{q} = Q_{\mathcal{A}_\delta}^{\Sigma\Delta,r}(\underline{x})$. Again, we use $x$ and $q$ to replace $\underline{x}$ and $\underline{q}$ for simplicity.
Then \eqref{eq:case2} reduces to
\begin{align}\label{eq:case2full}
\begin{split}
&\hat x=\argmin_{z} \| D_1^\beta z\|_1\quad  \\
&\textrm{subject to}\left\{
\begin{aligned}
&\| D^{-r}(z-q)\|_\infty\leq\delta/2, \\
&\|(D^{-r}(z-q))_{N-r+1:N}\|_\infty\leq\left(\frac{1}{2N}\right)^{r}\delta. 
\end{aligned}
\right.
\end{split}
\end{align}
There are two differences between this decoder and the decoder we used for Class 1: 1) Inside the $\ell_1$-norm, this decoder uses $D_1$ (the circulant difference matrix) instead of $D$. This is to ensure that the separation condition is well-defined at the boundary; and 2) In order for the separation assumption to improve the error bound of Class 1, we need to use a few more bits to encode the boundary pixels as explained in Sect.  \ref{sec:decoders}. The total number of boundary bits is of order $O(s\log N)$, which is negligible compared to the $O(N)$ bits needed for the interior pixels. 
\begin{theorem}\label{thm2}
For the high order $\Sigma\Delta$ quantization, i.e., $r\geq 2$, assume $x\in [0,1]^N$ and $D_1^\beta x$ satisfies the $\Lambda_M$-minimization separation condition {with $M\geq 128$}. Let $\hat x$ be a solution to \eqref{eq:case2full}. Then for an arbitrary resolution $L\leq N/2$, the following error bound holds:
\begin{equation}\label{eq:ineq2}
    \| P_L (\hat x-x)\|_\infty\leq C\frac{L^2}{N^r}{M}^{r+\beta-2}\delta.
\end{equation}
\end{theorem}
Here $P_L$ is the projection onto the low-frequency range $\{-L,\cdots,L\}$, i.e., $P_L=\frac{1}{N}F_L^*F_L$ with $F_L$ {containing the rows of the DFT matrix $F$ with frequencies in $\{-L,\cdots,L\}$.}
\begin{remark}
After applying the decoder \eqref{eq:case2full} to each column, we put the reconstructed columns together to obtain the reconstructed image $\hat X$, the overall error bound in max-norm (entry-wise maximum magnitude) is then
$$\| P_L (\hat X-X)\|_{\textrm{max}}\leq C\frac{L^2}{N^r}{M}^{r+\beta-2}\delta.$$
Substituting $L$ with $N/2$, we obtain
$$\|\hat X-X\|_F\leq C\frac{M^{r+\beta-2}}{N^{r-3}}\delta.$$
\end{remark}
\begin{remark} One should only use decoder \eqref{eq:case2full} for columns satisfying the minimum separation condition. For these columns, Theorem \ref{thm2} says that the worst-case $\ell_\infty$-norm error for an arbitrary resolution approaches $0$ as the number of pixels $N\rightarrow \infty$.
\end{remark}
\subsubsection{Proof of Theorem \ref{thm2}}
In order to prove Theorem \ref{thm2}, we perform the super-resolution analysis \cite{candes2013super,candes2014towards,li2017elementary} under the Sigma Delta framework. We first need the following lemma.
\begin{lemma}\label{lemma:lowfreq}
For any feasible $\widetilde x\in\mathbb{R}^N$  satisfying the constraints in \eqref{eq:case2full}, the following inequality holds:
$$\|P_M D_1^\beta (\widetilde x-x)\|_2\lesssim \left(\frac{M}{ N}\right)^{r+\beta}\sqrt{N}\delta.$$
\end{lemma}
\begin{remark}
Substituting the $\widetilde x$ in Lemma \ref{lemma:lowfreq} by $\hat x$ (the minimizer), we can see that the low-frequency error $\|P_M D_1^\beta(\hat x-x)\|$ decreases as the order $r$ of the $\Sigma\Delta$ quantization increases.
\end{remark}
\begin{proof}
Denote $z=\widetilde x-x$. Recall that the discrete Fourier coefficient of $z$ at frequency $k$ is $\mathcal{F}_k z=\sum_{n=1}^{N} z_n e^{-i2\pi{\frac{k(n-1)}{N}}}$, where $\mathcal{F}_k$ is the row vector with $e^{-i2\pi{\frac{k(n-1)}{N}}}$ as its $n^{\textrm{th}}$ entry. Also recall that the projection operator $P_M$ in the target inequality is equivalent to the sum of the outer products of $\mathcal{F}_k,\ k=-M,\cdots,M$, divided by $N$. For {an arbitrary} nonzero frequency $k\neq 0$, denote $\alpha=\frac{1}{ 1-e^{-i2\pi{\frac{k}{N}}}}.$ Then we have 
$$\mathcal{F}_k D_1 z=\sum_{n=1}^{N}(D_1z)_ne^{-i2\pi\frac{k(n-1)}{N}}=(1-e^{-i2\pi{\frac{k}{N}}})\mathcal{F}_k z=\alpha^{-1}\mathcal{F}_k z.$$
Therefore $\mathcal{F}_kD_1 z$ is equivalent to a scaling of $\mathcal{F}_k z$. Next, we show that $\mathcal{F}_k D^{-r}z$ is also close to a scaling of $\mathcal{F}_k z$, which will yield a simple relation between $\mathcal{F}_k D_1^\beta z$ and $\mathcal{F}_k D^{-r}z$. Specifically, a direct calculation gives
\begin{align*}
\mathcal{F}_k D^{-1}z&=\sum_{n=1}^{N}(\sum_{j=1}^n z_j)e^{-i2\pi\frac{k(n-1)}{N}}=\sum_{n=1}^{N}z_n\sum_{j=n}^{N}e^{-i2\pi\frac{k(j-1)}{N}}\\
&={\frac{1}{1-e^{-i2\pi{\frac{k}{N}}}}}\sum_{n=1}^{N}z_n(e^{-i2\pi{\frac{k(n-1)}{N}}}-1)=\alpha\mathcal{F}_k z-\alpha(D^{-1}z)_{N}.
\end{align*}
Similarly, we have
$$\mathcal{F}_k D^{-2}z =\alpha\mathcal{F}_k D^{-1}z-\alpha(D^{-2}z)_{N}=\alpha^2\mathcal{F}_k z-\alpha^2 (D^{-1}z)_{N}-\alpha(D^{-2}z)_{N},$$
$$\mathcal{F}_k D^{-3}z 
=\alpha^3\mathcal{F}_k z-\alpha^3 (D^{-1}z)_{N}-\alpha^2(D^{-2}z)_{N}-\alpha(D^{-3}z)_{N}.$$
More generally, for $\beta=1\ \text{or}\ 2,\;r\geq 2$,
\begin{align*}
\mathcal{F}_k D^{-r}z&=\alpha^r \mathcal{F}_k z-\alpha^r(D^{-1}z)_{N}-\alpha^{r-1}(D^{-2}z)_{N}-\cdots-\alpha(D^{-r}z)_{N}\\
&=\alpha^{r+\beta}\mathcal{F}_k D_1^\beta z-\alpha^r(D^{-1}z)_{N}-\alpha^{r-1}(D^{-2}z)_{N}-\cdots-\alpha(D^{-r}z)_{N}.
\end{align*}
Multiplying $\alpha^{-(r+\beta)}$ on both sides and rearranging the terms give,
\begin{align*}
\mathcal{F}_k D_1^\beta z&=(1-e^{-i2\pi {\frac{k}{N}}})^{r+\beta}\mathcal{F}_k D^{-r}z+(1-e^{-i2\pi{\frac{k}{N}}})^\beta(D^{-1}z)_{N}\\
&\quad+(1-e^{-i2\pi{\frac{k}{N}}})^{\beta+1}(D^{-2}z)_{N}+\cdots+(1-e^{-i2\pi{\frac{k}{N}}})^{\beta+r-1}(D^{-r}z)_{N}.
\end{align*}
Note that for $k=0$, $\mathcal{F}_0D_1^\beta z=\sum_{n=1}^{N}(D_1^\beta z)_n=0$. Then the equation above holds for all integers $k$ with $0\leq|k|\leq N/2$.  Denote $\Lambda\in \mathbb{C}^{2M+1,2M+1}$ as the diagonal matrix with diagonal entries being $1-e^{-i2\pi{\frac{k}{N}}}, -M\leq k\leq M$, and let $F_M$ be the matrix {consisting of} $\mathcal{F}_k$ with $ k=-M,\cdots,M$ as its first to the $(2M+1)^{\textrm{th}}$ rows, we obtain the matrix form of the previous equation
\begin{align*}
&F_M D_1^\beta z\\&= \Lambda^{r+\beta}F_M D^{-r} z+\Lambda^\beta (D^{-1}z)_{N}\mathbf{1}
+\Lambda^{\beta+1} (D^{-2}z)_{N}\mathbf{1}+\cdots+\Lambda^{\beta+r-1} (D^{-r}z)_{N}\mathbf{1}\\
&=\Lambda^{r+\beta}F_M D^{-r} z+\sum_{\ell=1}^r\Lambda^{\beta+\ell-1}(D^{-\ell}z)_{N}\mathbf{1}.
\end{align*}
Here $\mathbf{1}$ represents the all-one vector.  Multiplying $\frac{1}{N}F_M^*$ to both sides of the above equation and replacing $z$ with $z=\widetilde x-x$, we have 
\begin{align*}
P_M D_1^\beta (\widetilde x-x)={\frac{1}{N}}F_M^*\Lambda^{r+\beta}F_M D^{-r} (\widetilde x-x)+\sum_{\ell=1}^r\frac{1}{N}F_M^*\Lambda^{\beta+\ell-1}(D^{-\ell}(\widetilde x-x))_{N}\mathbf{1}.
\end{align*}
Note that $ \left|1-e^{-i 2\pi{\frac{k}{N}}}\right|\leq 2\pi{\frac{|k|}{N}}\leq 2\pi {\frac{M}{N}}$. Hence $\left\|{\frac{1}{N}}F_M^*\Lambda^\ell\right\|_2\lesssim {\frac{1}{\sqrt{N}}}\left({\frac{M}{N}}\right)^\ell$, and the following error bound in $\ell_2$-norm holds
\begin{align}\label{eq:D_l}
&\| P_M D_1^\beta(\widetilde x-x)\|_2\notag \\
&\quad\leq \left\| {\frac{1}{N}}F_M^*\Lambda^{r+\beta}F_M\right\|_2\| D^{-r}(\widetilde x-x)\|_2\\
&\quad\quad+\sum_{\ell=1}^r\left\|\frac{1}{N}F_M^*\Lambda^{\beta+\ell-1}\right\|_2\sqrt{2M+1}\left|(D^{-\ell}(\widetilde x-x))_{N}\right| \notag \\
&\quad\lesssim \left(\frac{M}{N}\right)^{r+\beta}\sqrt{N}\delta+\sum_{\ell=1}^r\left(\frac{M}{N}\right)^{\beta+\ell-\frac{1}{2}}\cdot 2^r\left(\frac{1}{2N}\right)^r\delta  \notag\\
&\quad\lesssim \left({\frac{M}{N}}\right)^{r+\beta}\sqrt{N}\delta. \notag
\end{align}
Here \eqref{eq:D_l} is due to the feasibilities of $\widetilde{x}$ and the true signal $x$:  
\[\|(D^{-r}(\widetilde x-q))_{N-r+1:N}\|_\infty\leq\left(\frac{1}{2N}\right)^{r}\delta,\ \|(D^{-r}(x-q))_{N-r+1:N}\|_\infty\leq\left(\frac{1}{2N}\right)^{r}\delta,\] where $(D^{-r}(\widetilde x-q))_{N-r+1:N}$ stands for the last $r$ entries of $D^{-r}(\widetilde x-q)$. These two inequalities and the triangle inequality  imply $ \|(D^{-r}(\widetilde x-x))_{N-r+1:N}\|_\infty\leq 2\left(\frac{1}{2N}\right)^{r}\delta$. From this last inequality,  we further have for all $1\leq\ell\leq r$,  $|(D^{-\ell}(\widetilde x-x))_{N}|\leq 2^r\left(\frac{1}{2N}\right)^r\delta$, which leads to \eqref{eq:D_l}.
\end{proof}
Denote the $\beta^{\textrm{th}}$-order derivative of the reconstruction error by $h=D_1^\beta(\hat x-x)$. For convenience of notation, we index the entries of $h$ from $0$ to $N-1$, i.e., $h=[h_0,...,h_{N-1}]$. We shall first show that $h$ is small, then using it, we prove that the overall reconstructing error $P_L(\hat x-x)$ is also small, where $L$ is the highest frequency we hope to super-resolve. To show that $h$ is small, we divide it into two parts, a part that contains elements within a neighborhood of some nonzero element of $D_1^\beta x$, and a part that contains the rest of the elements.

For $D_1^\beta x$ satisfying the $\Lambda_M$-minimum separation condition, suppose its support set is  $S=\{m_1,m_2,\cdots,m_s\}\subset \{0,\cdots,N-1\}$. Then $\{\xi_1,\cdots,\xi_s\}:=\frac{1}{N}S$ can be viewed as samples on the discretized torus, $\mathcal{T}=\{ t_n = n/N, n=0,1,\cdots,N-1\}$. As in \cite{li2017elementary}, we define
$$S_M(j)=\{x\in\mathcal{T}:|x-\xi_j|\leq 0.16 M^{-1}\},\ j=1,2,\cdots,s,$$
and
$$S_M=\bigcup_{j=1}^s S_M(j),\ S_M^c=\mathcal{T}\backslash S_M.$$
Then the following lemma shows that the energies of both parts of $h$ can be controlled.
\begin{lemma}\label{lemma2.3}[Discrete version of Proposition 2.3, \cite{li2017elementary}] If $D_1^\beta x$ satisfies the $\Lambda_M$-minimum separation condition, with the $S_M$ defined above, then there exists a constant $C>0$ such that the following holds
\begin{equation}\label{eq:lemma(a)}
    \sum_{\frac{n}{N}\in S_M^c}|h_n|\leq C\sqrt{N}\| P_M h\|_2,
\end{equation}
\begin{equation}\label{eq:lemma(b)}
    \sum_{j=1}^s\sum_{\frac{n}{N}\in S_M(j)}|h_n||t_n-\xi_j|^2\leq CM^{-2}\sqrt{N}\| P_M h\|_2.
\end{equation}
\end{lemma}
\begin{proof}
\sloppy{Recall that $S=\{m_j,j=1,\cdots,s\}$ is the support set of $D_1^{\beta}x$. Write $h_{m_j}=|h_{m_j}|e^{i\phi_j}$ with some $\phi_j$, for $ j=1,2,\cdots,s.$}  Invoking Lemma \ref{lemma1} in Section \ref{sec:append1} of the appendix, and taking $v_j = e^{i\phi_j},\ j=1,2\cdots, s$, there exist $f(t)=\sum_{k=-M}^M c_k e^{i2\pi kt}$ defined on $\mathbb{T}$ and constants $C_1,C_2$ such that
\begin{equation}
    f(t_j)=e^{i\phi_j},\;j=1,2\cdots, s, 
\end{equation}
\begin{equation}\label{eq:C1}
    |f(t)|\leq 1-C_1 M^2(t-\xi_j)^2, \;t\in S_M(j),
\end{equation}
\begin{equation}\label{eq:C2}
    |f(t)|<1-C_2,\; t\in S_M^c.
\end{equation}
Denote $f_n=f(t_n)$ for $t_n=\frac{n}{N},\ n=0,1,\cdots,N-1$. Then we have
\begin{align*}
    &\sum_{n\in S} |h_n|\\
    &\quad=\left|\sum_{n\in S}\bar f_n h_n\right|\\
    &\quad\leq \left|\sum_{n=0}^{N-1} \bar f_nh_n\right|+\left|\sum_{\frac{n}{N}\in S_M^c}\bar f_nh_n\right|+\left|\sum_{j=1}^s\sum_{\frac{n}{N}\in S_M(j)\backslash\{\xi_j\}}\bar f_nh_n\right|\\
    &\quad\leq \left|\sum_{n=0}^{N-1} \bar f_nh_n\right|+(1-C_2)\sum_{\frac{n}{N}\in S_M^c}|h_n|+\sum_{j=1}^s \sum_{\frac{n}{N}\in S_M(j)\backslash\{\xi_j\}}(1-C_1M^2(t_n-\xi_j)^2)|h_n|\\
    &\quad=\left|\sum_{n=0}^{N-1} \bar f_nh_n\right|+\sum_{n\in S^c}|h_n|-C_2\sum_{\frac{n}{N}\in S_M^c}|h_n|-C_1M^2\sum_{j=1}^s\sum_{\frac{n}{N}\in S_M(j)}(t_n-\xi_j)^2|h_n|,
\end{align*}
where the last inequality used \eqref{eq:C1} and \eqref{eq:C2}.
Rearranging the inequality, we obtain
\begin{equation}\label{eq:lemmaineq}
    C_2\sum_{\frac{n}{N}\in S_M^c}|h_n|+C_1M^2\sum_{j=1}^s\sum_{\frac{n}{N}\in S_M(j)}(t_n-\xi_j)^2|h_n|\leq \left|\sum_{n=0}^{N-1} \bar f_nh_n\right|+\sum_{n\in S^c}|h_n|-\sum_{n\in S} |h_n|.
\end{equation}
Using ${\bf f}$ to represent the $N$-dimensional vector $\{f_n\}_{n=0}^{N-1}$, we notice that $\left|\sum_{n=0}^{N-1}\bar f_nh_n\right|=|\langle {\bf f},h\rangle|=|\langle {\bf f},P_Mh\rangle|\leq \|{\bf f}\| _2 \| P_M h\| _2\leq \sqrt{N} \| P_Mh\| _2$. Also note that $\hat x$ is a solution to \eqref{eq:case2full}, so it holds that 
$$ \| D_1^\beta x\|_1\geq \| D_1^\beta \hat x\| _1= \| D_1^\beta x+h\| _1\geq \sum_{n\in S}|(D_1^\beta x)_n|-\sum_{n\in S}|h_n|+\sum_{n\in S^c}|h_n|.$$
Rearranging the above inequality, we have
$$\sum_{n\in S^c}|h_n|-\sum_{n\in S}|h_n|\leq 0.$$
Then \eqref{eq:lemmaineq} becomes 
$$ C_2\sum_{\frac{n}{N}\in S_M^c}|h_n|+C_1M^2\sum_{j=1}^s\sum_{\frac{n}{N}\in S_M(j)}(t_n-\xi_j)^2|h_n|\leq\sqrt{N}\| P_M h\|_2.$$
From this inequality we can derive \eqref{eq:lemma(a)} and \eqref{eq:lemma(b)}.
\end{proof}
Lemma \ref{lemma:lowfreq} and Lemma \ref{lemma2.3} together ensure that $h= D_1^{\beta}(\hat{x}-x)$ is small. In the following, we show that this small $h$ further implies a small $P_L(\hat{x}-x)$, which completes the proof of Theorem \ref{thm2}.
\begin{proof}[\textbf{Proof of Theorem \ref{thm2}}]
\sloppy To start with, we consider a kernel $K(t) \in C^{\infty}(\mathbb{T})$ that is an arbitrary  $C^{\infty}$ function on the 1D torus $\mathbb{T}$. Define the discretization of $K(t)$  by the boldface letter $\mathbf{K}$, that is, $\mathbf{K}$ is the $N$-dimensional vector with $\mathbf{K}=[K(t_0),K(t_1),\cdots,K(t_{N-1})]^T\in\mathbb{R}^N$ that contains samples of $K(t)$ at the grid points $t_n=\frac{n}{N},\  n=0,1,\cdots N-1$. In what follows, the normal font $K(t)$ always refers to the continuous kernel and the boldface ${\bf K}$ refers to the discretization. We need to frequently take their infinity norms, denoted as $\|K(t)\|_{L^{\infty}}$ and $\|{\bf K}\|_{\infty}$, respectively, and by definition, we have $\|{\bf K}\|_{\infty} \leq \|K(t)\|_{L^{\infty}}$.
 
The proof contains two steps. In the first step, we bound $\| {\bf K}*D_1^\beta(\hat x-x)\|_\infty$ for any general ${\bf K}$. In the second step, we pick a special ${\bf K}$ to obtain the desired result.

For an arbitrary $x_0\in\left\{0,\frac{1}{N},\cdots,\frac{N-1}{N}\right\}$, by the definition of $\mathbf{K}$ and the fact that  $\mathbf{K}$ and $h$ are  periodic, it holds that
\[
\mathbf{K}\ast h(x_0)=\sum_{n=0}^{N-1}K(x_0-t_n)h_n.
\]
Hence,
\begin{align}\label{eq:conv}
\begin{split}
    |\mathbf{K}*h(x_0)|&=\left|\sum_{n=0}^{N-1}K(x_0-t_n)h_n\right|
    \leq \left|\sum_{j=1}^s\sum_{\frac{n}{N}\in S_M(j)}K(x_0-t_n)h_n\right|+\| K\|_{L^\infty}\sum_{\frac{n}{N}\in S_M^c}|h_n|.
    \end{split}
\end{align}
On the interval $S_M(j)$, we approximate $K(x_0-t_n)$ with its first-order Taylor expansion around $x_0-\xi_j$,
$$K(x_0-t_n)=K(x_0-\xi_j)+K'(x_0-\xi_j)(\xi_j-t_n)+{\frac{1}{2}}K''(\mu_n)|t_n-\xi_j|^2,\;t_n\in S_M(j),$$
where {$\mu_n$ is some value between $x_0-t_n$ and $x_0-\xi_j$.} Inserting this into \eqref{eq:conv}, we obtain
\begin{align*}
|\mathbf{K}*h(x_0)|&\leq \left|\sum_{j=1}^s\sum_{\frac{n}{N}\in S_M(j)}(K(x_0-\xi_j)-K'(x_0-\xi_j)(t_n-\xi_j))h_n\right|\\
&\;\;\;\;+\frac{1}{2}\| K''\|_{{L^\infty}}\sum_{j=1}^s\sum_{\frac{n}{N}\in S_M(j)}|t_n-\xi_j|^2|h_n|+\| K\|_{{L^\infty}}\sum_{\frac{n}{N}\in S_M^c}|h_n|.  
\end{align*}
To bound the first term on the right hand side, we use an interpolation argument. Let $a,b\in \mathbb{C}^s$ such that $a_j=K(x_0-\xi_j),\;b_j=-K'(x_0-\xi_j)$ and by Proposition 2.4 in \cite{li2017elementary}, {when $M\geq 128$,} there exists a function $\widetilde f\in C(\mathbb{T};\Lambda_M)$ such that
\begin{equation}\label{eq:lemma_interpoation1}
\| \widetilde f\|_{{L^\infty}}\lesssim \| K\|_{{L^\infty}}+M^{-1}\| K'\|_{{L^\infty}},
\end{equation}
\begin{equation}\label{eq:lemma_interpoation2}
|\widetilde f(x)-a_j-b_j(x-\xi_j)|\lesssim(M^2\| K\|_{{L^\infty}}+M\| K'\|_{{L^\infty}})|x-\xi_j|^2,\;\forall x\in S_M(j),
\end{equation}
which gives
\begin{align*}
    \Big|\sum_{j=1}^s&\sum_{\frac{n}{N}\in S_M(j)}(K(x_0-\xi_j)-K'(x_0-\xi_j)(t_n-\xi_j))h_n\Big| \\
    &\leq \left|\sum_{j=1}^s\sum_{\frac{n}{N}\in S_M(j)}(\widetilde f_n-K(x_0-\xi_j)+K'(x_0-\xi_j)(t_n-\xi_j))h_n\right|+\left|\sum_{\frac{n}{N}\in S_M}\widetilde f_nh_n\right|\\
    &\lesssim (M^2\| K\|_{{L^\infty}}+M\| K'\|_{{L^\infty}})\sum_{j=1}^s\sum_{\frac{n}{N}\in S_M(j)}|t_n-\xi_j|^2|h_n|+\left|\sum_{n=0}^{N-1}\widetilde f_nh_n\right|+\left|\sum_{{\frac{n}{N}}\in S_M^c}\widetilde f_nh_n\right|.
\end{align*}
Here $\widetilde f_n=\widetilde f(t_n)$ with $t_n=\frac{n}{N},\ n=0,1,\cdots,N-1$. Below, we will use ${\bf \widetilde f}$ to represent the $N$-dimensional vector $\{\widetilde f_n\}_{n=0}^{N-1}$. From \eqref{eq:lemma_interpoation1} and \eqref{eq:lemma_interpoation2}, we obtain
$$\left|\sum_{\frac{n}{N}\in S_M^c}\widetilde f_nh_n\right|\lesssim (\| K\|_{{L^\infty}}+M^{-1}\| K'\|_{{L^\infty}})\sum_{\frac{n}{N}\in S_M^c}|h_n|,$$
$$\left|\sum_{n=0}^{N-1}\widetilde f_nh_n\right|\leq \| {\bf \widetilde f}\|_2\| P_Mh\|_2\lesssim (\| K\|_{{L^\infty}}+M^{-1}\| K'\|_{{L^\infty}})\sqrt{N}\| P_Mh\|_2.$$
Combining these results, we have
\begin{align}\label{eq:errbound}
\begin{split}
    |\mathbf{K}*h(x_0)|&\lesssim (2\| K\|_{{L^\infty}}+M^{-1}\| K'\|_{{L^\infty}})\sum_{\frac{n}{N}\in S_M^c}|h_n|\\
    &+(\| K\|_{{L^\infty}}+M^{-1}\| K'\|_{{L^\infty}})\sqrt{N}\| P_Mh\|_2\\
    &+(M^2\| K\|_{{L^\infty}}+M\| K'\|_{{L^\infty}}+\| K''\|_{{L^\infty}})\sum_{j=1}^s\sum_{\frac{n}{N}\in S_M(j)}|t_n-\xi_j|^2|h_n|\\
    &\lesssim (\| K\|_{{L^\infty}}+M^{-1}\| K'\|_{{L^\infty}}+M^{-2}\| K''\|_{{L^\infty}})\sqrt{N}\| P_Mh\|_2,
    \end{split}
\end{align}
where the last inequality used Lemma  \ref{lemma2.3} and $\lesssim$ hides a constant. 
Next, we plug in a special kernel to $K$ to prove the theorem. For an arbitrary resolution $L$, set $K_L(t)={\frac{1}{N}}\sum_{k=-L,\;k\neq 0}^L e^{i2 \pi kt}$, and denote by $\mathbf{K_L}$ the discretized vector that contains samples of $K_L$ at $t_n=\frac{n}{N}, n=0,1,\dots,N-1$. Then a direct calculation gives
$$P_L(\hat x-x)=\mathbf{K_L}*(\hat x-x)+\frac{1}{N}\sum_{n=1}^{N}(\hat x-x)_n\mathbf{1}.$$
We will bound the two terms on the right hand side separately. The second term is easy to bound due to the boundary constraint in \eqref{eq:case2full} which forces the absolute value of the last $r$ rows in $D^{-r}(\hat x-x)$ to be smaller than $2\cdot(\frac{1}{2N})^{r}\delta$. Then we have  $$\left|\sum_{n=1}^{N}(\hat x-x)_n\right|=|(D^{-1}(\hat x-x))_{N}| \leq 2^{r-1}\|D^{-r}(\hat x-x)_{N-r+1:N}\|_\infty\leq \frac{1}{N^r}\delta.$$
This gives $\|P_L(\hat x-x)-\mathbf{K_L}*(\hat x-x)\|_\infty\leq \frac{1}{N^{r+1}}\delta$. Next, we derive an upper bound on the first term $\|{\bf K_L}*(\hat x-x)\|_\infty$. 

\noindent Case \uppercase\expandafter{\romannumeral 1}: TV order $\beta=1$:

We construct an auxiliary function 
\[\widetilde K_L(t)=\frac{1}{N}\sum_{k=-L,\;k\neq 0}^L \frac{1-e^{i2\pi k(t+\frac{1}{N})}}{1-e^{i2\pi \frac{k}{N}}},\ t\in[0,1].
\]
Notice that
\begin{align*}
\widetilde K_L(t_j) &= \frac{1}{N}\sum_{k=-L,\;k\neq 0}^L \frac{1-e^{i2\pi k\frac{j+1}{N}}}{1-e^{i2\pi  \frac{k}{N}}}\\=&\frac{1}{N}\sum_{k=-L,\;k\neq 0}^L\sum_{n=0}^{j}e^{i2\pi k\frac{n}{N}}=\sum_{n=0}^{j} K_L(t_n), \;\;\;j = 0,1,\cdots,N-1.
\end{align*}
Hence $\widetilde K_L$ satisfies the property
\[
D_1\widetilde K_L(t_{j})=\widetilde K_L(t_j)-\widetilde K_L(t_{j-1})=K_L(t_j),\ \ j=1,2,\cdots,N-1.
\]
Also, since $\widetilde K_L(t_{N-1})=0$, we have $D_1\widetilde K_L(t_{j})=K_L(t_j),\ \forall j=0,1,\cdots,N-1.$ Then the bound $\|{\bf K_L}\ast (\hat x-x)\|_\infty$ is equivalent to 
$$\|{\bf K_L}\ast (\hat x-x)\|_\infty=\|D_1 {\bf \widetilde K_L}\ast (\hat x-x)\|_\infty=\| {\bf \widetilde K_L}\ast D_1(\hat x-x)\|_\infty=\| {\bf \widetilde K_L}\ast h\|_\infty,$$
where ${\bf \widetilde K_L} = \{\widetilde K_L(t_j)\}_{j=0}^{N-1}$. Now we show that the infinity norm of $\widetilde K_L(t)$ is bounded by some constant for arbitrary $L\leq N/2$ and $t\in [0,1]$, then we can bound $\|{\bf \widetilde K_L}\ast h\|_{\infty}$ by \eqref{eq:errbound}. Since for $k\in \mathbb{Z}$, $e^{i2\pi k t}$ is $1$-periodic, we have
\begin{align*}
    &\sup_t|\widetilde K_L(t)|\\&\quad=\frac{1}{N}\sup_t\left|\sum_{k=-L,\;k\neq 0}^L \frac{1-e^{i2\pi kt}}{1-e^{i2\pi \frac{k}{N}}}\right|\\
    &\quad=\frac{1}{N}\sup_t\left|\sum_{k=1}^L\frac{1-\cos(2\pi kt)-i\sin(2\pi kt)}{1-\cos(2\pi \frac{k}{N})-i\sin(2\pi\frac{k}{N})}+\frac{1-\cos(2\pi kt)+i\sin(2\pi kt)}{1-\cos(2\pi \frac{k}{N})+i\sin(2\pi\frac{k}{N})}\right|\\
    &\quad=\frac{1}{N}\sup_t\left|\sum_{k=1}^L \frac{(1-\cos(2\pi kt))(1-\cos(2\pi\frac{k}{N}))+\sin(2\pi kt)\sin(2\pi\frac{k}{N})}{1-\cos(2\pi\frac{k}{N})}\right|\\
    &\quad=\frac{1}{N}\sup_t\left|\sum_{k=1}^L (1-\cos(2\pi kt))+\frac{\sin(2\pi kt)\cos(\pi\frac{k}{N})}{\sin(\pi\frac{k}{N})}\right|\\
    &\quad\leq\frac{L}{N}+\frac{1}{N}\sup_t\left|\sum_{k=1}^L\frac{\sin(2\pi kt)\cos(\pi\frac{k}{N})-\sin(\pi\frac{k}{N})\cos(2\pi kt)}{\sin(\pi\frac{k}{N})}\right|\\
    &\quad=\frac{L}{N}+\frac{1}{N}\sup_t\left|\sum_{k=1}^L\frac{\sin(2\pi k(t-\frac{1}{2N}))}{\sin(\pi\frac{k}{N})}\right|\\
    &\quad=\frac{L}{N}+\frac{1}{N}\sup_t\left|\sum_{k=1}^L\frac{\sin(2\pi kt)}{\sin(\pi\frac{k}{N})}\right|.
\end{align*}
Notice that for $k\leq L\leq N/2$, we have $\pi \frac{k}{N}\leq \frac{\pi}{2}$, and then $\sin(\pi\frac{k}{N})$ is of the same order as $\pi\frac{k}{N}$. Also for all $0<t\leq \frac{\pi}{2}$, we have $t-\frac{t^3}{6}\leq \sin(t)<t$, which implies that $0.58\pi\frac{k}{N}\leq\pi\frac{k}{N}-(\pi\frac{k}{N})^3/6\leq\sin(\pi\frac{k}{N})<\pi\frac{k}{N}$. Then we can see that 
\begin{align*}
    \left|\frac{1}{N}\sum_{k=1}^L\frac{\sin(2\pi kt)}{\sin(\pi\frac{k}{N})}-\sum_{k=1}^L\frac{\sin(2\pi kt)}{\pi k}\right|&=\frac{1}{N}\left|\sum_{k=1}^L \sin(2\pi kt)\left(\frac{1}{\sin(\pi\frac{k}{N})}-\frac{1}{\pi\frac{k}{N}}\right)\right|\\
    &=\frac{1}{N}\left|\sum_{k=1}^L \sin(2\pi kt)\frac{\pi\frac{k}{N}-\sin\left(\pi\frac{k}{N}\right)}{\sin\left(\pi\frac{k}{N}\right)\pi\frac{k}{N}}\right|\\
    &\leq \frac{1}{N}\sum_{k=1}^L \frac{\frac{1}{6}\left(\pi\frac{k}{N}\right)^3}{0.58\left(\pi\frac{k}{N}\right)^2}\\
    &< 0.23.
\end{align*}
This is saying that $\left|\frac{1}{N}\sum_{k=1}^L\frac{\sin(2\pi kt)}{\sin(\pi\frac{k}{N})}\right|$ is close to 
$\left|\sum_{k=1}^L \frac{\sin(2\pi kt)}{\pi k}\right|$. 
It is known that the latter is uniformly bounded by some constant smaller than $2/\pi$ (\cite{alzer2004sharp}) for  arbitrary $L\in \mathbb{N}$ and $t\in\mathbb{R}$, hence the former is also bounded. Therefore there exists some constant $C$ such that $\|\widetilde K_L\|_{L^\infty}\leq C$. Then we use Bernstein's inequality for trigonometric sums \cite{bernstein1912ordre} to obtain $ \| \widetilde K_L'\|_{L^\infty}\leq CL,\ \| \widetilde K_L''\|_{L^\infty}\leq CL^2$.
Thus by \eqref{eq:errbound} we have,
\begin{align*}
\|{\bf K_L}*(\hat x-x)\|_\infty&\leq\| D_1{\bf \widetilde K_L}*(\hat x-x)\|_\infty\\&=\|{\bf \widetilde K_L}*h\|_\infty\leq C{\frac{L^2}{M^2}}\sqrt{N}\cdot{\frac{M^{r+1}}{N^{r+\frac{1} {2}}}}\delta=C \frac{L^2}{N}\left(\frac{M}{N}\right)^{r-1}\delta.
\end{align*}

\noindent Case \uppercase\expandafter{\romannumeral 2}: TV order $\beta=2$:

Consider $\widetilde K_L(t)=-\frac{1}{N}\sum_{k=-L,k\neq 0}^L \frac{e^{i2\pi\frac{k}{N}}-e^{i2\pi k(\frac{2}{N}+t)}}{(1-e^{i2\pi\frac{k}{N}})^2}$. Similar to Case \uppercase\expandafter{\romannumeral 1},  we can show $D_1^2\widetilde K_L(t_{j})=K_L(t_j)$, for all $j=0,1,\cdots,N-1$, and $\|\widetilde K_L\|_{L^\infty}\leq N$. Then
\begin{align*}
\| {\bf K_L}\ast(\hat x-x)\|_\infty&\leq\| D_1^2{\bf \widetilde K_L}\ast(\hat x-x)\|_\infty\\&=\| {\bf \widetilde K_L}\ast h\|_\infty\leq C{\frac{L^2}{M^2}}N^{\frac{3}{2}}\cdot{\frac{M^{r+2}}{N^{r+\frac{3}{2}}}}\delta=CL^2\left(\frac{M}{N}\right)^{r}\delta.
\end{align*}
In conclusion, for $\beta=1 \ \text{or} \ 2$, we have the $\ell_\infty$-norm error bound
$$\|{\bf K_L}\ast(\hat x-x)\|_\infty\leq C\frac{L^2}{N^r}M^{r+\beta-2}\delta,$$
which further gives 
$$\|P_L(\hat x-x)\|_\infty\leq \|{\bf K_L}\ast(\hat x-x)\|_\infty+\left\|\frac{1}{N}\sum_{n=0}^{N-1}(\hat x-x)_n\mathbf{1}\right\|_\infty\leq C\frac{L^2}{N^r}M^{r+\beta-2}\delta.$$
\end{proof}
\section{Numerical simulation}\label{sec:experiments}
In this section, we present the numerical simulation results of the proposed schemes on 1D synthetic signals as well as 2D natural and medical images. The specific algorithms that we have used to obtain these results will be discussed in Section \ref{sec:optimization} of the appendix. 

We first define several terms. An RGB image is a composite of three gray-scale images, corresponding to the red, green, and blue channels. To evaluate the quality of the 1D  reconstruction, we use the SNR (Signal-to-Noise Ratio) defined as $$\text{SNR}_{\text{dB}}(x,y)=20\log_{10}\left(\frac{\|x\|_2}{\|x-y\|_2}\right),$$ where $x\in\mathbb{R}^N$ is the true signal, and $y\in\mathbb{R}^N$ is the reconstructed one. For 2D images, we adopt the common evaluation metric PSNR (Peak Signal-to-Noise Ratio),  $$\text{PSNR}_{\text{dB}}(X,Y)=20\log_{10}\left(\frac{MAX_I}{\sqrt{MSE(X,Y)}}\right),$$ where $X$ and $Y$ represent the clean and the reconstructed images, respectively, $MSE(X,Y)$ is the mean squared error between $X$ and $Y$, i.e., $MSE(X,Y) = \frac{1}{N^2}\sum_{i,j=1}^N (X_{i,j}-Y_{i,j})^2 $,
and $MAX_I$ is the maximum possible pixel value of the image. For grayscale images, $MAX_I=1$, and for RGB images,  $MAX_I=255$. In the following experiments, the unit for SNR and PSNR is decibel (dB).
\subsection{1D synthetic signal}
We design an experiment to confirm the advantage of the proposed methods over MSQ on 1D signals (representing columns of images) proved in Theorem \ref{thm1}. The signal to be quantized is piece-wise constant or piece-wise linear with random heights and slopes. Such signals satisfy our Assumption \ref{assumption1}. We compare MSQ with Sigma Delta quantization coupled with the decoder \eqref{eq:cbccase1}. The result is displayed in Figure \ref{fig:1D_noseparation}. We see from \ref{fig:exp1a} that, with the same number of bits, the reconstructed signal using the $1^{\text{st}}$-order $\Sigma\Delta$ quantization and decoder \eqref{eq:cbccase1} is closer to the true signal than MSQ and better preserves the piece-wise constant structure. \ref{fig:exp1b} shows a similar result for piece-wise linear signals. As real signals may be subject to noise, in \ref{fig:exp1c}, the reconstruction result under random Gaussian noise is shown.  We see that the proposed encoder-decoder pair is pretty robust to noise.
\begin{figure}[htbp]
     \centering
     \begin{subfigure}[t]{0.32\textwidth}
         \centering
         \includegraphics[width=\textwidth]{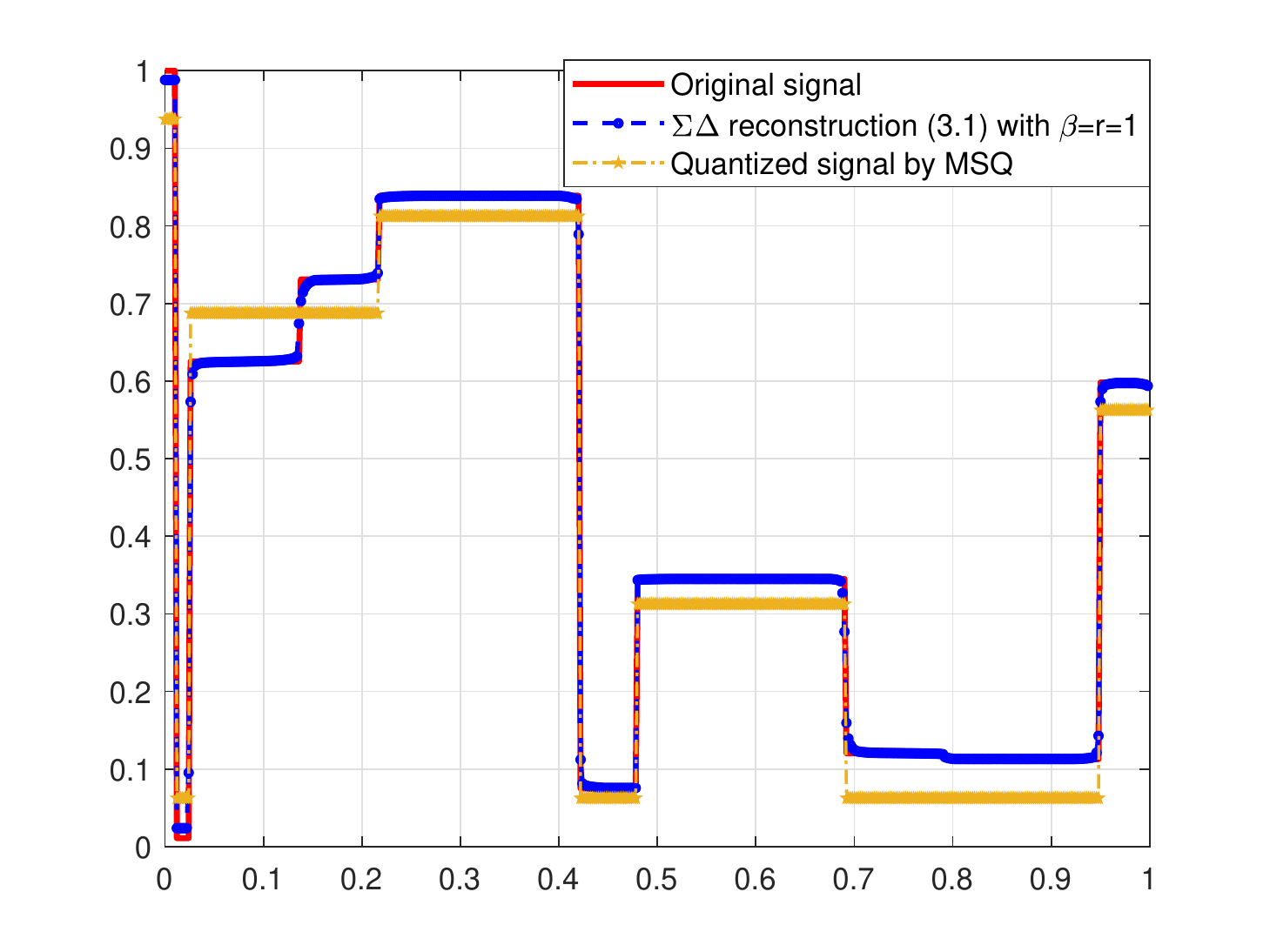}
         \caption{piece-wise constant signal}
         \label{fig:exp1a}
     \end{subfigure}
     \hfill
     \begin{subfigure}[t]{0.32\textwidth}
         \centering
         \includegraphics[width=\textwidth]{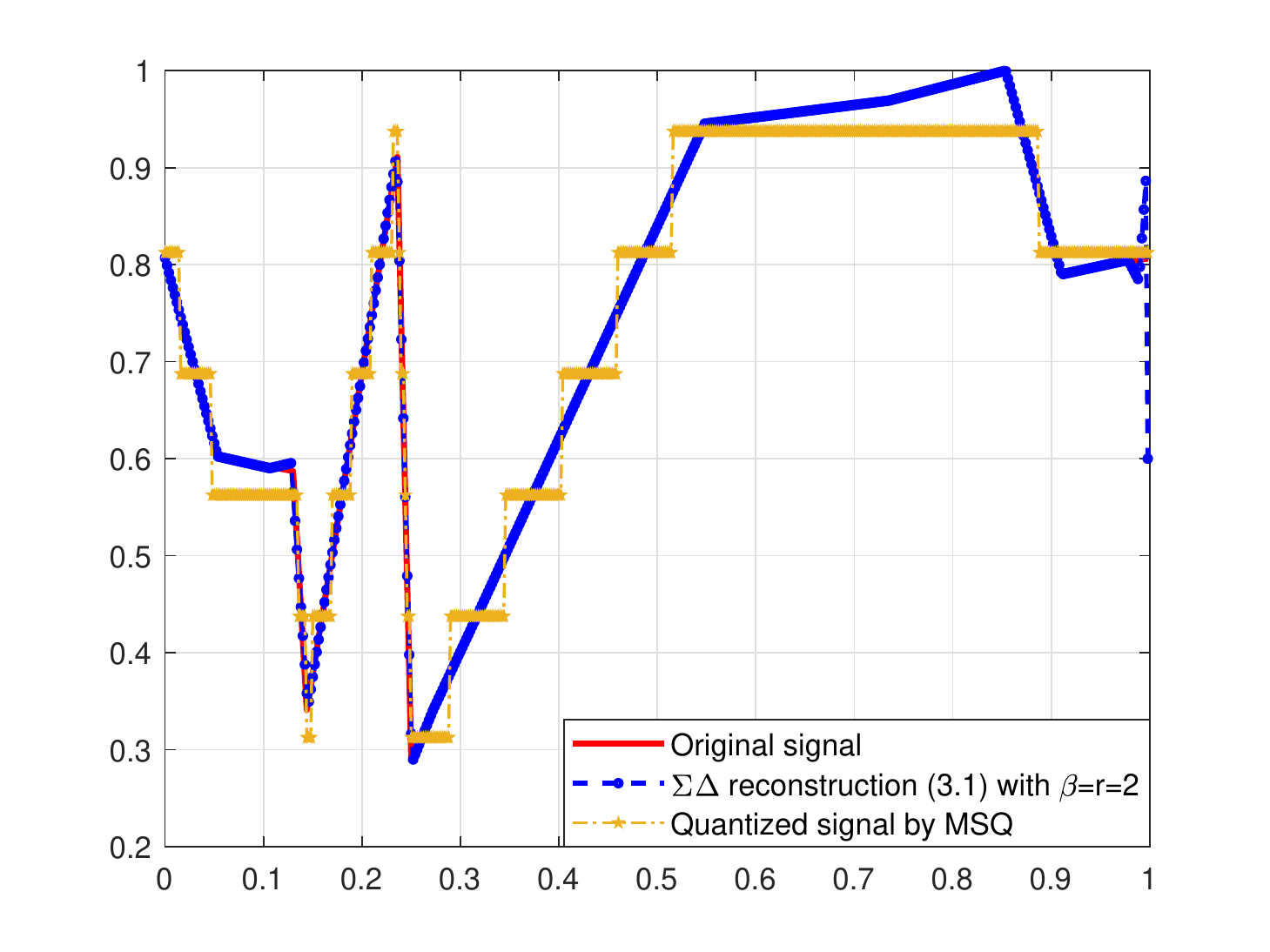}
         \caption{piece-wise linear signal}
         \label{fig:exp1b}
     \end{subfigure}
     \hfill
     \begin{subfigure}[t]{0.32\textwidth}
         \centering
         \includegraphics[width=\textwidth]{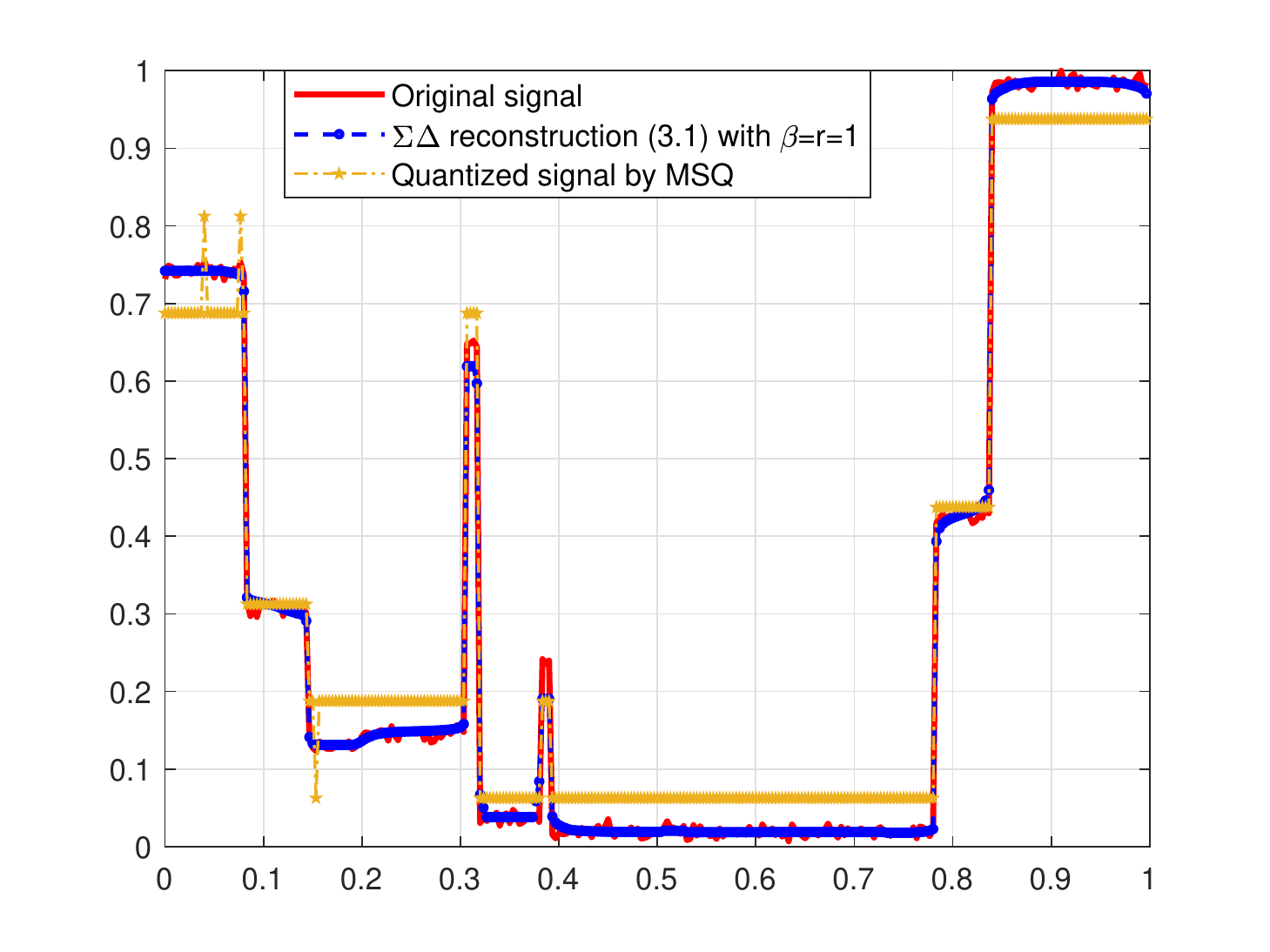}
         \caption{piece-wise constant signal contaminated by random noise}
         \label{fig:exp1c}
     \end{subfigure}
        \caption{A comparison of various 1D signal reconstruction results between the proposed encoder-decoder pairs and MSQ. Here we used a 3-bit alphabet. (A) The signal is piece-wise constant with a small edge separation. We compare MSQ with the $1^{\text{st}}$-order $\Sigma\Delta$ quantization coupled with the decoder  \eqref{eq:cbccase1} with $\beta=1$. The SNR of the reconstructed signal from $\Sigma\Delta$ quantization is 37.00 dB, and that of MSQ is 21.95 dB. (B) The signal is piece-wise linear, we used the $2^{\text{nd}}$-order $\Sigma\Delta$ quantization and the decoder \eqref{eq:cbccase1} with $\beta=2$, the SNR of the $\Sigma\Delta$ reconstruction is 37.30 dB, and the SNR of MSQ is 27.33 dB. (C) The signal is piece-wise constant with random Gaussian noise, we used the $1^{\text{st}}$-order $\Sigma\Delta$ quantization and decoder \eqref{eq:cbccase1} with $\beta=1$. The SNR of $\Sigma\Delta$ reconstruction is 33.40 dB, and that of MSQ is 20.93 dB.}
        \label{fig:1D_noseparation}
\end{figure}

Our analysis in Section \ref{sec:cbcsepa} predicted that, if the minimum separation condition is met, the quantization error can be further reduced through using a higher order $\Sigma\Delta$ quantization. This is confirmed in both (A) and (B) of Figure \ref{fig:1D_separation}, where we see that the reconstructed signal from the $3^{\textrm{rd}}$-order $\Sigma\Delta$ quantization is indeed closer to the true signal than those from the first and second order quantization.

\begin{figure}[htbp]
     \centering
     \begin{subfigure}[t]{0.49\textwidth}
         \centering
         \includegraphics[width=\textwidth]{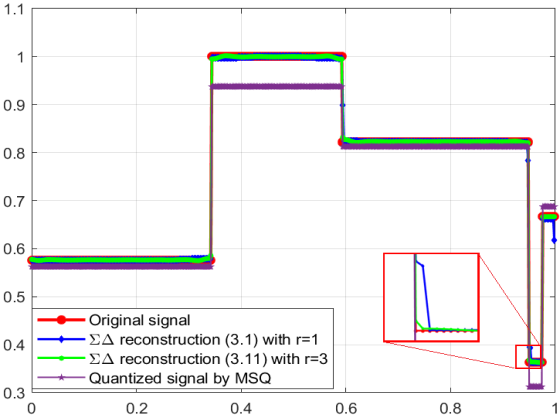}
         \caption{piece-wise constant signal}
     \end{subfigure}
     \hfill
     \begin{subfigure}[t]{0.49\textwidth}
         \centering
         \includegraphics[width=\textwidth]{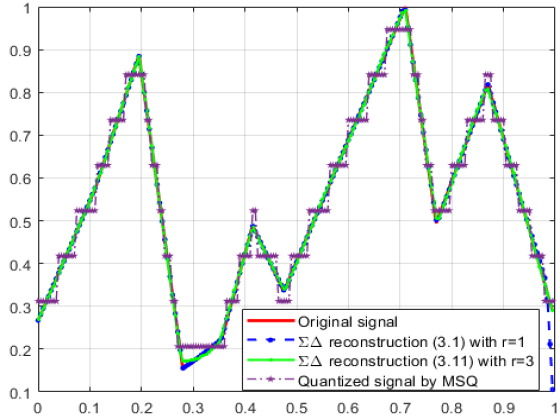}
         \caption{piece-wise linear signal}
     \end{subfigure}
        \caption{Reconstruction result of signals that satisfy the minimum separation condition. Here we used a 3-bit alphabet. (A) The signal is piece-wise constant.  We used the $1^{\text{st}}$-order $\Sigma\Delta$ quantization coupled with the decoder \eqref{eq:cbccase1} with $\beta=r=1$, as well as the $3^{\text{rd}}$-order $\Sigma\Delta$ quantization coupled with the decoder \eqref{eq:case2full} with $\beta=1,\ r=3$. The SNRs of the reconstructed signals are 43.90 dB and 56.36 dB, respectively. The SNR of MSQ is 27.28 dB. (B) The signal is piece-wise linear,  we used $2^{\text{nd}}$-order $\Sigma\Delta$ quantization coupled with the decoder \eqref{eq:cbccase1} with $\beta=r=2$, as well as $3^{\text{rd}}$-order $\Sigma\Delta$ quantization and decoder \eqref{eq:case2full} with $\beta=2,\ r=3$. The SNRs of the reconstructed signals from $2^{\text{nd}}$ and $3^{\text{rd}}$-order $\Sigma\Delta$ quantization are 31.95 dB and 48.83 dB, respectively. The SNR of MSQ is 25.70 dB.} 
        \label{fig:1D_separation}
\end{figure}
\subsection{2D natural and {medical images}}
In this section, we present numerical results for 2D images. We observe that the best performances are usually achieved when the TV order and the $\Sigma\Delta$ quantization order are both set to 1, perhaps because the test images fit our Class 2 model with $\beta=1$ better.

In the first example, on gray-scale images, we compare the 2D Sigma Delta quantization $Q_{2D}$ coupled with decoder \eqref{eq:case3} (sd2D), the 1D Sigma Delta quantization $Q_{col}$ coupled with the decoder \eqref{eq:case1} (sd1D) and the MSQ quantization. For each quantizer, we employ its own optimal alphabet but require them to be subject to the same bit budget: 3 bits per pixel. In Figure \ref{fig:cameraman}, we see that in terms of visual quality (or the amount of artifact), (sd2D) is better than (sd1D) and much better than MSQ. In terms of the PSNR, (sd1D) is slightly better than (sd2D) and much better than MSQ.  
\begin{figure}[htbp]
     \centering
     \begin{subfigure}[t]{0.24\textwidth}
         \centering
         \includegraphics[width=1\textwidth]{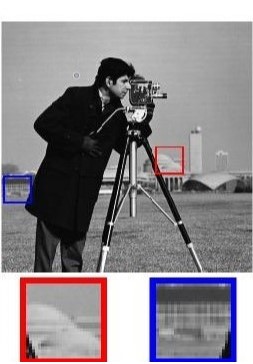}
         \caption{Original image}
         \label{Orig}
     \end{subfigure}
     \hfill
     \begin{subfigure}[t]{0.24\textwidth}
         \centering
         \includegraphics[width=\textwidth]{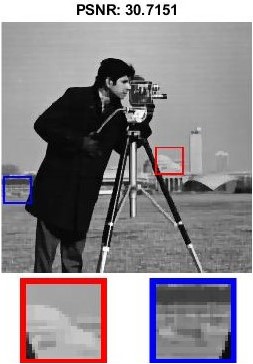}
         \caption{sd2D}
     \end{subfigure}
     \hfill
     \begin{subfigure}[t]{0.24\textwidth}
         \centering
         \includegraphics[width=1\textwidth]{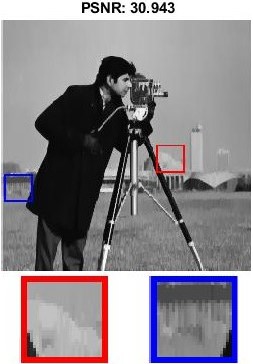}
         \caption{sd1D}
     \end{subfigure}
     \hfill
     \begin{subfigure}[t]{0.24\textwidth}
         \centering
         \includegraphics[width=\textwidth]{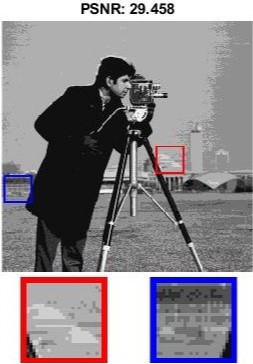}
         \caption{MSQ}
     \end{subfigure}
            \caption{Cameraman reconstruction results from 3-bit quantization. (A) The ground truth. (B) The image is quantized by the 2D $\Sigma\Delta$, $Q_{2D}$, and decoded by \eqref{eq:case3}. (C) Both quantization and reconstruction are carried out column by column. For each column, we used the $1^{\text{st}}$-order $\Sigma\Delta$ quantization and the decoder \eqref{eq:case1}. (D) MSQ quantization.}
        \label{fig:cameraman}
\end{figure}
\begin{figure}[htbp]
     \centering
     \begin{subfigure}[t]{0.24\textwidth}
         \centering
         \includegraphics[width=\textwidth]{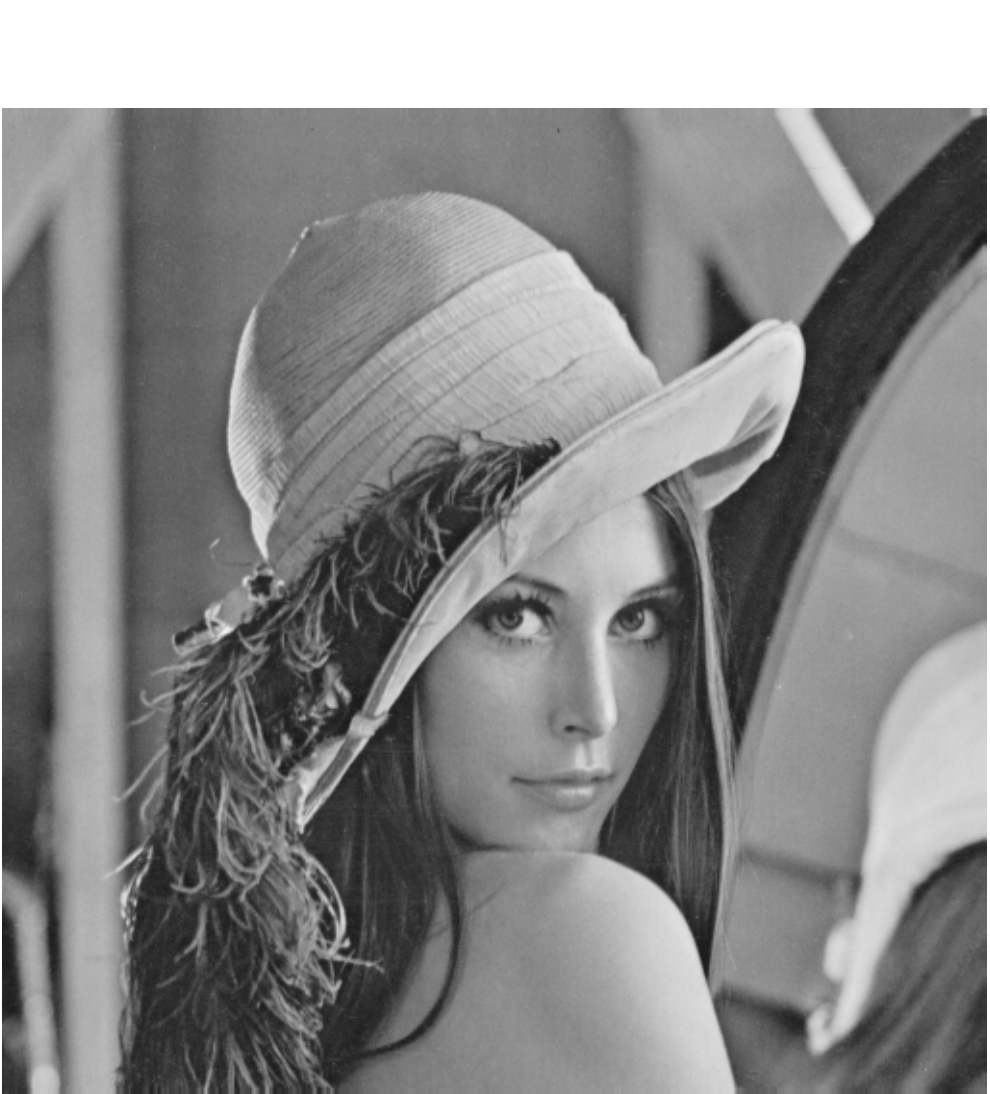}
         \caption{Original image}
         \label{fig:lena_orig}
     \end{subfigure}
     \hfill
     \begin{subfigure}[t]{0.24\textwidth}
         \centering
         \includegraphics[width=\textwidth]{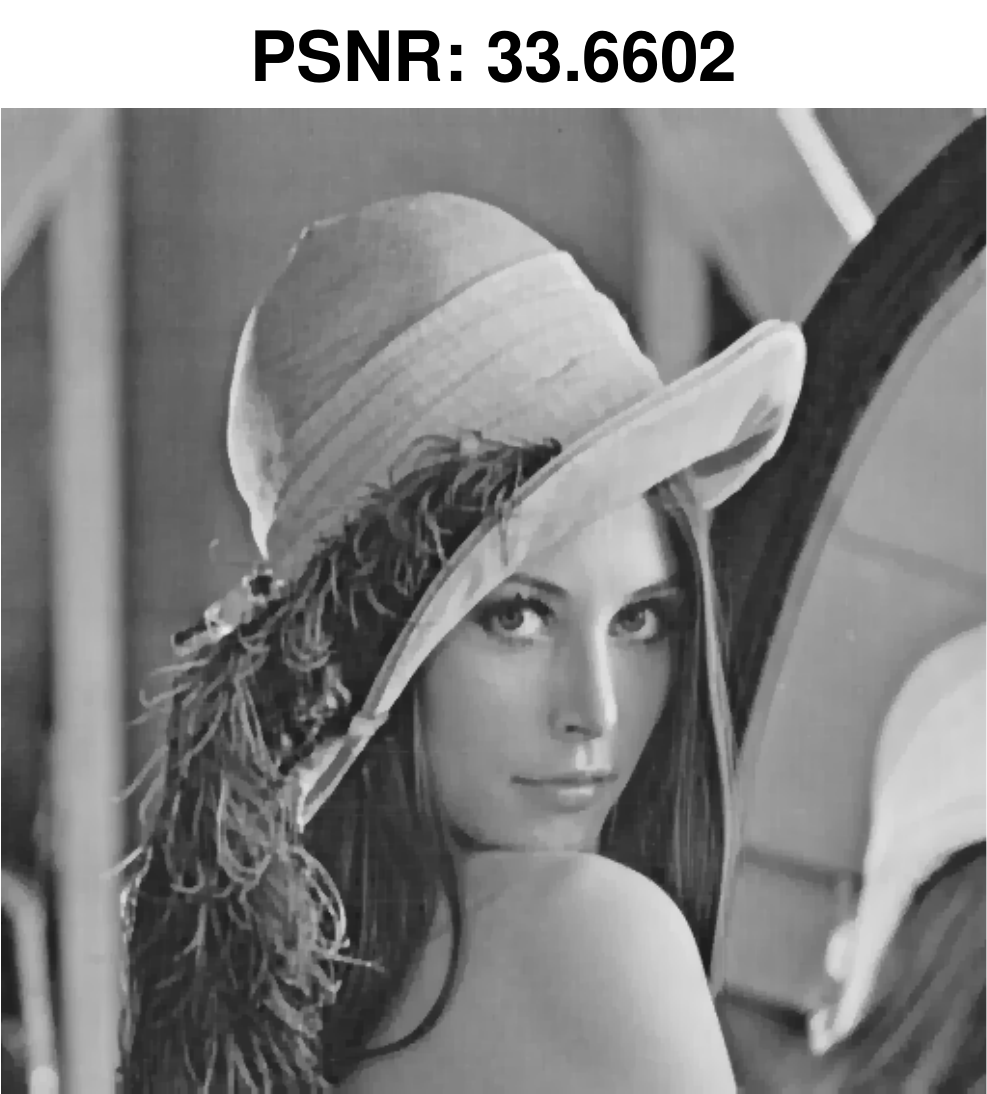}
         \caption{sd2D on the full image}
         \label{fig:lena_2DSD}
     \end{subfigure}
     \hfill
     \begin{subfigure}[t]{0.24\textwidth}
         \centering
         \includegraphics[width=\textwidth]{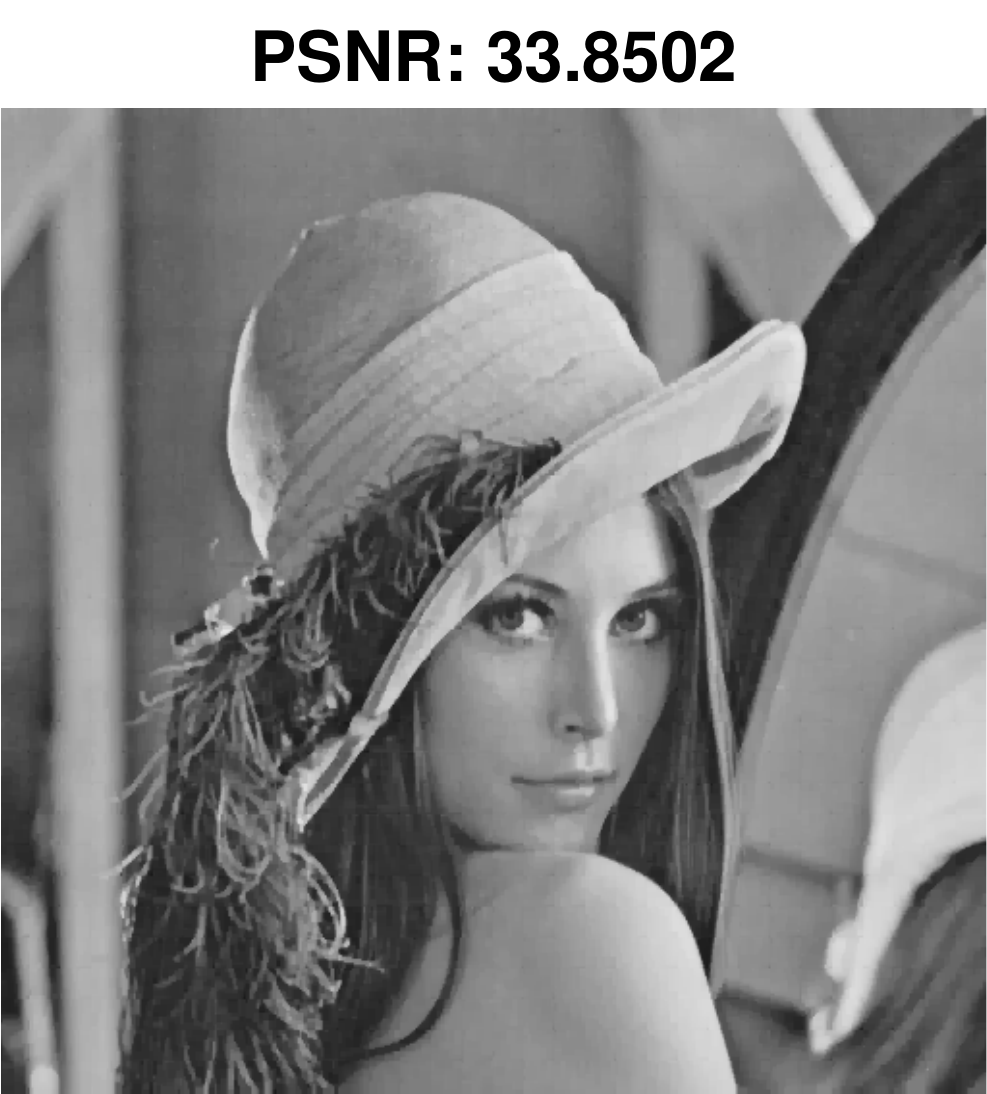}
         \caption{sd2D on size-$16\times 16$ patches}
         \label{fig:lena_2DSD_P}
     \end{subfigure}
     \hfill
     \begin{subfigure}[t]{0.24\textwidth}
         \centering
         \includegraphics[width=\textwidth]{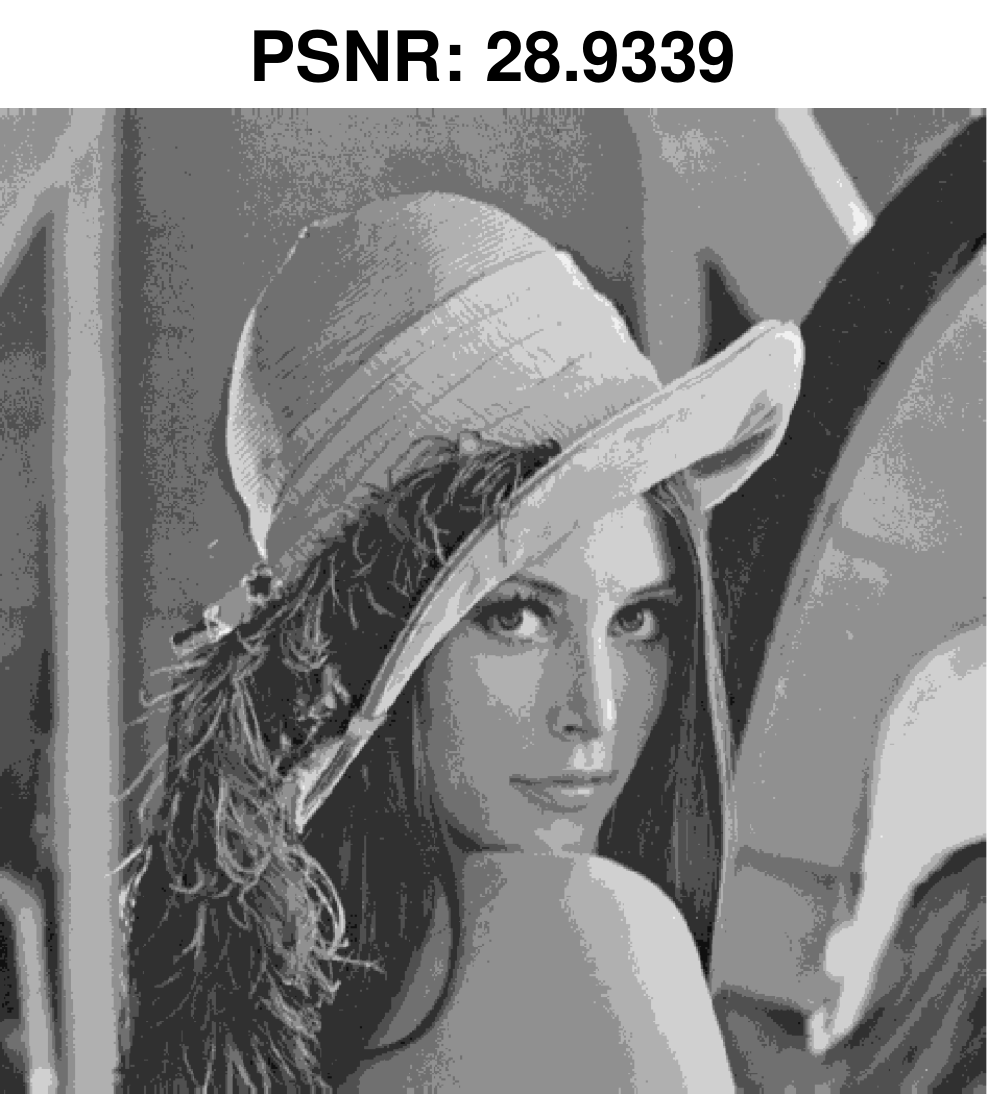}
         \caption{MSQ}
         \label{fig:lena_MSQ}
     \end{subfigure}
     \caption{Reconstructions from 2D $\Sigma\Delta$ quantization of the entire image, parallel 2D  $\Sigma\Delta$ quantization of $32^2$  patches each with size $16\times 16$, and MSQ quantization. All used 3 bits per pixel.} 
     \label{fig:lena}
\end{figure}

\begin{figure}[htbp]
     \centering
     \begin{subfigure}[t]{0.24\textwidth}
         \centering
         \includegraphics[width=1\textwidth]{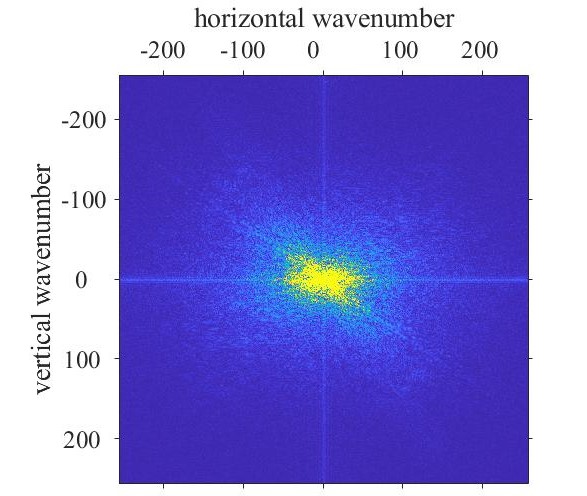}
         \caption{2D absolute spectrum of Figure \ref{fig:lena_orig}}
         \label{lena_freq_GT}
     \end{subfigure}
     \hfill
     \begin{subfigure}[t]{0.24\textwidth}
         \centering
         \includegraphics[width=1\textwidth]{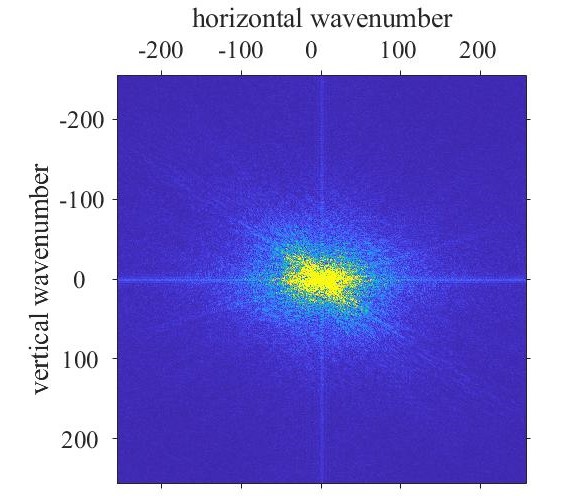}
         \caption{2D absolute spectrum of Figure \ref{fig:lena_2DSD}}
     \end{subfigure}
     \hfill
     \begin{subfigure}[t]{0.24\textwidth}
         \centering
         \includegraphics[width=1\textwidth]{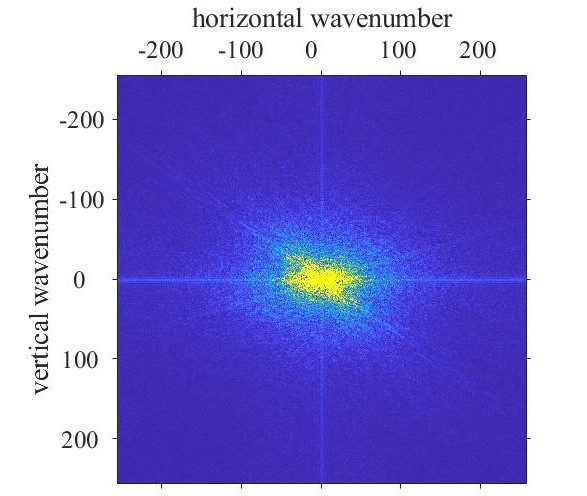}
         \caption{2D absolute spectrum of Figure \ref{fig:lena_2DSD_P}}
     \end{subfigure}
     \hfill
     \begin{subfigure}[t]{0.24\textwidth}
         \centering
         \includegraphics[width=1\textwidth]{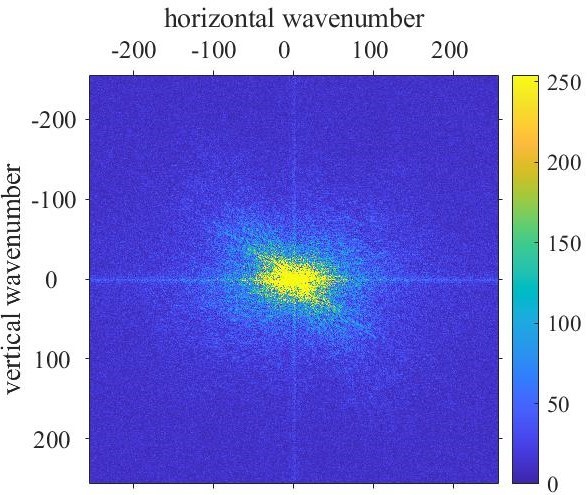}
         \caption{2D absolute spectrum of Figure \ref{fig:lena_MSQ}}
     \end{subfigure}
     \caption{Comparison of 2D absolute spectra of the images in Figure \ref{fig:lena}. From left to right are absolute spectra of the true image, that of the 2D single patch $\Sigma\Delta$ reconstruction, that of the 2D patch-based $\Sigma\Delta$ reconstruction, and that of the MSQ, respectively. The sub-figures share the same colormap.}
     \label{fig:lena_freq_orig}
\end{figure}
\begin{figure}[htbp]
     \centering
     \hfill
     \begin{subfigure}[t]{0.32\textwidth}
         \centering
         \includegraphics[width=0.8\textwidth]{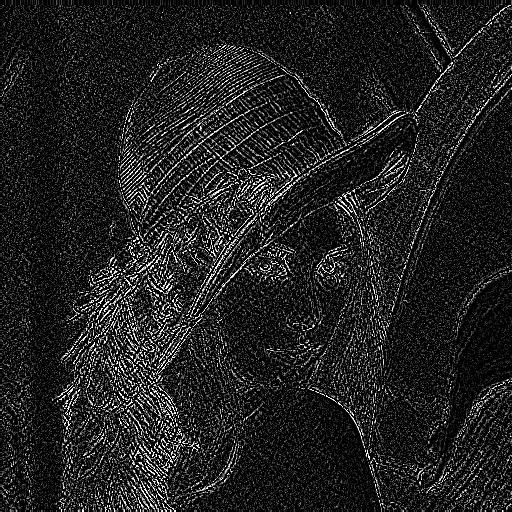}
     \caption{Residue image of Figure \ref{fig:lena_2DSD}}
     \label{fig:lena_2DSD_full_residue}
     \end{subfigure}
     \hfill
     \begin{subfigure}[t]{0.32\textwidth}
         \centering
         \includegraphics[width=0.8\textwidth]{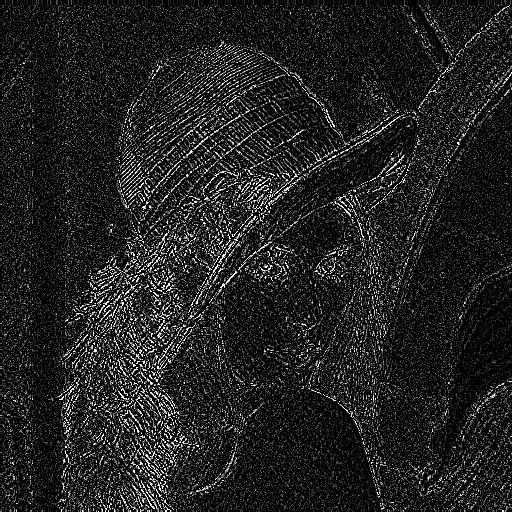}
         \caption{Residue image of Figure \ref{fig:lena_2DSD_P}}
         \label{fig:lena_2DSD_P_residue}
     \end{subfigure}
     \hfill
     \begin{subfigure}[t]{0.32\textwidth}
         \centering
         \includegraphics[width=0.985\textwidth]{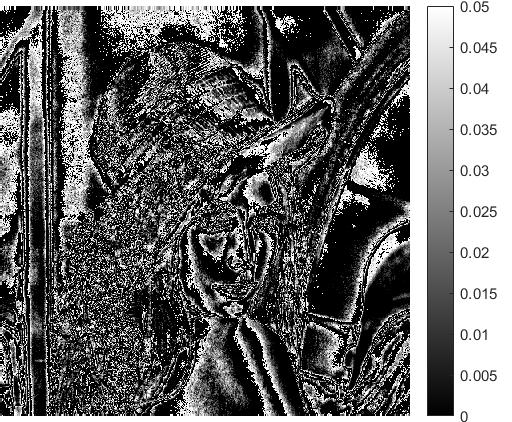}
          \caption{Residue image of Figure \ref{fig:lena_MSQ}}
          \label{fig:lena_MSQ_residue}
     \end{subfigure}
     \caption{Comparison of residues of the images in Figure \ref{fig:lena}. From left to right are the residue of the 2D single patch $\Sigma\Delta$ reconstruction, the residue of the 2D patch-based $\Sigma\Delta$ reconstruction, and that of MSQ, respectively.}
     \label{fig:lena_residue}
\end{figure}
\begin{figure}[htbp]
     \centering
     \hfill
     \begin{subfigure}[t]{0.32\textwidth}
         \centering
         \includegraphics[width=1\textwidth]{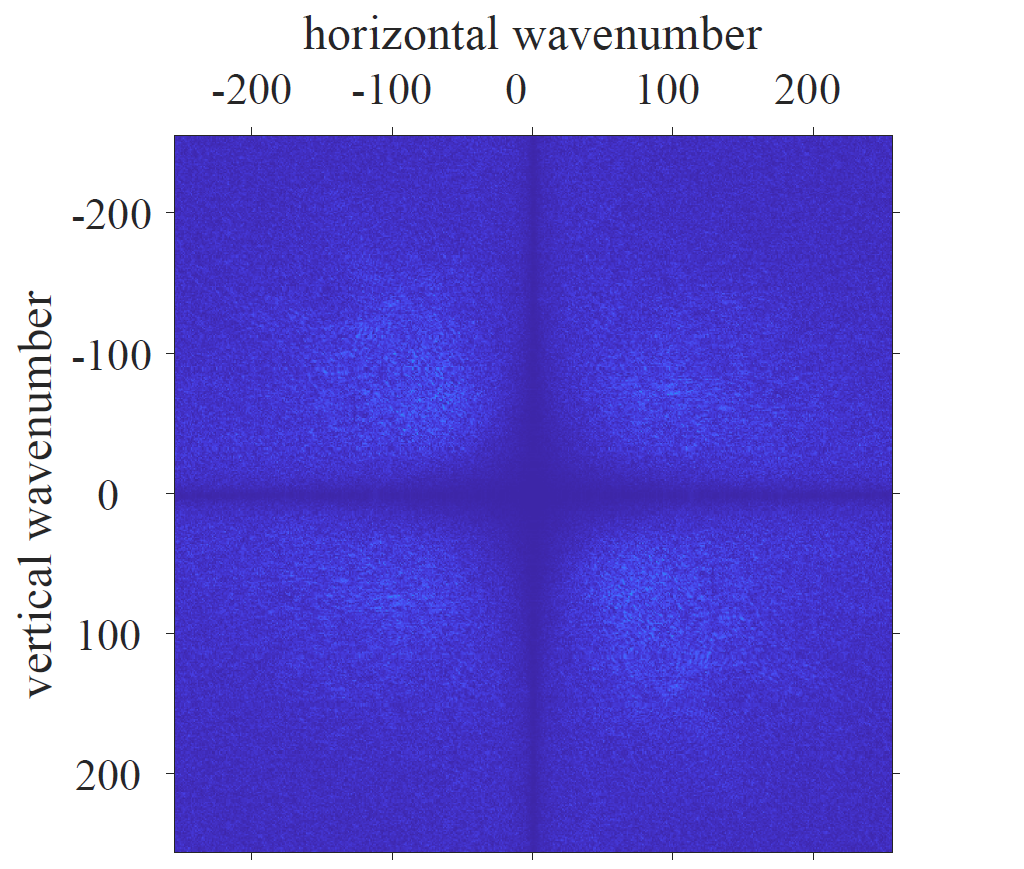}
     \caption{2D absolute spectrum of Figure \ref{fig:lena_2DSD_full_residue}}
     \end{subfigure}
     \hfill
     \begin{subfigure}[t]{0.32\textwidth}
         \centering
         \includegraphics[width=1\textwidth]{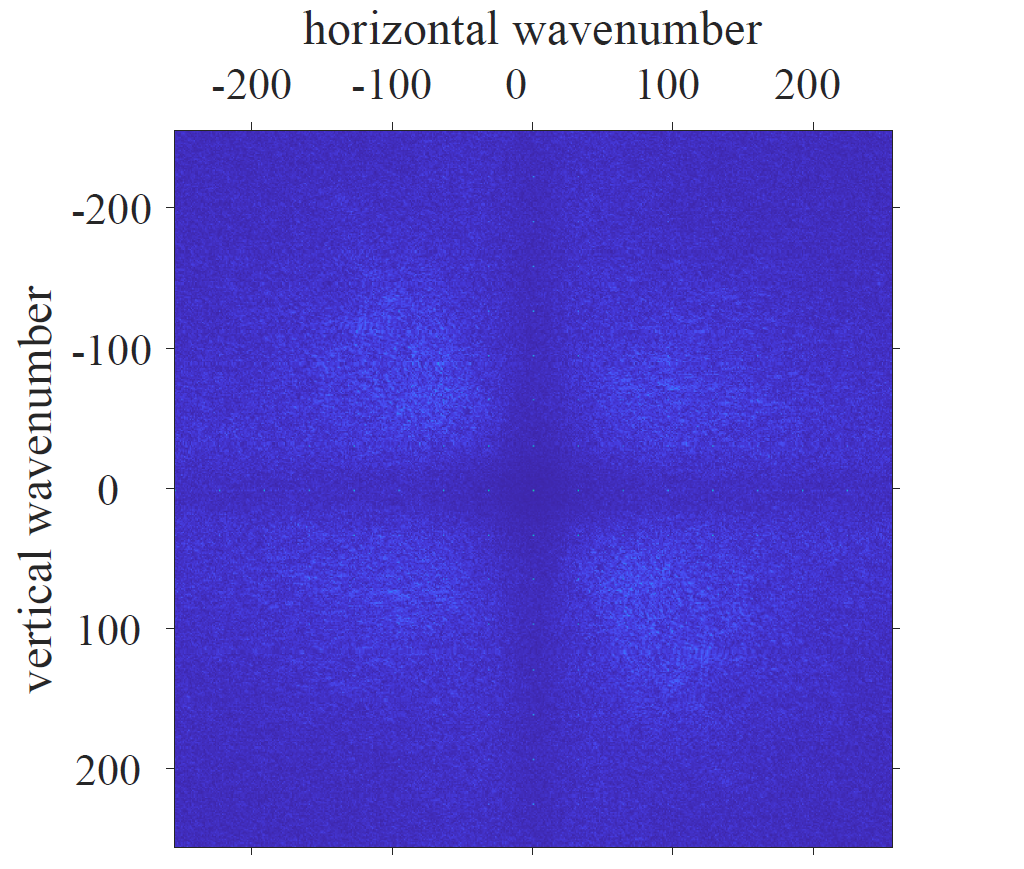}
         \caption{2D absolute spectrum of Figure \ref{fig:lena_2DSD_P_residue}}
     \end{subfigure}
     \hfill
     \begin{subfigure}[t]{0.32\textwidth}
         \centering
         \includegraphics[width=1\textwidth]{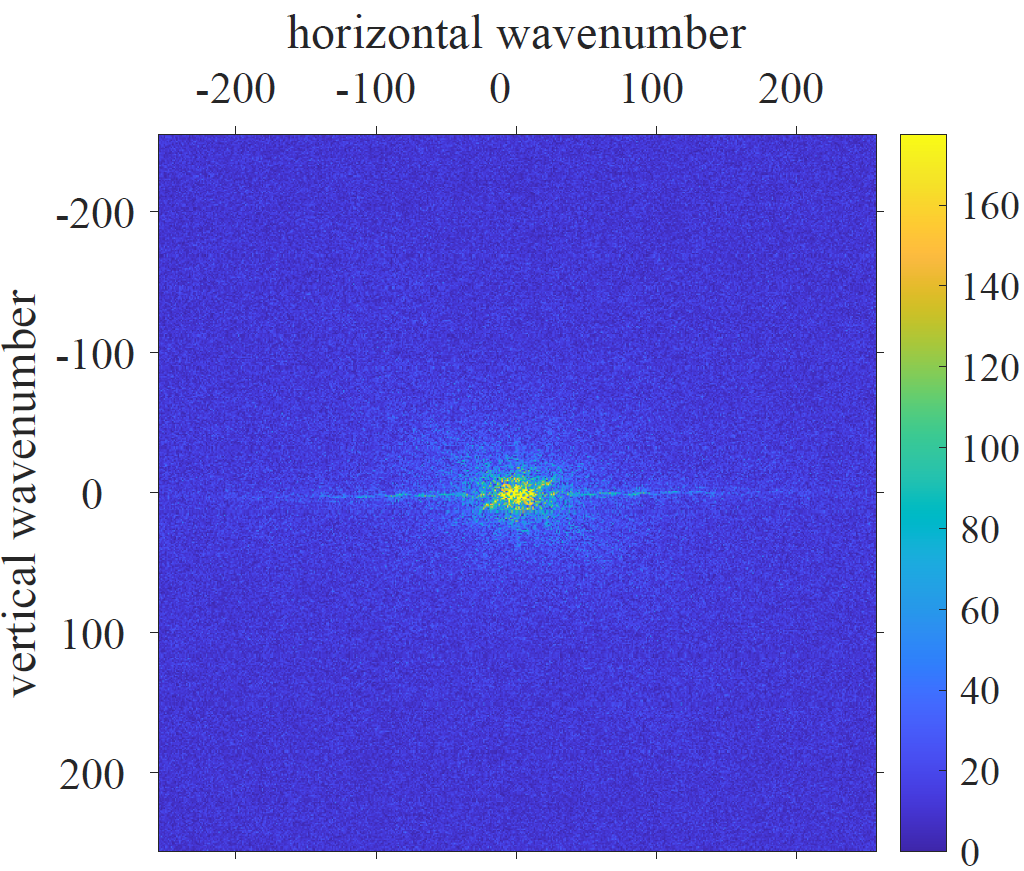}
          \caption{2D absolute spectrum of Figure \ref{fig:lena_MSQ_residue}}
     \end{subfigure}
     \caption{Comparison of the 2D absolute spectra of the residue images in Figure \ref{fig:lena_residue}. From left to right are  2D single patch $\Sigma\Delta$ reconstruction,  2D patch-based $\Sigma\Delta$ reconstruction, and  MSQ, respectively.}
     \label{fig:lena_freq_diff}
\end{figure}
In the second experiment, on the test image Lena, we evaluate the effect of dividing the image into multiple rectangle patches in (sd2D), quantizing and reconstructing each patch individually. The process can be done in parallel, which significantly reduces the reconstruction time. As shown in Figure \ref{fig:lena}, there is no visible difference between the reconstruction by a single patch and that by multiple patches. In both cases, the images look more natural and closer to the original image than MSQ, especially around the face and shoulder areas. 

To further investigate where the improved PSNR of the  Sigma Delta reconstructions comes from, we plotted in Figure \ref{fig:lena_freq_diff} the absolute spectra of the three reconstructed images in Figure \ref{fig:lena_freq_orig} as well as the absolute spectra of their residue images. The residue images (Figure \ref{fig:lena_residue}) are obtained by taking differences between the reconstructed and the original images. From Figure \ref{fig:lena_freq_diff}, we see that as predicted, our decoders can indeed retain the high-frequency information while effectively compressing the low-frequency noise.

Next, we test the performances of the proposed decoders on RGB images.
\begin{figure}[H]
     \centering
     \begin{subfigure}[t]{0.24\textwidth}
         \centering
         \includegraphics[width=\textwidth]{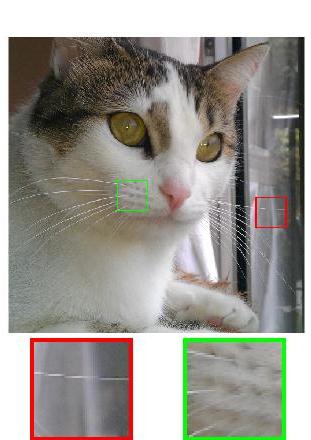}
         \caption{Original image}
         \label{fig:cat4_Orig}
     \end{subfigure}
     \hfill
     \begin{subfigure}[t]{0.24\textwidth}
         \centering
         \includegraphics[width=\textwidth]{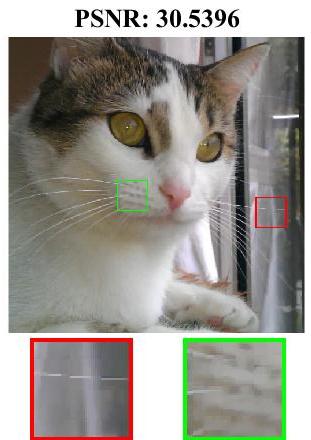}
         \caption{Parallel 2D patch-by-patch $\Sigma\Delta$ reconstruction}
         \label{fig:cat4_2D}
     \end{subfigure}
     \hfill
     \begin{subfigure}[t]{0.24\textwidth}
         \centering
         \includegraphics[width=\textwidth]{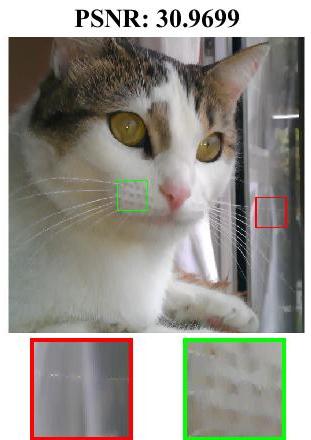}
         \caption{Parallel column-by-column $\Sigma\Delta$ reconstruction}
         \label{fig:cat4_cbc}
     \end{subfigure}
     \hfill
     \begin{subfigure}[t]{0.24\textwidth}
         \centering
         \includegraphics[width=\textwidth]{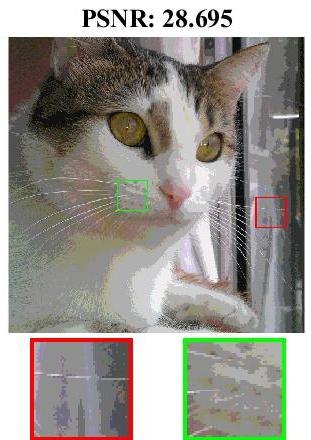}
         \caption{MSQ}
         \label{fig:cat4_MSQ}
     \end{subfigure}
     \caption{Reconstruction results of an RGB image. All methods used 3-bit alphabets.}
     \label{fig:cat4}
\end{figure}
Figure \ref{fig:cat4} shows that compared with MSQ, the proposed $\Sigma\Delta$ reconstructions did a better job at preserving the original color and getting rid of the halos. Similar to the gray-scale image case (Figure \ref{fig:cameraman}), the 2D $\Sigma\Delta$ quantization introduced less horizontal artifact than the column-by-column $\Sigma\Delta$ quantization.
\begin{figure}[H]
     \centering
     \begin{subfigure}[t]{0.24\textwidth}
         \centering
         \includegraphics[width=\textwidth]{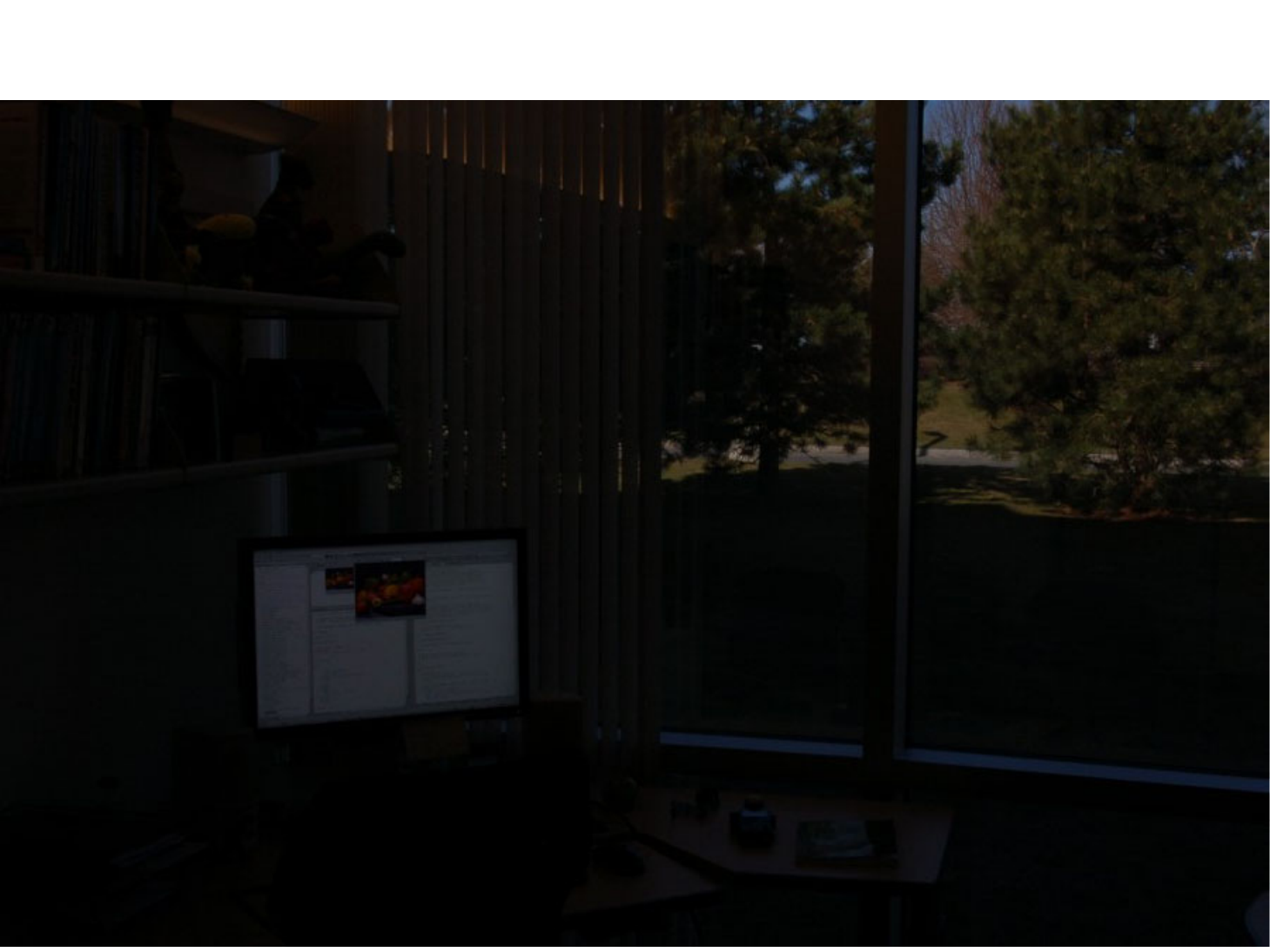}
         \caption{Original image}
         \label{fig:office_Orig}
     \end{subfigure}
     \hfill
     \begin{subfigure}[t]{0.24\textwidth}
         \centering
         \includegraphics[width=\textwidth]{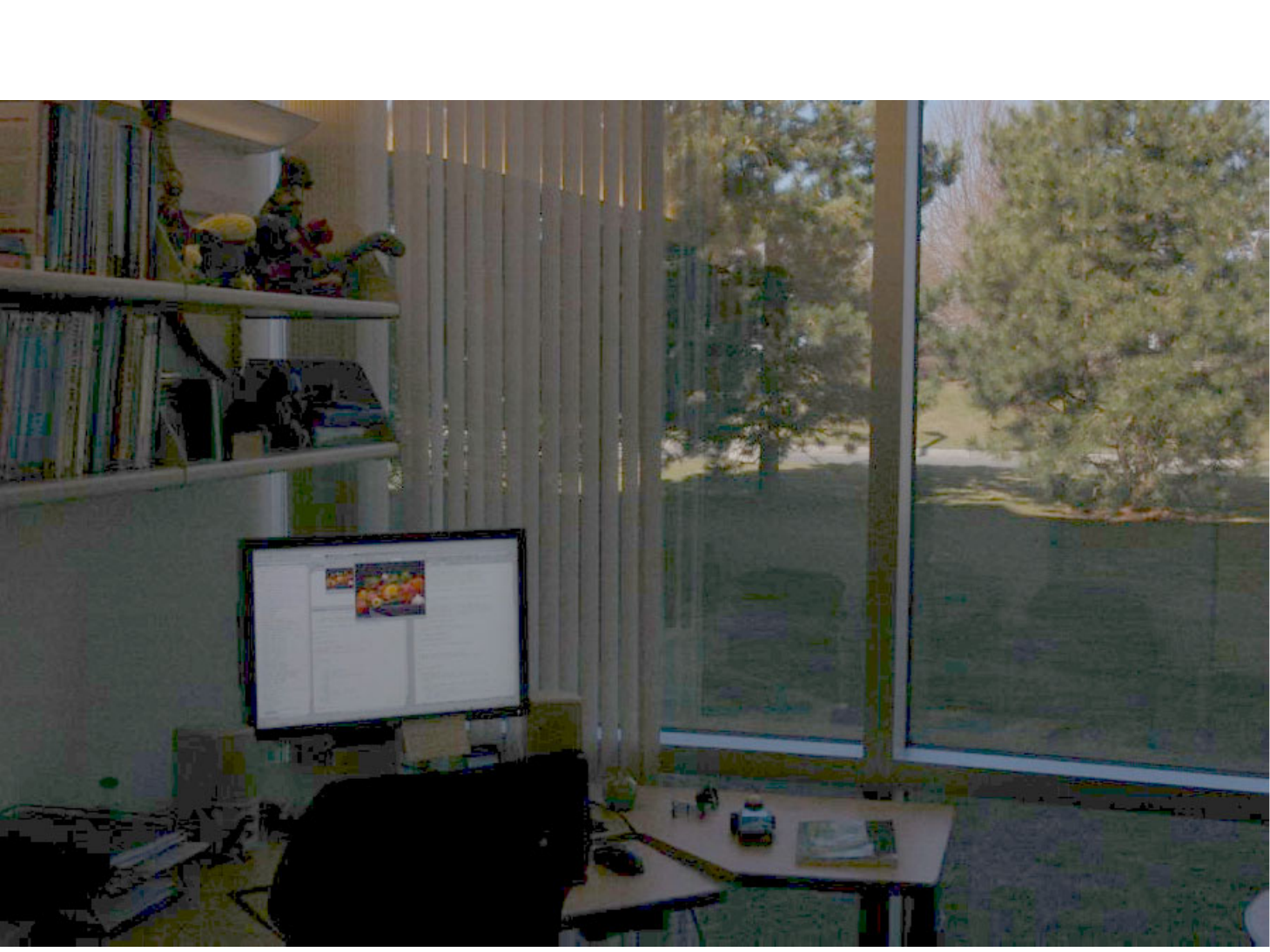}
         \caption{Original image with enhanced brightness}
         \label{fig:office_HDR}
     \end{subfigure}
     \hfill
     \begin{subfigure}[t]{0.24\textwidth}
         \centering
         \includegraphics[width=\textwidth]{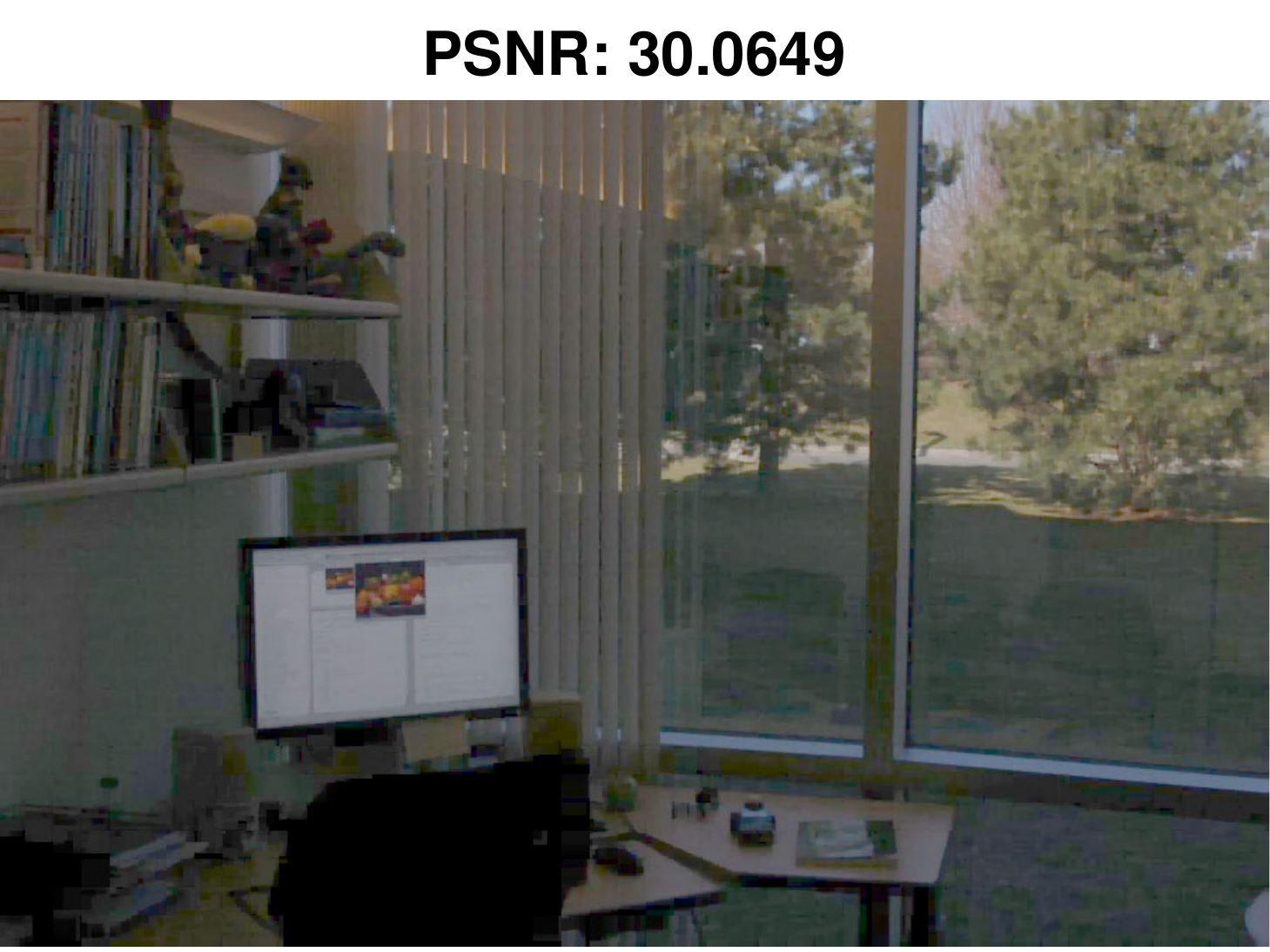}
         \caption{2D $\Sigma\Delta$ reconstruction plotted with enhanced brightness}
         \label{fig:rec_HDR}
     \end{subfigure}
     \hfill
     \begin{subfigure}[t]{0.24\textwidth}
         \centering
         \includegraphics[width=\textwidth]{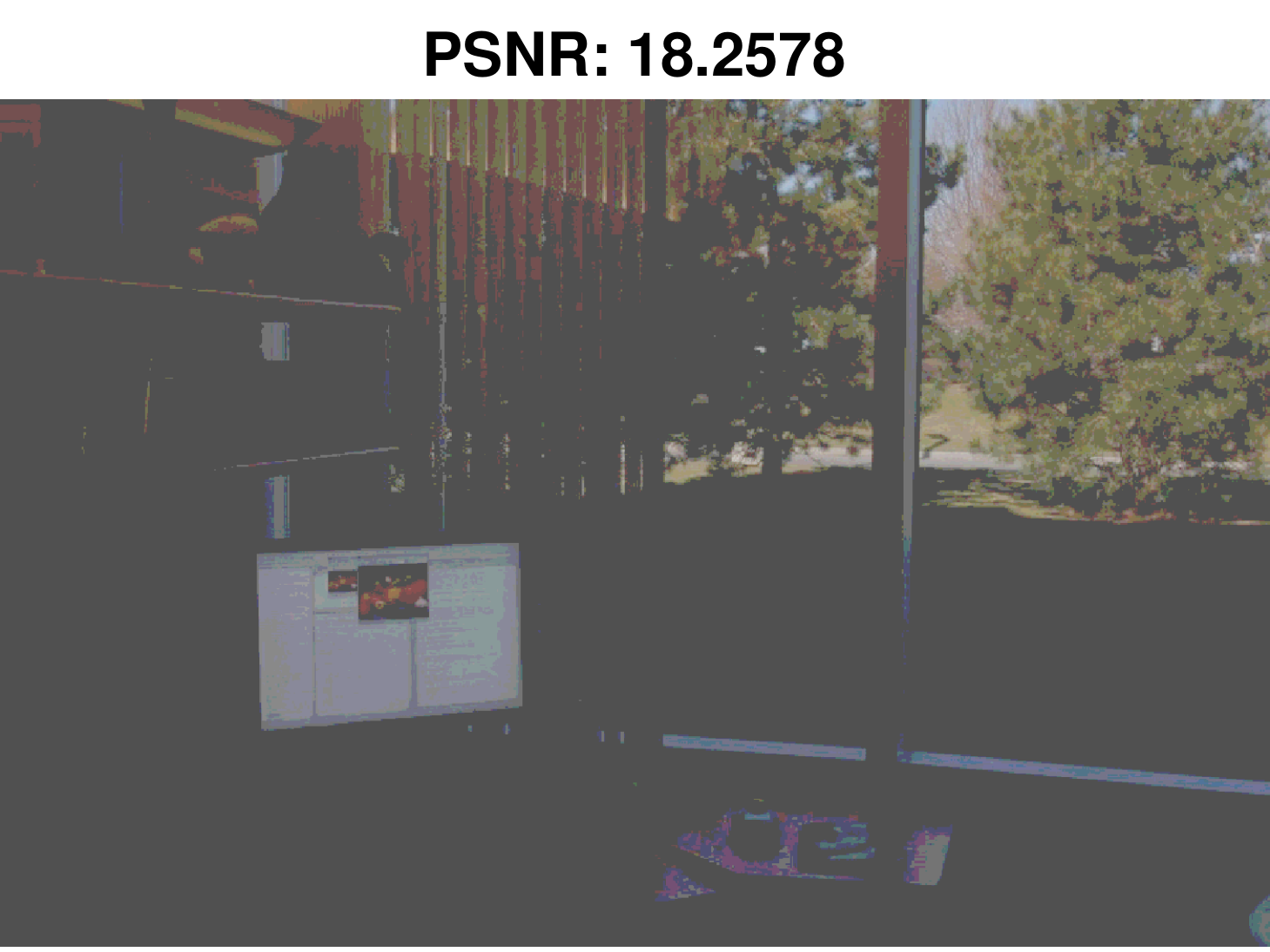}
         \caption{MSQ result plotted with enhanced brightness}
         \label{fig:MSQ_HDR}
     \end{subfigure}
            \caption{Low-pixel-intensity image quantization and reconstruction. Here we used 4-bit quantization for both $\Sigma\Delta$ quantization and MSQ. (A) The ground truth. (B)  Ground-truth image with enhanced  brightness via formula \eqref{eq:adj}. (C) Reconstructed image from the proposed encoder $Q_{2D}$ and decoder \eqref{eq:case3}. The reconstructed image is plotted with enhanced brightness via formula \eqref{eq:adj} (D) MSQ quantization with enhanced brightness. The PSNRs were computed after the brightness enhancement.}
        \label{fig:office}
\end{figure}
Next, we present an interesting result that shows the proposed scheme can greatly assist image acquisition under weak illuminations. When shooting in a dark environment or without an enough exposure time, the captured image will have a low overall pixel intensity (Figure \ref{fig:office_Orig}).  We call such an image to have a low Dynamic Range (DR), meaning that the ratio between the maximum and minimum pixel intensities is small. Visually, it means that the contrast is small. A typical way to increase the contrast is through post-processing. For example, if we use the formula \begin{equation}\label{eq:adj}
X(i,j)=X(i,j)^{1/3}
\end{equation}
to adjust the brightness, Figure \ref{fig:office_Orig} is lightened up to Figure \ref{fig:office_HDR}. However, in practice, since the brightness adjustment is a digital processing step, it has to be performed after quantization.  If the quantization is performed using MSQ, because MSQ is not good at handling images with a low DR, the image quality becomes very poor (Figure \ref{fig:MSQ_HDR}).  With current commercial cameras, this (i.e., Figure \ref{fig:MSQ_HDR}) is more or less what we could get for night scenes. In constrast, by replacing MSQ with the proposed Sigma Delta quantization, a significant improvement can be obtained under the same bit budget. As shown in Figure \ref{fig:rec_HDR},  hidden details in the dark scene are better revealed and less artifact is introduced. This promising performance under weak illuminations also carries over to the situations with extremely strong illuminations or with large and small intensity portions co-existing in one image. 
\begin{figure}[htbp]
     \centering
     \begin{subfigure}[t]{0.31\textwidth}
         \centering
         \includegraphics[width=1\textwidth]{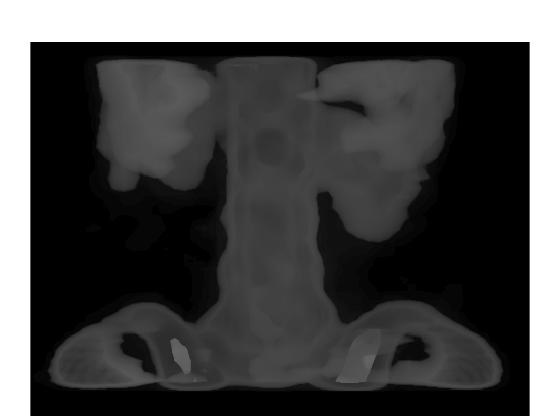}
         \caption{Original image}
     \end{subfigure}
     \hfill
     \begin{subfigure}[t]{0.31\textwidth}
         \centering
         \includegraphics[width=1\textwidth]{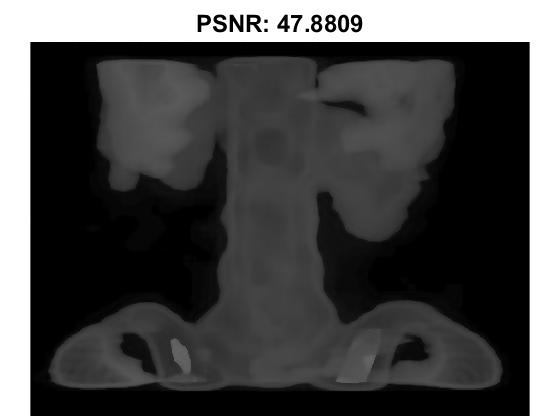}
         \caption{2D $\Sigma\Delta$ reconstruction}
     \end{subfigure}
     \hfill
     \begin{subfigure}[t]{0.31\textwidth}
         \centering
         \includegraphics[width=1\textwidth]{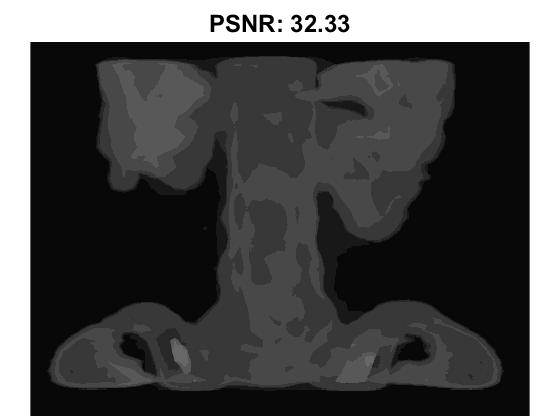}
         \caption{MSQ}
     \end{subfigure}
        \caption{{Reconstruction results of the spine image from 2D $\Sigma\Delta$ quantization and MSQ quantization, both with a 4-bit alphabet.}}
        \label{fig:spine}
\end{figure}

Last, we show that the proposed method is well-suited for medical imaging applications. Medical imaging has a higher tolerance for a longer decoding time, so there is no need to divide the image into patches.  Also, in the previous experiments, $\Sigma\Delta$ methods were shown to be superior to MSQ at preserving the intensity variations among the pixels, which is crucial for medical images, as different intensity levels indicate different tissue types, that help doctors to detect the abnormally. In Figure \ref{fig:spine}, on the spine image, the 2D $\Sigma\Delta$ reconstruction from a 4-bit alphabet has almost no visual difference than the 8-bit original image, while MSQ quantization results in a larger distortion. This experiment indicates that the proposed $\Sigma\Delta$ quantization may be especially suitable for medical imaging.  
        

\section{Appendix}
\subsection{Proof of Proposition \ref{pro:opt}}\label{sec:append1}
\begin{lemma}\label{lemma1} [Lemma 2.4 \cite{candes2013super}]
Let $\mathbb{T}$ be the 1-dimensional torus, suppose $T=\{t_1,t_2,\cdots, t_s\} \subset \mathbb{T}$ satisfies the $\Lambda_M$-minimum separation condition, i.e., $\min_{t,t'\in T,t\not= t'}|t-t'|\geq {\frac{2}{M}},\ M\geq 128$. Let $v\in \mathbb{C}^{|T|}$ be an arbitrary vector with $|v_j|=1,j=1,2,\cdots, s.$ Then there exists a low-frequency trigonometric polynomial
$$q(t)=\sum_{k=-M}^M c_k e^{i2\pi kt}, t\in [0,1]$$
obeying the following properties:
$$q(t_j)=v_j,\;\;t_j\in T,$$
$$|q(t)|\leq 1-{C_a M^{2}{(t-t_j)^2}},\;\;\;t\in S_M(j),$$
$$|q(t)|<1-C_b,\;\;t\in S_M^c,$$
with $0<C_b\leq 0.16^2C_a<1$.
\end{lemma}
\begin{proof}[\textbf{Proof of Proposition \ref{pro:opt}}]
We proceed by contradiction, assume there exists a $d$-bit alphabet $\widetilde{\mathcal{A}}$  whose stability constant $\widetilde{C}$ is smaller than $C$, i.e., $\widetilde C<C$. Let the alphabet $\widetilde{\mathcal{A}}$ be $c_1<c_2<\cdots<c_N$, with $N=2^d$. Assume $c_1<\cdots<c_i< a\leq c_{i+1}<\cdots<c_j\leq b< c_{j+1}<\cdots<c_N$. Note that there is no restriction on the range of values of the alphabet, that is, it is possible that $a\leq c_1$ or $b\geq c_N$. Also, denote the largest interval length in the alphabet within $[a,b]$ as $2I$, i.e., $I=\frac{1}{2}\max\{c_{i+2}-c_{i+1},\cdots,c_j-c_{j-1}\}$ and define $\widetilde I:=\frac{1}{2}\max\{2I,b-c_j\}$. The case $d\geq 3$ is slightly easier than $d=2$, so we first prove the result for $d\geq 3$.

For $d\geq 3$, we start with proving that the alphabet has at least two elements within $[a,b]$, so $I$ is well defined. Notice that $C=\frac{b-a}{2(2^d-3)}\leq\frac{b-a}{10}$. If there is no or one $c_\ell$ between $a$ and $b$, i.e., $a\leq c_\ell\leq b,\ 1\leq\ell\leq N$, then we can choose $a\leq y_{1,1}\leq b$ properly such that $|y_{1,1}-Q_{\widetilde{\mathcal{A}}}(y_{1,1})|\geq\frac{b-a}{4}>C>\widetilde C$, which leads to a contradiction.

Next, we consider the following cases:
\begin{itemize}
    \item $a\leq c_1$ or $b\geq c_N$. For this case, we have the following two sub-cases that are mutually exclusive: 1) $a\leq c_1$ and $c_1-a\geq b-c_N$ and 2) $b \geq c_N$ and $b-c_N \geq c_1-a$. A closer look indicates that these two cases are exactly the same upon exchanging the roles of $a$ and $b$. Hence without loss of generality, we assume 1) holds: $a\leq c_1$ and $c_1-a\geq b-c_N$. Next, we specify the following possibilities:
    
    (a) $c_N\leq b$. Let $[c_{\ell},c_{\ell+1}]$ be the largest interval in $\widetilde{\mathcal{A}}$ for some $\ell$, i.e., $c_{\ell+1}-c_\ell=2I$. Choose $y_{1,1}=\frac{c_\ell+c_{\ell+1}}{2}-\epsilon,\ y_{1,2}=y_{2,1}=c_\ell+2\epsilon$, $y_{2,2}=a$ with small enough $\epsilon>0$ such as $\epsilon=10^{-10}\cdot(C-\widetilde C)$. This leads to $u_{1,1}=I-\epsilon,\ u_{1,2}=u_{2,1}=-I+\epsilon$, then $q_{2,2}=Q_{\widetilde{\mathcal{A}}}(y_{2,2}+u_{1,2}+u_{2,1}-u_{1,1})= Q_{\widetilde{\mathcal{A}}}(a-3I+3\epsilon)$, the quantization error
    \begin{align*}
    |u_{2,2}|&=|y_{2,2}+u_{1,2}+u_{2,1}-u_{1,1}-q_{2,2}|\\
    &= |Q_{\widetilde{\mathcal{A}}}(a-3I+3\epsilon)-(a-3I+3\epsilon)|\\
    &=c_1-a+3I-3\epsilon\\
    &\geq c_1-a+\frac{3}{2}\cdot\frac{c_N-c_1}{2^d-1}-3\epsilon\\
    &\geq c_1-a+\frac{3}{2}\cdot\frac{b-a-2(c_1-a)}{2^d-1}-3\epsilon\\
    &\geq \frac{3(b-a)}{2(2^d-1)}-3\epsilon\\
    &\geq \frac{b-a}{2(2^d-3)}-3\epsilon\\
    &=C-3\epsilon>\widetilde C.
    \end{align*}
    The second inequality used the assumption $c_1-a\geq b-c_N$, and the third one used $c_1-a\geq 0$ and $ d\geq 3$. Then this contradicts the assumption $\|u\|_{\textrm{max}}\leq \widetilde C$.
    
    (b) $c_N>b$ and $c_{j+1}-c_j< 3(b-c_j)$. If $2\widetilde I= \max\{2I,b-c_j\}=b-c_j$, let $y_{1,1}=c_j+\widetilde I-\epsilon,\ y_{1,2}=y_{2,1}=\frac{c_{j+1}+c_j}{2}-\widetilde I+2\epsilon\leq b$ with a sufficiently small $\epsilon>0$ as in (a), then $u_{1,1}=\widetilde I-\epsilon,\ u_{1,2}=u_{2,1}=-\frac{c_{j+1}-c_j}{2}+\epsilon<-\widetilde I+\epsilon$. If $2\widetilde I=2I$, we can choose $y_{1,1},y_{1,2},y_{2,1}$ as in (a), then $u_{1,1}=\widetilde I-\epsilon,\ u_{1,2}=u_{2,1}=-\widetilde I+\epsilon$. In both cases, let $y_{2,2}=a$, we have $y_{2,2}+u_{1,2}+u_{2,1}-u_{1,1}\leq a-3\widetilde I+3\epsilon<c_1$, and $q_{2,2}=Q_{\widetilde{\mathcal{A}}}(y_{2,2}+u_{1,2}+u_{2,1}-u_{1,1})=c_1$. The quantization error at $q_{2,2}$ is 
    \begin{align*}
    |u_{2,2}|&=|c_1-(y_{2,2}+u_{1,2}+u_{2,1}-u_{1,1})|\\
    &\geq c_1-a+3\widetilde I-3\epsilon\\
    &\geq c_1-a+\frac{3}{2}\cdot\frac{b-c_1}{2^d-1}-3\epsilon\\
    &\geq c_1-a+\frac{3}{2}\cdot\frac{b-a-(c_1-a)}{2^d-1}-3\epsilon\\
    &\geq \frac{3(b-a)}{2(2^d-1)}-3\epsilon\\
    &\geq \frac{b-a}{2(2^d-3)}-3\epsilon\\
    &=C-3\epsilon>\widetilde C.
    \end{align*}
    This leads to a contradiction.
    
    (c) $c_N>b$ and $c_{j+1}-c_j\geq  3(b-c_j)$. If we choose $y_{1,1}=\frac{b+c_j}{2}-\epsilon$ with some small $\epsilon$ and $y_{1,2}=b$, then $u_{1,1}=\frac{b-c_j}{2}-\epsilon,\ u_{1,2}=\frac{3}{2}(b-c_j)-\epsilon\leq\widetilde C$. Since it holds for arbitrary small $\epsilon$, we must have $b-c_j\leq\frac{2}{3}\widetilde C$. This gives $b-\frac{2}{3}\widetilde C\leq c_j\leq b$ and $2I\geq\frac{c_j-c_1}{j-1}\geq \frac{b-\frac{2}{3}\widetilde C-c_1}{2^d-2}$, where the last inequality is due to the assumption $c_N>b$ so that $j\leq N-1$. Same as in (a), we can choose $y_{1,1},y_{1,2},y_{2,1}$ properly and $y_{2,2}=a$, such that $u_{1,1}=I-\epsilon,\ u_{1,2}=u_{2,1}=-I+\epsilon$, provided that $\epsilon$ is small enough. Then the quantization error at $q_{2,2}=Q_{\widetilde{\mathcal{A}}}(a-3I+3\epsilon)=c_1$ is
    \begin{align*}
    |u_{2,2}|&=|Q_{\widetilde{\mathcal{A}}}(a-3 I+3\epsilon)-(a-3 I+3\epsilon)|\\
    &=c_1-a+3 I-3\epsilon\\
    &\geq c_1-a+\frac{3}{2}\cdot\frac{b-\frac{2}{3}\widetilde C-c_1}{2^d-2}-3\epsilon\\
    &\geq c_1-a+\frac{3}{2}\cdot\frac{b-a-\frac{b-a}{3(2^d-3)}-(c_1-a)}{2^d-2}-3\epsilon\\
    &\geq \frac{3}{2}\cdot\frac{b-a-\frac{b-a}{3(2^d-3)}}{2^d-2}-3\epsilon\\
    &\geq \frac{b-a}{2(2^d-3)}-3\epsilon\\
    &=C-3\epsilon>\widetilde C.
    \end{align*}
    This also leads to a contradiction. 
    \item $a>c_1$ and $b<c_N$. If we also have $a>c_2$ and $b<c_{N-1}$, then we can easily choose a proper $a\leq y_{1,1}\leq b$ with quantization error at least ${\frac{1}{2}}\max\{2I,c_{i+1}-a,b-c_j\}\geq \frac{b-a}{2(2^d-4+1)}=C>\widetilde C$, which leads to contradiction.
    
    Therefore, without loss of generality, assume $c_1<a\leq c_2$ and $c_2-a\geq b-c_{N-1}$. Similar to above, we specify the following sub-cases:
    
    (d) $c_{N-1}\leq b$. We first show that for an arbitrary constant $a-3I<\xi<c_2$, one can choose $y_{1,1},y_{1,2},y_{2,1},y_{2,2}$ properly such that $\xi=y_{2,2}+u_{1,2}+u_{2,1}-u_{1,1}$. If $a\leq\xi< c_2$, set $y_{1,1}=y_{1,2}=y_{2,1}=c_2,\ y_{2,2}=\xi$, then $u_{1,1}=u_{1,2}=u_{2,1}=0$ and $ y_{2,2}+u_{1,2}+u_{2,1}-u_{1,1}=\xi$. If $a-3I<\xi<a$, denote $w=\frac{a-\xi}{3}$, then $0<w<I$. Let $[c_{\ell},c_{\ell+1}]$ be the largest interval in $\widetilde{\mathcal{A}}$ within $[a,b]$ for some $\ell$, i.e., $c_{\ell+1}-c_\ell=2I$. Choose $y_{1,1}=c_\ell+w,\ y_{1,2}=y_{2,1}=c_{\ell+1}-2w,\ y_{2,2}=a$, then $u_{1,1}=w,\ u_{1,2}=u_{2,1}=-w$, we also have $y_{2,2}+u_{1,2}+u_{2,1}-u_{1,1}=a-3w=\xi$.
    
    Hence whatever $c_1$ is, we can always obtain the quantization error 
    \begin{align*}
    \max_{a-3I<\xi<c_2}|\xi-Q_{\widetilde{\mathcal{A}}}(\xi)|&\geq\frac{1}{3}(c_2-(a-3I))\\
    &=\frac{1}{3}(c_2-a)+I\\
    &\geq \frac{1}{3}(c_2-a)+\frac{1}{2}\cdot\frac{c_{N-1}-c_2}{2^d-3}\\
    &\geq \frac{1}{3}(c_2-a)+\frac{1}{2}\cdot\frac{b-a-2(c_2-a)}{2^d-3}\\
    &\geq \frac{b-a}{2(2^d-3)}+(\frac{1}{3}-\frac{1}{2^d-3})(c_2-a)\\
    &\geq \frac{b-a}{2(2^d-3)}\\
    &=C>\widetilde C.
    \end{align*}
    This leads to a contradiction. Here to see the first inequality, we consider the following sub-cases: \romannumeral 1) if $c_1\geq\frac{1}{3}(c_2+2a-6I), \max_{a-3I<\xi<c_2}|\xi-Q_{\widetilde{\mathcal{A}}}(\xi)|\geq |a-3I- Q_{\widetilde{\mathcal{A}}}(a-3I)|=|a-3I-c_1|\geq \frac{1}{3}(c_2+2a-6I)-(a-3I)=\frac{1}{3}(c_2-(a-3I))$; \romannumeral 2) if $c_1< \frac{1}{3}(c_2+2a-6I)$, take $\xi=\frac{1}{3}(2c_2+a-3I)$, then $\max_{a-3I<\xi<c_2}|\xi-Q_{\widetilde{\mathcal{A}}}(\xi)|\geq |\frac{1}{3}(2c_2+a-3I)- Q_{\widetilde{\mathcal{A}}}(\frac{1}{3}(2c_2+a-3I))|=|\frac{1}{3}(2c_2+a-3I)- c_2|=\frac{1}{3}(c_2-(a-3I))$. 
    
    (e) $c_{N-1}>b$ and $c_{j+1}-c_j\leq \frac{3}{2}(b-c_j)$. Similar to (d), we first show that one can choose $y_{1,1},y_{1,2},y_{2,1},y_{2,2}$ properly to make $y_{2,2}+u_{1,2}+u_{2,1}-u_{1,1}$ equal to an arbitrary constant between $a-3\widetilde I$ and $c_2$. If $2\widetilde I=2I$, it follows the same reasoning as in (d), here we discuss the case when $2\widetilde I=b-c_j$.
    If $a\leq\xi<c_2$, let $y_{1,1}=y_{1,2}=y_{2,1}=c_j,\ y_{2,2}=\xi$, then $y_{2,2}+u_{1,2}+u_{2,1}-u_{1,1}=\xi$. If $a-3\widetilde I<\xi<a$, denote $w=a-\xi$, then $0<w< 3\widetilde I$, we specify the following sub-cases: \romannumeral 1) if $ 0<w\leq \widetilde I$, let $y_{1,1}=c_j+w,\ y_{1,2}=y_{2,1}=c_j-w,\ y_{2,2}=a$, then $u_{1,1}=w,\ u_{1,2}=u_{2,1}=0,\ y_{2,2}+u_{1,2}+u_{2,1}-u_{1,1}=a-w=\xi$; \romannumeral 2) if $\widetilde I<w\leq 2\widetilde I$, let $y_{1,1}=c_j+w-\widetilde I, y_{1,2}=c_{j+1}-w,\ y_{2,1}=c_j-(w-\widetilde I)$, and $y_{2,2}=a$, then $u_{1,1}=w-\widetilde I,\ u_{1,2}= -\widetilde I,\ u_{2,1}=0,\ y_{2,2}+u_{1,2}+u_{2,1}-u_{1,1}=a-w=\xi$; \romannumeral 3) if $2\widetilde I <w<3\widetilde I$, let $y_{1,1}=c_j+w-2\widetilde I,\ y_{1,2}=y_{2,1}=c_{j+1}-w+\widetilde I$ and $y_{2,2}=a$, then $u_{1,1}=w-2\widetilde I,\ u_{1,2}=u_{2,1}=-\widetilde I,\ y_{2,2}+u_{1,2}+u_{2,1}-u_{1,1}=a-w=\xi$.

    Therefore, whatever $c_1$ is, the worst case quantization error can reach 
    \begin{align*}
    \max_{a-3\widetilde I<\xi<c_2}|\xi-Q_{\widetilde{\mathcal{A}}}(\xi)|&=\frac{1}{3}(c_2-(a-3\widetilde I))\\
    &=\frac{1}{3}(c_2-a)+\widetilde I\\
    &\geq \frac{1}{3}(c_2-a)+\frac{1}{2}\cdot\frac{b-c_2}{2^d-3}\\
    &\geq \frac{1}{3}(c_2-a)+\frac{1}{2}\cdot\frac{b-a-(c_2-a)}{2^d-3}\\
    &\geq \frac{b-a}{2(2^d-3)}+(\frac{1}{3}-\frac{1}{2(2^d-3)})(c_2-a)\\
    &\geq \frac{b-a}{2(2^d-3)}\\
    &=C>\widetilde C.
    \end{align*}
    This leads to a contradiction.
    
    (f) $c_{N-1}>b$ and $c_{j+1}-c_j> \frac{3}{2}(b-c_j)$. We must have $j\leq N-2$ and $b-c_j\leq\frac{4}{3}\widetilde C$. Then {similar to (d), we can obtain the following quantization error}
    \begin{align*}
    \max_{a-3I<\xi<c_2}|\xi-Q_{\widetilde{\mathcal{A}}}(\xi)|&=\frac{1}{3}(c_2-(a-3I))\\
    &=\frac{1}{3}(c_2-a)+I\\
    &\geq \frac{1}{3}(c_2-a)+\frac{1}{2}\cdot\frac{b-\frac{4}{3}\widetilde C-c_2}{2^d-4}\\
    &\geq \frac{1}{3}(c_2-a)+\frac{1}{2}\cdot\frac{b-a-\frac{2(b-a)}{3(2^d-3)}-(c_2-a)}{2^d-4}\\
    &\geq \frac{b-a}{2(2^d-3)}+\big(\frac{1}{3}-\frac{1}{2(2^d-4)}\big)(c_2-a)\\
    &\geq \frac{b-a}{2(2^d-3)}\\
    &=C>\widetilde C.
    \end{align*}
This also leads to a contradiction. We have exhausted all the cases for $d\geq 3$.
\end{itemize}
For the case $d=2$, there are only $4$ elements in the alphabet $\widetilde{\mathcal{A}}=\{c_1,c_2,c_3,c_4\}$ with $c_1<c_2<c_3<c_4$. Consider the case $a\leq c_1$ or $b\geq c_4$, if there are at least two elements in $\widetilde{\mathcal{A}}$ that are within $[a,b]$, the proof follows the same reasoning as $d\geq 3$. Here we discuss the case that $a\leq c_1$ and only one element of $\widetilde{\mathcal{A}}$ lies within $[a,b]$, i.e., $a\leq c_1\leq b<c_2<c_3<c_4$. In this case, {we must have $c_1-a<C=\frac{b-a}{2}$, and $c_2-b<\frac{b-a}{2}$. For $\epsilon$ that is small enough, let $y_{1,1}=\frac{c_1+c_2}{2}-\epsilon,\ y_{1,2}=y_{2,1}=c_1+2\epsilon,\ y_{2,2}=a$, then $u_{1,1}=\frac{c_2-c_1}{2}-\epsilon,\ u_{1,2}=u_{2,1}=-\frac{c_2-c_1}{2}+\epsilon,\ u_{2,2}=a-\frac{3}{2}(c_2-c_1)+3\epsilon-c_1$. Hence $|u_{2,2}|\geq\frac{3}{2}(c_2-c_1)>\frac{3}{2}(b-c_1)>\frac{3}{4}(b-a)>C>\widetilde C$,}
which leads to a contradiction.

Next, we discuss the two remaining cases when both $a>c_1$ and $b<c_4$ hold: $c_1< a\leq c_2\leq b< c_3<c_4$ and $c_1<a\leq c_2<c_3\leq b<c_4$, which are the two cases that there is $1$ or $2$ elements in the alphabet between $a$ and $b$, respectively.
\begin{itemize}
\item $c_1<a\leq c_2<c_3\leq b<c_4$. Without loss of generality, assume $c_2+c_3\geq b+a$. Since $\widetilde C<C=\frac{b-a}{2}$,  we must have $c_2-c_1<b-a$. Combining these two inequalities we get $c_1+c_3> 2a$, then $ a<\frac{1}{2}(c_1+c_3)<b$. Choose some small $\epsilon$ and the first $3\times 3$ entries of $y$ as
$$y=\begin{pmatrix}
\frac{1}{2}(c_3+c_2)-\epsilon & c_2 & c_2+2\epsilon &\cdots \\
c_2 & c_2 &\frac{1}{2}(c_1+c_3) & \cdots \\
c_2+2\epsilon&
\frac{1}{2}(c_1+c_3) & a &\cdots\\
\cdots & \cdots & \cdots & \cdots 
\end{pmatrix}.
$$
Provided that $\epsilon$ is small enough, one can check that the first $3\times 3$ entries in $u$ are as follows
\begingroup
\renewcommand*{\arraystretch}{2}
$$u=\begin{pmatrix}
\frac{1}{2}(c_3-c_2)-\epsilon & \frac{1}{2}(c_3-c_2)-\epsilon & -\frac{1}{2}(c_3-c_2)+\epsilon& \cdots \\
\frac{1}{2}(c_3-c_2)-\epsilon & \frac{1}{2}(c_3-c_2)-\epsilon & -\frac{1}{2}(c_2-c_1)+\epsilon & \cdots \\
-\frac{1}{2}(c_3-c_2)+\epsilon &
-\frac{1}{2}(c_2-c_1)+\epsilon & a-\frac{1}{2}(c_3+c_2)+3\epsilon &\cdots\\
\cdots & \cdots & \cdots & \cdots 
\end{pmatrix}.
$$
\endgroup
By assumption, we have $c_2+c_3\geq b+a$, then for small enough $\epsilon$, $|u_{3,3}|\geq \frac{b-a}{2}-3\epsilon=C-3\epsilon>\widetilde C$, this leads to a contradiction.
\item $c_1<a\leq c_2<b\leq c_3<c_4$, we specify the following two cases: 

(i) $c_2\geq \frac{b+a}{2}$. Notice that we must have $c_2-c_1<b-a$, so $c_1+c_2>2a$. We can choose $y_{1,1}=c_2,\ y_{1,2}=y_{2,1}=\frac{1}{2}(c_1+c_2)+\epsilon$  and $y_{2,2}=a$, then with a sufficiently small $\epsilon>0$, $u_{1,1}=0,\ u_{1,2}=u_{2,1}=-\frac{1}{2}(c_2-c_1)+\epsilon$ and $Q_{\widetilde{\mathcal{A}}}(y_{2,2}+u_{1,2}+u_{2,1}-u_{1,1})=Q_{\widetilde{\mathcal{A}}}(a-(c_2-c_1)+2\epsilon)=c_1$ and $|y_{2,2}+u_{1,2}+u_{2,1}-u_{1,1}-Q_{\widetilde{\mathcal{A}}}(y_{2,2}+u_{1,2}+u_{2,1}-u_{1,1})|=c_2-a-2\epsilon\geq \frac{b-a}{2}-2\epsilon>\widetilde C$, which leads to a contradiction.
 
(ii) $c_2<\frac{b+a}{2}$, also notice that $c_3-c_2<b-a$, then we can choose $y$ as follows 
\begingroup
\renewcommand*{\arraystretch}{2}
$$y=\begin{pmatrix}
\frac{1}{2}(b+c_2)-\epsilon & c_2 & c_2+\frac{1}{2}(c_3-b)+2\epsilon& \cdots \\
c_2 & c_2& \frac{1}{2}(c_1+c_3) & \cdots \\
c_2+\frac{1}{2}(c_3-b)+2\epsilon &
\frac{1}{2}(c_1+c_3)& a &\cdots\\
\cdots & \cdots & \cdots & \cdots
\end{pmatrix}.
$$
\endgroup
Provided that $\epsilon>0$ is small enough, the corresponding $u$ is
\begingroup
\renewcommand*{\arraystretch}{2}
$$U=\begin{pmatrix}
\frac{1}{2}(b-c_2)-\epsilon & \frac{1}{2}(b-c_2)-\epsilon & -\frac{1}{2}(c_3-c_2)+\epsilon& \cdots \\
\frac{1}{2}(b-c_2)-\epsilon & \frac{1}{2}(b-c_2)-\epsilon & -\frac{1}{2}(c_2-c_1)+\epsilon & \cdots \\
-\frac{1}{2}(c_3-c_2)+\epsilon &
-\frac{1}{2}(c_2-c_1)+\epsilon & -\frac{c_2+b}{2}+a+3\epsilon &\cdots\\
\cdots & \cdots & \cdots & \cdots 
\end{pmatrix}.
$$
\endgroup
Since $c_2\geq a$, then $|u_{3,3}|=\frac{c_2+b}{2}-a-3\epsilon\geq\frac{b-a}{2}-3\epsilon>\widetilde C$, which leads to a contradiction.
\end{itemize}
\end{proof}

\subsection{Optimization}\label{sec:optimization}
Although the proposed decoders \eqref{eq:cbccase1}, \eqref{eq:case3}, \eqref{eq:case2} are standard convex optimization problems,  off-the-shelf optimization solvers do not converge in meaningful time due to the existence of the matrix $D^{\beta+r}$, which is very ill-conditioned. In this section, we suggest a special implementation of the primal-dual algorithm to achieve a much shorter decoding time. For a $512 \times 512$ grayscale image, the decoding process now takes about 1 minute on an Intel Core i7 CPU@2.2GHz 16GB RAM PC. The decoding time can  be further shortened to several seconds by dividing the image into patches and implementing quantization and decoding in parallel.

Since the treatments for all the decoders are similar, we only discuss \eqref{eq:cbccase1}. For simplicity, consider the case when the TV order $\beta$ and the quantization order $r$ are both set to 1,  \eqref{eq:cbccase1} then reduces to 
\begin{align}\label{eq:opt}
\min_{z} \|D^Tz\|_1  \quad \textrm{subject to } \|D^{-1}(z -q)\|_{\infty} \leq \delta/2.
\end{align}
Let us start with writing out the Lagrangian dual of \eqref{eq:opt}
\begin{equation}\label{eq:pd}
\mathcal{L}(z,y)=\|D^T z\|_1-\frac{\delta}{2}\|y\|_1+\langle y,D^{-1}(z-q)\rangle,
\end{equation}
which is a special case of the general form 
\begin{equation}\label{eq:Primaldual}
\max_y\min_x\; \langle Lx,y\rangle+g(x)-f^*(y).
\end{equation}
\eqref{eq:Primaldual} can be solved by primal-dual algorithms such as Chambolle-Pock (Algorithm \ref{alg:CP}).   \\
\begin{algorithm}[H]
\SetAlgoLined
\textbf{Initializations: }{$\tau,\sigma>0,\tau\sigma\|L\|^2<1,\theta\in[0,1],x_0,y_0$, and set $\bar x_0=x_0$}

\textbf{Iterations: }Update $x_n,y_n,\bar x_n$ as follows:
$$ \left\{
\begin{aligned}
y_{n+1} &= \text{Prox}_{\sigma f^*}(y_n+\sigma L\bar x_n)\;\;\;\; \\
x_{n+1} &= \text{Prox}_{\tau g}(x_n-\tau L^*y_{n+1})\;\;\;\; \\
\bar x_{n+1} & = x_{n+1} + \theta(x_{n+1}-x_n)\;\;\;\;
\end{aligned}
\right.
$$
\caption{Solve \eqref{eq:Primaldual} using Chambolle-Pock Method}
\label{alg:CP}
\end{algorithm}
To apply Algorithm \ref{alg:CP} to our problem, we compare corresponding terms in \eqref{eq:Primaldual} and \eqref{eq:pd} and naturally recognize that $x=D^Tz$, $L=D^{-1}D^{-T} $, $g(x)=\|x\|_1$, $f^*(y)=\frac{\delta}{2}\|y\|_1{+\langle y,D^{-1}q\rangle}$. Then 
\begin{equation}\label{eq:pc}
    \mathcal{L}(x,y)=\langle Lx,y\rangle+g(x)-f^*(y)=\|x\|_1-\frac{\delta}{2}\|y\|_1+\langle y,D^{-1}D^{-T}x-D^{-1}q\rangle.
\end{equation}
However, applying Algorithm \ref{alg:CP} to this $\mathcal{L}(x,y)$ results in a very slow convergence, because it contains the matrix $L=D^{-1}D^{-T}$ with a large condition number (of order $O(N^2)$).

{We propose to  move the ill-conditioned matrix $L$ into the sub-problem (b) in Algorithm \ref{alg:1} via a change of variable $x=D^{-1}(z-q)$, or $z=Dx+q$.  Then the primal-dual objective becomes}
\begin{equation}\label{eq:dual}
\mathcal{L}(x,y)=\langle y,x\rangle+\|D^TDx+D^Tq\|_1-\frac{\delta}{2}\|y\|_1.
\end{equation}
Recognize that this is equivalent to setting $g(x)=\|D^TDx+D^Tq\|_1$ and $f^*(y)=\frac{\delta}{2}\|y\|_1$ in \eqref{eq:Primaldual}. With this special $g(x)$ and $f^*(y)$,  Algorithm \ref{alg:CP} becomes the following Algorithm \ref{alg:1}. 

\begin{algorithm}[H]
\SetAlgoLined
\textbf{Initializations: }{$\tau,\sigma>0,\tau\sigma<1,\theta\in[0,1],x_0,y_0.$ and set $\bar x_0=x_0$}

\textbf{Iterations: }Update $x_n,y_n,\bar x_n$ as follows:
$$ \left\{
\begin{aligned}
y_{n+1} &= \argmin_y\; \frac{\sigma\delta}{2}\|y\|_1+\frac{1}{2}\|y-(y_n+\sigma\bar x_n)\|^2 &(a)\\
x_{n+1} &= \argmin_x\tau\|D^TDx+D^Tq\|_1+\frac{1}{2}\|x-(x_n-\tau y_{n+1})\|^2 &(b)\\
\bar x_{n+1} & = x_{n+1} + \theta(x_{n+1}-x_n)\;\;\;\;&(c)
\end{aligned}
\right.
$$
\caption{Solve \eqref{eq:dual} using Chambolle-Pock Method}
\label{alg:1}
\end{algorithm}
In Algorithm \ref{alg:1}, the ill-conditioned matrix $D^TD$ appears in the sub-problem (b), but it does not make (b) an ill-conditioned problem thanks to the existence of the extra quadratic term. We can then solve (b) using standard ADMM and solve (a) by writing out its closed-form solution.






\normalem
\bibliographystyle{plain}
\bibliography{manu}



                                


\end{document}